\documentclass[11pt]{article}
\usepackage{cancel}
\usepackage[pagebackref, colorlinks=true,linkcolor=red,citecolor=blue]{hyperref}
\usepackage{amssymb}
\usepackage{amsmath}
\usepackage{graphicx, tikz-cd}
\input{liemacs10.sty} 
\usetikzlibrary{fadings}
\numberwithin{equation}{section}

\usepackage{color}
 
\newcommand{\red}[1]{{\color{red} #1}}

\renewcommand{\down}{{\mathop{\downarrow}}}

\newcommand\be{{\bf{e}}}

\renewcommand\up{{\uparrow}}
\newcommand{\Mob}{{\rm\textsf{M\"ob}}}
\newcommand\RR{{\mathbb R}}
\newcommand\ZZ{{\mathbb Z}}

\newcommand{\sH}{{\sf H}}
\newcommand{\vNA}{{\rm{vNA}}}
\newcommand{\sK}{{\sf K}}
\newcommand{\sV}{{\tt V}}

\newcommand{\sE}{{\tt E}}

\newcommand{\BGL}{\mathop{{\rm BGL}}\nolimits}

\newcommand{\PdS}{\mathop{{\rm PdS}}\nolimits}

\newcommand{\POh}{\mathrm{P}\mathcal{O}_h}

\renewcommand{\bO}{\mathbb O}

\renewcommand{\phi}{\varphi}

\newcommand{\Stand}{\mathop{{\rm Stand}}\nolimits}

\renewcommand\mlabel{\label}

\begin{document}

\title{Orthogonal pairs of Euler elements II: \\Geometric Bisognano--Wichmann and Spin--Statistics Theorems}
\author{Vincenzo Morinelli, Karl-Hermann Neeb, Gestur \'Olafsson} 

\maketitle

\begin{abstract}
Models in Algebraic Quantum Field Theory (AQFT) may be generalized including Lie groups of symmetries whose Lie algebras admit an Euler element $h$, characterized by the property that $\ad h$ is diagonalizable with eigenvalues in $\{-1, 0, 1\}$. These elements becomes fundamental to the formal description of wedge localization. In this paper, we extend the geometric analysis of Euler wedges and investigate their applications within the AQFT framework.
  We {call  a pair of Euler elements $(h, k)$ orthogonal if
    $e^{i \pi \operatorname{ad} h}(k) = -k.$} 
Using the geometric framework established in our previous work, we derive both a Bisognano–Wichmann Theorem and a Spin--Statistics Theorem for nets of standard subspaces and von Neumann algebras. Our results {show} how this generalized approach recovers classical results in the AQFT literature while providing a deeper structural understanding of the underlying geometry in established models.
\end{abstract}
  
\section{Introduction}

The interplay between geometric properties, algebraic structures, and physically relevant quantities is a distinctive feature of Algebraic Quantum Field Theory (AQFT). A central interest of the algebraic approach to QFT is to isolate the fundamental structural ingredients responsible for the salient properties of a theory. Notable examples include the study of the relation between localization properties and the braiding symmetries of charged Doplicher--Haag--Roberts (DHR) sectors, when different index sets of local algebras are considered
(see  \cite[\S 23]{Ro04} and reference therein). They also include
the notion of geometric modular action introduced in  \cite{BDFS00}, which plays a key role in the analysis of spontaneous symmetry breaking in AQFT
in \cite{BS05}.
These and related works concern models that can be described by nets of von Neumann algebras over general index sets, developed within abstract frameworks that can be adapted to concrete spacetime settings. Depending on the application, the index set may consist of double cones, spacelike cones, wedge regions, or even all   open subsets.

Motivated by the construction of new models from fundamental assumptions of
AQFT, most notably the Bisognano--Wichmann (BW) property and PCT symmetry
(cf.~\cite{SW00}), and a geometric insight, a generalized framework for the study of AQFT models has been developed in recent years. In this approach, the index set is described either abstractly in terms of abstract Euler wedges (cf.~Section~\ref{subsec:4}) or
concretely by ``wedge'' regions in causal homogeneous spaces. This dual perspective allows a systematic construction
and analysis of a broad class of models, formulated either as nets on
abstract Euler wedges  
or as nets on open subsets of homogeneous spaces, that covers many well known examples from the AQFT literature (\cite{MN21, MN22, FNO25a, NO21, Ne26,
NO26}).

A fundamental role in this framework is played by Euler elements
of Lie algebras. They are the generators of modular
  groups in the symmetry Lie algebra (\cite{MN24}),
  and thus provide a natural bridge between
representation-theoretic structures, localization properties in AQFT,
and the modular theory of operator algebras. 
Within this context, the language of standard subspaces emerges as a central tool, providing a unifying perspective on the geometric and modular aspects of the theory.

In this paper, we present a Bisognano--Wichmann Theorem and a Spin--Statistics Theorem in the generalized  geometric setting we introduced in our previous works for nets of standard subspaces (\cite{MN21, MN24, MNO25, FNO25a, Ne26}), and we explain how these results apply to nets of von Neumann algebras.
In this way, we generalize some key results in the existing literature
and embed them into a unified geometric framework for AQFT:
\begin{itemize}
\item Models in AQFT on Minkowski spacetime are described by nets of von Neumann algebras indexed by open regions {(non-empty, connected, open subsets)} of spacetime, satisfying the fundamental principles of quantum theory and relativity. These include, in particular, isotony, locality, Poincar\'e covariance, the spectral condition, and the existence of a vacuum vector enjoying the Reeh--Schlieder property.

Nets of standard subspaces arise in at least two natural ways: first, as one-particle nets associated with irreducible antiunitary representations of the Poincar\'e group, from which free quantum fields are constructed via second quantization; and second, by applying the self-adjoint parts of local von Neumann algebras to a cyclic and separating vacuum vector.

\item The term ``Spin--Statistics'' is motivated by the Algebraic Quantum Field Theory setting where, from a mathematical point of view,
it expresses a relation between the action of a spatial
    $2\pi$-rotation (in the center of the spin group)
and the statistics phase of a superselection sector (\cite{GL95, GL96}). In the context of field algebras on Minkowski spacetime (see \cite{DR89,DR90}),
it follows from a relation between a spatial
$2\pi$-rotation and the commutation relations of local algebras, as in
condition (SS) in Subsection~\ref{subsec:nets-real}
and (SSA) in Subsection~\ref{subsec:res-vNa} (cf.~\cite{GL95}).
We refer to \cite{Do09} for a detailed discussion of this property. \textit{In our context, we interpret it as a relation between the representation of some central
  elements of the symmetry group $G$ and the twisted locality property for a net of standard subspaces} (cf.~\cite{LMR16, MT19}).

\item  To any Rindler wedge $W$ 
  in Minkowski spacetime we can associate a
  canonical one-parameter group of Lorentz boosts that leave $W$ invariant.
  The Bisognano--Wichmann  property now states that, for any Rindler wedge $W$,
  the modular group of the von Neumann algebra associated
  to $W$, with respect to the vacuum state, coincides with this
   one-parameter group. 
  As a consequence, via Tomita--Takesaki Theory,
  the algebraic structure of the model encodes information about
  the symmetry group of the theory which is mostly generated by
  modular groups. Among other results, this property is remarkably helpful in establishing spin–statistics theorems (\cite{GL95,GL96}),
in uncovering new relations between different field theories
(\cite{GLW98, LMPR19, MR20}), and in computing entropy in QFT (\cite{W18});
recent advances are \cite{DMMR25, HLM26}.
\end{itemize}

We now recall the main objects we have developed in recent articles
\cite{MN21, MNO23, MN24, MNO25}
and summarize the main results established in our paper.
 {Throughout the paper, illustrative examples are provided to clarify and motivate the underlying geometric structures.}
In the context of nets of real subspaces, as they arise by modular localization in  AQFT, the unitary representation $U \colon G \to \U(\cH)$ of the
connected symmetry group $G$ is required to satisfy the Bisognano--Wichmann condition:
there exists an element $h \in \g$, the Lie algebra of $G$,
for which $\Delta := e^{2\pi i \cdot \partial U(h)}$ is the modular
operator of a standard subspace associated to a so-called wedge region.
Assuming $\ker U$ to be discrete, by the Euler Element Theorem (\cite{MN24}),
basic AQFT axioms imply that it is natural to assume that $h$ is an {\it Euler element}, i.e., $\ad h$ is non-zero and the Lie algebra
is the direct sum of the eigenspaces $\g_j(h) = \ker(\ad h - j \1)$ for
$j = 1,0,-1$. Here it is allowed that one, but not both, of the
spaces $\fg_{\pm 1}(h)$ is trivial. Simple Lie algebras containing
Euler elements are listed in Appendix \ref{app:list}.
Then $\tau_h^\g := e^{\pi i \ad h}$ is an involutive automorphism of $\g$
and if it integrates to an automorphism $\tau_h$ of $G$
(which is always the case if $G$ is simply connected),
we can form the extended group
\begin{equation}
  \label{eq:gtauh}
  G_{\tau_h} := G\rtimes \{ \1,\tau_h\}.
  \end{equation} 
  We can also assume by \cite{MN24} that the unitary representation $U$
  extends to an antiunitary representation of~$G_{\tau_h}$.
  Now
  \[\sV = \sV(U,h) := \Fix(U(\tau_h)e^{\pi i \partial U(h)}) \]
is a standard subspace 
    and its orbit $U(G).\sV$ in the set $\Stand(\cH)$ of standard subspaces of~$\cH$ represents a net of
    standard subspaces on wedge regions in a causal manifold.
    According to \cite{FNO25a}, such nets exist in abundance for
    all unitary representations of semisimple Lie groups; see
    \cite[\S\S 5.4-5.6]{Ne26} for
      extensions of these results and a  survey of the state of the theory. 

\textit{We are interested in the relation 
between pairs of standard subspaces contained in
such an orbit, and in particular in understanding
how they encode geometric locality conditions for a geometric AQFT.
} 

By shifting the focus from geometric wedge regions,
  the fundamental localization regions for AQFT, 
to Euler elements in the 
Lie algebra, the locality data is encoded abstractly
as follows: The {\it abstract wedge space} is 
  \begin{equation}
    \label{eq:cGGtau}
 \cG(G_{\tau_h}) := \{ (x,\sigma)\in \g \times G\tau_h
 \subeq \g \times G_{\tau_h}    \colon \sigma^2 = e, \Ad(\sigma)x = x\}.
  \end{equation}
  Abstract wedges $(x,\sigma)$ are, on the Lie group level, an analog of the
  Tomita--Takesaki pairs $(i\log\Delta, J)$ satisfying the modular relation
  $J\Delta^{it}J=\Delta^{it}, t \in \R,$
  \footnote{Given a von Neumann algebra $\cM\subset\cB(\cH)$ with a cyclic and separating vector $\Omega\in\cH$, the Tomita operator $S$ is the closure of the anti-linear involution $\cM\Omega\ni a\Omega\mapsto a^*\Omega\in\cM\Omega$.
    Its polar decomposition $S=J\Delta^{1/2}$ defines the antiunitary modular conjugation $J$ and the selfadjoint modular operator $\Delta$.}. 
In view of \cite{MN21, MN24}, the most relevant ones are the {\it
  abstract Euler wedges}, denoted $\cG_E(G_{\tau_h})$, for which
$x$ is an Euler element and $\Ad(\sigma) = \tau_x$.
The set $\cG(G_{\tau_h})$ carries a natural
$G$-action by $g.(x,\sigma)= (\Ad(g)x, g\sigma g^{-1})$,
order structures defined by invariant cones,
and a complementation operation (see Section~\ref{subsec:4}).
For a classification of Euler elements in simple Lie algebras 
  we refer to 
  \cite{MN21}.
The orbit
\begin{equation}
  \label{eq:cw+}
  \cW_+ := G.(h,\tau_h) \subeq \cG_E(G_{\tau_h})
\end{equation}
of the basic pair $W_0 := (h,\tau_h)$ plays the role of an
abstract Lie theoretic version 
of the orbits $U(G).\sV \subeq\Stand(\cH)$, arising from covariant 
nets of standard subspaces. 
 
A direct connection between abstract wedges and standard
subspaces is implemented by
the Brunetti--Guido--Longo (BGL) construction (cf.~\cite{MN21}).
It specifies for
every antiunitary representation $(U,\cH)$ of $G_{\tau_h}$ a natural $G$-equivariant map
\begin{equation}
  \label{eq:bgli}
 \cG(G_{\tau_h}) \to\Stand(\cH), \quad
 \sH_U(x,\sigma) = \Fix(U(\sigma)e^{\pi i \cdot \partial U(x)}), 
\end{equation}
mapping the base pair $(h,\tau_h)$ to $\sV(U,h) = \Fix(U(\tau_h)e^{\pi i \cdot \partial U(h)})$. To understand locality properties of nets on $\cW_+$, we
observe that $h$ is {\it symmetric}, in the sense that $-h \in
\cO_h := \Ad(G)h$, if and only if $\cW_+$ contains a {\it twisted central  complement} 
\[ W^{'\alpha} = (-h,\alpha \tau_h) \quad \mbox{ for some } \quad \alpha \in Z(G).\]
To study locality properties of the kind
\[ \sH(W^{'\alpha}) = Z_\alpha  \sH(W)' \quad \mbox{ for some } \quad
  Z_\alpha \in \U(\cH) \cap U(G)'\]
on the level of  nets on $\cW_+$, it is therefore
necessary that $h$ is symmetric.
The BGL construction \eqref{eq:bgli} assigns
for an antiunitary representation $U$ to the abstract wedge
$W^{'\alpha}$ the standard
subspace
\[ \sV^{'\alpha} := \Fix(U(\alpha)J \Delta^{-1/2})
  = Z_\alpha \sV',\]
where $Z_\alpha \in U(G)'$ is a unitary operator that satisfies 
\[ J Z_\alpha J = Z_\alpha^{-1} \quad \mbox{ and }  \quad
  Z_\alpha^2 = U(\alpha)\quad \mbox{ for  } \quad
J = U(\tau_h)\]
(cf.\ \cite[\S 2.1]{MN21}).
This locality condition appears in the literature on 
standard subspaces \cite{LMR16} and in connection with
von Neumann algebras \cite{GL95, DLR07}.
The possible complements of the abstract wedge
$(h,\tau_h) \in \cG_E(G_{\tau_h})$ are para\-met\-rized by a central
subgroup
\[ Z_2(G) := \{ \alpha \in Z(G) \colon W_0^{\alpha} = (h, \alpha \tau_h) \in \cW_+\}.\] 
If $G$ is simply connected, then
$Z_2(G) \cong \pi_1(\cO_h)$ is the fundamental group
of the adjoint orbit $\cO_h = \Ad(G)h \subeq \g$, 
and the map $\cW_+ \to \cO_h, (x,\sigma) \mapsto x$,
is a simply connected covering.
One of the central results in Part I  \cite{MNO25} 
of this project  asserts that,
if $\g$ is simple and $h \in \g$ is an
Euler element, then $\pi_1(\cO_h)$ is either trivial, $\Z$ or $\Z_2$.
More specifically, infinite fundamental
groups occur precisely 
when $\g$ is hermitian (e.g.~$\so_{2,d}(\R)$, $d \in \{1,3,4,\ldots\}$,
cf.~Appendix~\ref{app:list}),
$\Z_2$ occurs for the so-called split-type cases
(e.g.~$\fsl_{2r}(\R), r>1$), and for $\g = \so_{1,d}(\R), d > 2,$
it is trivial
(cf.\ Table $2$ and Theorem~4.1 in \cite{MNO25}).
For general Lie algebras, \cite[Thm.~4.2]{MNO25} provides a
reduction to the case of simple ones.

In this way, we obtain an abstract framework for nets of standard subspaces and nets of von Neumann algebras with fundamental structural properties that naturally extend those of AQFT to nets with symmetry group $G$ 
(see Definitions~\ref{def:ssnet} and \ref{def:vnnet}).

Within this framework, i.e., a unitary $G$-representation $(U,\cH)$
and a net $\sH$ of standard subspaces on $\cW_+$, we show that,
if $U$ satisfies the spectrum condition
with respect to a generating  cone, and $\sH$ satisfies
isotony and central twisted locality, then it also
has the Bisognano--Wichmann property  (Theorem~\ref{thm:1BW}).
If, in addition, the net is regular, then it has
the Spin--Statistics property  (Theorem~\ref{thm:spst1}):  
\[ U(\alpha)=Z_\alpha^2, \quad \mbox{ where } \quad
  \alpha \in Z_3(G) := Z_2(G) \cup \{ \alpha \in Z(G) \colon
  W_0^{'\alpha} = (-h, \alpha \tau_h) \in \cW_+\}.\] 
Here a key point is that our symmetry data contains only a unitary
  representation $U$ of $G$ and not an antiunitary extension to $G_{\tau_h}$.
  Accordingly, our strategy is to specify properties of the net $\sH$ that ensure
  the existence of an extension $U \colon G_{\tau_h} \to \AU(\cH)$, for which
  $\sH$ is the BGL net~$\sH_U$ from \eqref{eq:bgli}.

We further clarify a relevant geometric aspect of this property by showing that,   given two centrally complementary wedges $W, W^{'\alpha}\in \cW_+$, then  $\alpha$ can be written as a product of central elements of the integral subgroups of 
$\fsl_2(\R)$-subalgebras, namely
{\[ \alpha=\zeta_1,\ldots\zeta_n,\qquad \zeta_j\in Z(S_j) \subeq Z(G),\, j=1,\ldots, n,
 \] 
where $S_j\subset G$ is an integral subgroup with Lie algebra
$\Lie(S_j)\cong \fsl_2(\RR)$. Then
\[  \sH(W^{'\alpha}_0) = Z_{\zeta_1}\cdots Z_{\zeta_n} \sH(W_0)'
  \quad \mbox{ with }  \quad Z_{\zeta_j}^2=U(\zeta_j),\quad  W_0 = (h, \tau_h).\] 
\textit{Thus, the Spin--Statistics property  and the locality property of the net
  rely on the restriction of $U$ to $Z(G)$, resp.,
  the centers of the integral subgroups of $\fsl_2(\RR)$-subalgebras of $\fg$ (Theorem~\ref{thm:slspsp1}). } This provides a general framework into
which the results of \cite{GL95,MT19} naturally fit. We apply these results to nets of von Neumann algebras with a cyclic and separating vacuum vector,
to obtain the Bisognano--Wichmann Theorem \ref{thm:vnaBW} and the 
  Spin--Statistics Theorem~\ref{thm:spst2}. These theorems naturally
extend the results in \cite{Lo08, BGL93, GL95, DLR07}. 

In this picture, (symmetrically) orthogonal pairs of Euler elements
  emerge as a relevant geometric structure. They have been used
  in \cite{MN21} and classified in \cite{MNO25}.
On the level of standard subspaces, they correspond to pairs 
$(\sV_1, \sV_2)$  satisfying the ``orthogonality relations''
\label{eq:j1j2}
\begin{equation}
 J_{\sV_2}  \Delta_{\sV_1} J_{\sV_2}  = \Delta_{\sV_1} \quad \mbox{ and }  \quad 
 J_{\sV_1} \Delta_{\sV_2} J_{\sV_1}  = \Delta_{\sV_2}.
\end{equation}
  Such pairs arise naturally from antiunitary representations
  of the group $\tilde\SL_2(\R)_{\tau_h}$ and they are first visible in \cite{GL95} for Poincar\'e covariant theories and in \cite[\S~5]{Lo08} for
  M\"obius covariant theories. On the level of Euler elements, 
  this leads to the concept of an {\it orthogonal pair}
  $(h, k)$ of Euler elements, which means that
  $\tau_h^\g(k) = -k$. In
  \cite{MN21} it is shown that
  an Euler element $h$ in a simple Lie algebra $\g$ is
   symmetric if and only if
  there exists a second Euler element $k$ such that $(h,k)$ is orthogonal.
  If this is the case, then $(k,h)$ is also orthogonal and $h,k$ generate a
  real Lie subalgebra $\fs_{h,k} \cong \fsl_2(\R)$. Accordingly,
  a typical example of an
  orthogonal pair of  Euler elements in $\fsl_2(\R)$ is 
\begin{equation}
  \label{eq:h0k0-intro}
 h_0 := \frac{1}{2}\pmat{1 & 0 \\ 0 &-1} \quad \mbox{ and } \quad 
 k_0 := \frac{1}{2}\pmat{0 & 1 \\ 1 &0}.
\end{equation}

Concerning pairs of Euler elements, we find the following geometric  picture very instructive.  
  The first results on pairs of modular groups
  are due to Borchers and Wiesbrock (cf.\ \cite[Thm.~3.21]{Lo08}, \cite{Wi93},
  \cite{Wi97}, \cite{GLW98}, \cite{Bo00}). They concern halfsided  modular inclusions. These are pairs $\sV_1 \subeq \sV_2$, for which there exists 
a positive energy representation of the affine group $\Aff(\R)$ such that
$\sV_1 = U(1,1)\sV_2$, where $(1,1).x = x + 1$ represents a unit translation.
More generally, pairs $(\sV_1, \sV_2)$ {are said to have
  {\it $+$modular intersection} if $\sV_1\cap\sV_2$ is cyclic and
  \[ \Delta_{\sV_j}^{-it}(\sV_1\cap\sV_2)\subset \sV_j 
    \quad \mbox{ for } \quad t>0, j =1,2,\] 
and the strong limit (a scattering operator) 
\[ S := \lim_{t \to \infty} \Delta_{\sV_1}^{it} \Delta_{\sV_2}^{-it}
\quad \mbox{ satisfies } \quad J_{\sV_1} S J_{\sV_1} = S^{-1}.\] }
This is equivalent to the existence of a general antiunitary
representation of $\Aff(\R)$ with $\sV_1= U(1,1)\sV_2$}
(cf.\ \cite[Thm.~3.22]{NO17}). 
With these concepts, Wiesbrock obtained in
\cite[Thm.~6]{Wi97} a characterizations of triples
$(\sV_1, \sV_2, \sV_3)$ of standard subspaces such that their
modular groups and conjugations generate an antiunitary representation
of the extended M\"obius group $\PGL_2(\R) \cong \PSL_2(\R)_{\tau_h}
{=\Mob_{\tau_h}}$.
In \cite{Lo08} it is shown that such a triple  generates a $\Mob$-covariant net
of standard subspaces on the circle. An analogous result for von Neumann algebras appears in \cite{GLW98}.

The space 
$\cE(\fsl_2(\R))$ of Euler elements in $\fsl_2(\R)$
can be identified with a $2$-dimensional one-sheeted hyperboloid
(de Sitter space~$\dS^2$) and the space of Euler involutions is
$\PdS^2 =\dS^2/\{\pm\1\}$. 
Here pairs with modular intersections correspond to pairs
of Euler elements lying on a 
  lightray in $\dS^2$, i.e., an affine line in the
  hyperboloid (\cite[Thm.~3.22]{NO17}).
  In contrast, orthogonal pairs lie on a closed spacelike geodesic in such a way
  that their involutions are antipodal in $\PdS^2$, resp., related by a
  $90^\circ$-rotation in $\SO_{1,2}(\R)$. They may also be characterized
  by the involutions $\tau_h$ and $\tau_k$ to commute.

A classification of adjoint orbits of orthogonal pairs $(h,k)$ of Euler elements
in simple real Lie algebras is contained in \cite{MNO25}
(and recalled in Appendix~\ref{app:a.1}). It is based on our classification of Euler elements in \cite{MN21} and uses some structural results on $3$-graded
simple Lie algebras. 
If the simple Lie algebra $\g$ is hermitian, i.e., contains a
pointed generating invariant closed convex cone $C$,
then $C$ specifies a non-trivial order
structure on the abstract wedge space $\cW_+$ (see Sect.~\ref{subsec:4}).
In this context,  orthogonal pairs $(h,k)$ occur in $3$ types:
\begin{itemize}
\item positive/negative timelike: $[h,k] \in \pm C$, and 
\item spacelike $[h,k] \not\in C \cup - C$. 
\end{itemize}

The relation between orthogonal pairs of Euler elements and the 
  Spin--Statistics Theorem is that adjoint orbits of orthogonal
  pairs $(h,k)$ correspond to adjoint orbits of $\fsl_2(\RR)$ subalgebras
  $\fs_{h,k} = \R h + \R k + \R [h,k]$, and the corresponding central elements
  $\zeta_{h,k} = \exp(2\pi [h,k]) \in Z(G)$ are
  the structural origin of the twisted locality relations.

  As an illustration, we discuss in Section~\ref{sect:12ssconf} the massless conformal free field in $1$+$2$ dimensions. Owing to the absence of the Huygens principle, the associated one-particle net satisfies spacelike locality,
  but only timelike twisted locality. This phenomenon is related to the fact that the unitary positive energy representation of the conformal group restricts to a unitary representation of the Lorentz group and to the double covering $\Mob^{(2)}$ of the group $\Mob$, fixing the compactification of the time axis.
We stress that the Lorentz group is isomorphic to $\PSL_2(\RR)$, and that the origin of the different locality properties for spacelike and timelike regions in this geometric framework lies in the fact that the Lorentz subalgebra and the
Lie subalgebra $\L(\Mob)$ of the conformal Lie algebra $\so_{2,3}(\R)$
are generated by inequivalent orthogonal pairs of Euler elements.  
This is illustrated by the fact that a unitary representation of the conformal group satisfying the spectral condition restricts to a unitary representation of the Lorentz group that does not satisfy any spectral condition, whereas, 
{the restriction to $\Mob^{(2)}$ does.}

The structure of  this paper is as follows:
\begin{itemize}
\item In Section~\ref{sec:2} we introduce the geometric and analytic framework, recalling results from previous works. Readers from the AQFT community  will be guided through examples from the AQFT literature to see how they fit into our setting. 
\item In Section~\ref{sec:spinstat}
  we present the Bisognano--Wichmann and Spin--Statistics theorems and relate them to the regularity condition. 
 \item In Section~\ref{sect:appssn} we  
  describe how our results apply to
  one particle nets on $1$+$2$-dimensional Minkowski
  space and to $2$-dimensional conformal nets. 
\item In Section~\ref{sect:appvnan} we prove the corresponding theorems
  for nets of von Neumann algebras. We also
  compare our results with AQFT results in $1$+$3$-dimensional Minkowski space, on the chiral circle, and on the Einstein universe.
\item The final section contains an outlook on future perspectives.
\end{itemize}

\nin\textbf{Acknowledgements:}

  The work of Vincenzo Morinelli is partly supported by INdAM-GNAMPA, University of Rome Tor Vergata funding OANGQS CUP E83C25000580005, and the MIUR Excellence Department Project MatMod@TOV awarded to the Department of Mathematics, University of Rome Tor Vergata, CUP E83C23000330006.

  Karl-H. Neeb is supported by the DFG grant  Ne 413/10-2, 
  ``Nets of standard subspaces on ordered symmetric spaces''.

\pagebreak 

\nin {\bf Notation and conventions:} 
\begin{itemize}
\item Throughout this paper $G$ denotes a finite-dimensional connected 
  Lie group with Lie algebra $\g = \L(G)$, $e \in G$ its identity element, and
  $G_e$ its identity component. 
\item For $x \in \g$, we write $G^x := \{ g \in G \colon \Ad(g)x = x \}$ 
for the stabilizer of $x$ in the adjoint representation 
and $G^x_e = (G^x)_e$ for its identity component. 
\item For $h \in \g$ and $\lambda \in \R$, we write 
$\g_\lambda(h) := \ker(\ad h - \lambda \1)$ for the corresponding eigenspace 
in the adjoint representation.
\item For a Lie subalgebra $\fs \subeq \g$, we write 
$\Inn_\g(\fs)= \la e^{\ad \fs} \ra \subeq \Aut(\g)$ for the subgroup 
  generated by $e^{\ad \fs}$. We call $\fs$ {\it compactly embedded} if
  the group $\Inn_\g(\fs)$ has compact closure.
\item For $x \in \g$, we write $\cO_x := \Inn(\g)x$ for its adjoint orbit.   
\item If $\g$ is a Lie algebra, we write $\cE(\g)$ for the set of 
{\it Euler elements} $h \in \g$, i.e., $\ad h$ is non-zero and diagonalizable 
with $\Spec(\ad h) \subeq \{-1,0,1\}$. We call $h$ {\it symmetric} if 
$-h \in\cO_h$. The involution of $\g$ specified by~$h$
is denoted $\tau_h^\g := e^{\pi i \ad h}$. 
\item If $G$ is a group with Lie algebra $\g$, we
write $\tau_h \in \Aut(G)$ for the corresponding involution on $G$,
provided it exists, and $G_{\tau_h} := G \rtimes \{\id_G, \tau_h\}.$
\item Abstract Euler wedges:
  $\cG_E(G_{\tau_h}) = \{ (x,\sigma) \in
  \cE(\g) \times G\cdot \tau_h \colon
  \sigma^2 =e, \Ad(\sigma) = \tau_x^\g\}.$
\item An {\it antiunitary representation} of $G_{\tau_h}$ is a
  homomorphism $U \colon G_{\tau_h} \to \AU(\cH)$ for which
  $U(\tau_h)$ is an antiunitary involution, i.e., a conjugation.
  It is always assumed to be strongly continuous, i.e., to have continuous orbit maps.
\item Given a Lie group $G$, a non trivial pointed cone $C\subset \fg$ and a unitary (or antiunitary) representation $U$ of $G$,
  we shall say that $U$ satisfies the \textit{spectral condition}
  or is a \textit{positive energy representation} if
  $ C \subeq C_U := \{ x \in \g \colon -i\cdot \partial U(x) \geq 0\}.$ 
  \item For a unitary representation $(U,\cH)$ of $G$ we write:
$\partial U(x) = \derat0 U(\exp tx)$ for the infinitesimal
    generator of the unitary one-parameter group $(U(\exp tx))_{t \in\R}$
    in the sense of Stone's Theorem. 
\end{itemize}

\tableofcontents

\section{Preliminaries}
\mlabel{sec:2}

Let $G$ be a connected Lie group
with Lie algebra $\g$ and $h \in \g$ be an Euler element.
We assume that the involution $\tau_h^\g = e^{\pi i \ad h}$ integrates
to  an involution $\tau_h$ on $G$ and form the non-connected Lie group
\[ G_{\tau_h} := G \rtimes \{  \1, {\tau_h}\}.\]

\subsection{Abstract Euler wedges} 
\mlabel{subsec:4}

In this section we introduce some basic notation related to the space
  $\cG_E$, in particular the $\alpha$-twisted complements and the
  order structures. Following
  \cite[Def.~2.5]{MN21}, we recall the {\it
  abstract (Euler) wedge space 
\[ \cG_E := \cG_E(G_{\tau_h}) := \{ (x,\sigma)\in \cE(\g) \times G\tau_h
  \subeq \g \times G_{\tau_h}    \colon \sigma^2 = e, \Ad(\sigma) = e^{\pi i \ad x}\}\]
of $G_{\tau_h}$} as a structure that
encodes much information about the $G$-action on standard subspaces
of Hilbert spaces, modular groups and modular conjugations.
We refer to the introduction for more details. 
The group $G_{\tau_h}$ acts naturally on $\cG_E$  by 
\begin{equation}
  \label{eq:cG-act}
 g.(x,\sigma) := (\Ad^\eps(g)x, g\sigma g^{-1}),
\end{equation}
where $\Ad^\eps(g) = \eps_G(g) \Ad(g)$ is the {\it twisted adjoint action}.
Here $\eps_G \colon G_{\tau_h} \to \{\pm 1\}$ is
the surjective group homomorphism with $\ker \eps_G = G$.

We define the {\it complementary wedge} of 
$W = (x,\sigma) \in \cG_E$ by
$W' := (-x,\sigma) = \sigma.W.$
Note that $(W')' = W$ and $(gW)' = gW'$ for $g \in G$ 
by \eqref{eq:cG-act}.  
The relation $\sigma.W = W'$ is our main 
motivation to introduce  the twisted adjoint action.
This fits the relation $J_\sV \sV = \sV'$ for standard
subspaces (cf.\ Subsection~\ref{subsec:stand-subs}). 
It also matches the geometric interpretation in terms
of wedge regions in spacetime manifolds 
and the modular theory of operator algebras.

We consider the pair 
\[ W_0 := (h, {\tau_h}) \in \cG_E \]
as a base point and call its $G$-orbit 
\begin{equation}
  \label{eq:cw+b}
  \cW_+ := G.W_0 = \{ (\Ad(g)h, g {\tau_h}(g)^{-1} \cdot {\tau_h}) \colon g \in G\}
\end{equation}
the corresponding {\it abstract wedge space}.
We write
\[  Z(G)^- := \{  z \in Z(G) \colon  \tau_h(z) = z^{-1} \},\]
and consider the homomorphism 
\begin{equation}
  \label{eq:partialW}
  \partial_h \colon  G^{\{\pm h\}} := \{ g \in G \colon \Ad(g)h \in \{\pm h \}\} \to Z(G)^-,  \quad g \mapsto g \tau_h(g)^{-1}
\end{equation}
 (cf.\ \cite[\S 2.4]{MN21}). 
 For $\alpha \in Z(G)^-$, we define the 
{\it $\alpha$-twisted  complement}  of 
$W = (x,\sigma) \in \cG_E(G)$ by 
\begin{equation}
  \label{eq:xsigma'}
 W^{'\alpha} = (x,\sigma)^{'\alpha} := {(-x,\alpha\sigma)}.
\end{equation}
We will refer to  couples of the form $W^{'\alpha}$ as 
{\it central complements of the abstract wedge $W$.} 
We consider $W^{'\alpha}$ as a ``complement'' of {$W = (x,\sigma)$ 
  because the corresponding Euler element is $-x$}.

 The kernel of the restriction of $\partial_h$
 to $G^h$ is the open subgroup
 \begin{equation}
   \label{eq:gw0}
   G^{W_0} = G^{h,\tau_h} = \{ g \in G \colon  g.W_0 = W_0 \}
   \subeq G^h. 
 \end{equation}
 We consider the central subgroups
\begin{equation}
  \label{eq:z123}
 Z_1(G) := \partial_h(Z(G)) \subeq Z_2(G)
  := \partial_h(G^h) \subeq Z_3(G) := \partial_h(G^{\{\pm h\}})
  \subeq Z(G)^-.
\end{equation}
The subgroup $Z_2(G)$ describes the orbit
$G^h.W_0 = \{h\}  \times Z_2(G) \cdot{\tau_h} \subeq \cW_+.$
It is the fiber over the Euler element $h$ for the projection map
\begin{equation}
  \label{eq:qtoOh}
  q \colon \cW_+ \to \cO_h =  \Ad(G)h, \quad (x,\sigma) \mapsto x.
\end{equation}
As $\cW_+ \cong G/G^{W_0} = G/G^{h,\tau_h}$, 
and $G^{h,\tau_h}$ is an open subgroup of $G^{\tau_h}$, the projection
$q$ is a $G$-equivariant covering map.
If $h$ is not symmetric, then $G^{\{\pm h\}} = G^h$ and $Z_2(G) = Z_3(G)$,
and if $h$ is symmetric, then $G^h \subeq G^{\{\pm h\}}$ is a subgroup
of index $2$. So we always have $|Z_3(G)/Z_2(G)| \leq 2$.
Equality is characterized by
\begin{equation}
  \label{eq:lem:5.2}
  W_0' \in \cW_+ \quad \Leftrightarrow \quad
  (\exists g \in G^{\tau_h}) \ \Ad(g)h = -h
  \quad \Leftrightarrow \quad
  Z_3(G) = Z_2(G).  
\end{equation}
{\rm(\cite[Lemma~3.3]{MNO25})} 

To a pointed closed convex $\Ad^\eps(G_{\tau_h})$-invariant cone 
    $C \subeq \g$,
    we associate the {\it $C$-order on $\cG_E(G_{\tau_h})$} as follows
(\cite[Def.~2.5]{MN21}).  
  The subset  
\[ S_{W_0} := \exp(C_+) G^{W_0} \exp(C_-)  \subeq G
  \quad \mbox{ with} \quad
  C_\pm := \pm C \cap \g_{\pm 1}(h)\]
is a closed subsemigroup of $G$, and its unit group
is the stabilizer $G^{W_0}$ (cf.\ \cite[Thm.~3.4]{Ne22}).
It determines a $G$-invariant partial order on $\cW_+$ by 
\begin{equation}
  \label{eq:cG-ord}
g_1.W_0 \leq g_2.W_0  \quad :\Longleftrightarrow \quad 
g_2^{-1}g_1 \in S_{W_0}.
\end{equation}
In particular, $g.W_0 \leq W_0$ is equivalent to $g \in S_{W_0}$.
Here $C = \{0\}$ is also permitted and specifies the trivial order.

\begin{rem} \mlabel{rem:z2pi1oh} 
  Let $q_G \colon \tilde G \to G$ denote the simply connected covering group 
  (on which $\tau_h$ always exists) 
and put $\tilde G_{\tau_h} = \tilde G \rtimes \{ \1,\tau_h\}$.

\nin (a) The subgroup $\tilde G^{\tau_h}$ is connected
(\cite[Thm.~IV.3.4]{Lo69}) with Lie algebra $\g^{\tau_h} = \g_0(h) = \ker(\ad h)$,  
hence equal to $\tilde G^h_e$, so that
$\tilde W_0 := (h, \tau_h) \in \cG_E(\tilde G_{\tau_h})$ has the connected
stabilizer group $\tilde G^{h,\tau_h} = \tilde G^h_e$. Therefore $\tilde \cW_+ := \tilde G.\tilde W_0$
is simply connected, $\tilde W_0' \not\in  \tilde\cW_+$, and the map
\begin{equation}
  \label{eq:tildew-cov}
  \tilde q \colon \tilde \cW_+ \to \cO_h, \quad (x,\sigma) \mapsto x
\end{equation}
is the simply connected covering space of the adjoint orbit $\cO_h$. We thus 
obtain the following topological interpretation of $Z_2(\tilde G)$: 
\begin{equation}
  \label{eq:pi1oh}
  Z_2(\tilde G) = \partial_h(\tilde G^h)
  \cong \tilde G^h/\tilde G^h_e 
  \cong \pi_0(\tilde G^h)
 \cong \pi_1(\cO_h).
\end{equation}

\nin (b) If $h$ is symmetric, then $-\cO_h = \cO_h$, so that we also obtain a
projective adjoint orbit
\[ \POh := \cO_h/\{ \pm \1\} \cong \tilde G/\tilde G^{\{\pm h\}}, \]
for which $\tilde \cW_+$ also is the universal covering.
Accordingly,
\begin{equation}
  \label{eq:pi1z3}
  Z_3(\tilde G) = \partial_h(\tilde G^{\{\pm h\}}) \cong
  \pi_0(\tilde G^{\{\pm h\}}) \cong \pi_1(\POh). 
\end{equation}
\end{rem}

\subsection{Orthogonal pairs of Euler elements} 
\begin{defn} \mlabel{def:diag-act}
We call a pair $(h,k)$ of Euler elements in a Lie algebra
  $\g$ {\it orthogonal} if $\tau_h^\g(k) = -k$.
We write $(h,k) \sim (h',k')$ if there exists a $\phi \in \Inn(\g)$ with
$(\phi(h),\phi(k)) = (h',k')$. 
\end{defn}

\begin{ex}
  \mlabel{ex:sl2} {\rm(Euler elements in $\fsl_2(\R)$)} 
In $\g= \fsl_2(\R)$ we have the orthogonal Euler elements 
\begin{equation}
  \label{eq:h0k0}
 h_0 = \frac{1}{2}\pmat{1 & 0 \\ 0 &-1} \quad \mbox{ and } \quad 
 k_0 := \frac{1}{2}\pmat{0 & 1 \\ 1 &0}. 
\end{equation}
The group $\Inn(\g)$ acts transitively on the set 
  \[\cE(\fsl_2(\R))
    = \big\{ x \in \fsl_2(\R) \colon \det(x) = -{\textstyle\frac{1}{4}}\big\} = 
    \Big\{ \pmat{a & b \\ c & -a} \colon a^2 + bc = \frac{1}{4}\Big\}
  \]
  of Euler elements, which is diffeomorphic to $2$-dimensional de Sitter space
$\dS^2$ (a one-sheeted hyperboloid) (cf.~\cite[Rem.~2.7]{MN22}).
From
  \[ \tau_{h_0}\pmat{a & b \\ c & -a}
    = \pmat{a & -b \\ -c & -a} \]
  it follows that the Euler elements orthogonal to $h_0$ are those of the
  form
  \[ k = \frac{1}{2} \pmat{0 & b \\ b^{-1} & 0},\]
  and we obtain in particular $k_0$ for $b = 1$.
We  also have   $\tau_{k_0}(h_0) = -h_0$. 

Now it is easy to see that the pairs $(h_0, k_0)$ and $(h_0, -k_0)\sim (k_0,h_0)$ represent the two conjugacy classes of orthogonal pairs.
For 
\begin{equation}
  \label{eq:z0}
  z_0 :=\frac{1}{2} \pmat{0 & 1 \\ -1 & 0} =  [h_0, k_0] \in \so_2(\R).
\end{equation}
we then have
    \begin{equation}
      \label{eq:z0-rot}
e^{\mp\frac{\pi}{2} \ad z_0} h_0 = \mp [z_0, h_0] = \pm k_0 \quad \mbox{ and  } \quad        
e^{-\pi \ad z_0} h_0 =  e^{\pi \ad z_0} h_0 =  -h_0.
\end{equation}
In particular, $e^{\pi \ad z_0} \in \Aut(\fsl_2(\R))$ maps $h_0$ to $-h_0$, but
this element is fixed by $\tau_h$ because 
$\tau_h(e^{\pi \ad z_0}) = e^{-\pi \ad z_0} = e^{\pi \ad z_0}$.
If $S$ is a connected Lie group with Lie algebra
  $\fs = \fsl_2(\R)$, we have
  \begin{equation}
    \label{eq:Zi(S)-sl2}
Z(S) = \exp(2\pi \Z z_0) = Z_3(S) \supeq
Z_1(S) = Z_2(S) = \exp(4\pi \Z z_0).
  \end{equation}

    We also note that $\theta(x) = - x^\top$ defines a Cartan involution
    on $\fsl_2(\R)$ for which $\Fix(\theta) = \so_2(\R) = \R z_0$. The
   $\ad h$-eigenvector $e_0 :=  \pmat{0 & 1 \\ 0 & 0}$ 
satisfies
    \[ h_0 = -\frac{1}{2}[e_0,\theta(e_0)], \quad
      k_0 = \frac{1}{2}(e_0 - \theta(e_0)), \quad \mbox{ and } \quad
      z_0 = \frac{1}{2}(e_0 + \theta(e_0)).\]
    The unique pointed generating invariant cone $C \subeq \fsl_2(\R)$,
    containing $z_0$, is given by
    \[ C = \Big\{ \pmat{a & b \\ c & -a} \colon
      a^2 + bc \leq 0, b \geq c \Big\}.\] 
    \end{ex}

\begin{rem}\label{rmk:gengen} It follows from Example~\ref{ex:sl2}
that any Euler element $h$ contained in some
$\fsl_2(\R)$-subalgebra $\fs \subeq \g$ 
is symmetric and admits a partner $k$ for which
$(h,k)$ and $(k,h)$ are orthogonal. 
We call pairs with this property {\it symmetrically orthogonal}.
  If $\g$ is simple, all  orthogonal pairs of Euler elements are
  symmetrically orthogonal by \cite[Thm.~3.13]{MN21}.
Conversely, if $(h,k)$ is symmetrically orthogonal, then $h$, $k$ and
$z_{h,k} := [h,k]$  span a
subalgebra $\fs_{h,k}$ isomorphic to $\fsl_2(\R)$
(\cite[Thm.~2.6]{MNO25}) and
\begin{equation}
  \label{eq:zetahkintro}
  \zeta_{h,k} := \exp(2\pi z_{h,k}) \in Z(G)
\end{equation}
generates the cyclic center of the integral
subgroup  $S_{h,k} := \la \exp \fs_{h,k} \ra \subeq G$. 
In view of \eqref{eq:z0-rot},
the element $r_{h.k} = \exp(\frac{\pi}{2} z_{h,k}) \in S$ satisfies
\[ \Ad(r_{h,k})h = -k \quad \mbox{ and }  \quad
 r_{h.k}^2.(h, \tau_h) = (-h, \zeta_{h,k} \tau_h),\]
  so that
  \[ \tau_{k} \tau_{h}
    =  r_{h,k}\tau_{h} r_{h,k}^{-1} \tau_{h}
    =  r_{h,k}^2,\]
and we obtain the representation of the central element
  $\zeta_{h,k}$ as a commutator: 
  \begin{equation}
    \label{eq:gengen}
\tau_{k} \tau_{h}\tau_{k} \tau_{h}
= r_{h,k}^4 =  \partial_h(r_{h,k}^2) = \zeta_{h.k} \in Z_3(G)
  \end{equation}
\end{rem}

The following theorem shows that all twisted complements
  of abstract Euler wedges can be obtained in several
  steps by acting with central elements of $\SL_2(\R)$-subgroups.

\begin{thm} \mlabel{cor:z3gen} {\rm(\cite[Cor.~5.11]{MNO25})}
  {\rm($Z_3$-Theorem)} Let
  $G$ be a connected Lie group with Lie algebra $\g$ 
and $h \in \g$ a symmetric Euler element for
  which the Lie algebra involution  $\tau_h^\g$ integrates to~$G$.
  Then $Z_3(G)$ is generated by finitely many elements
  of the form $\zeta_{h,k} := \exp(2\pi [h,k])$, where $(h,k)$ is a pair of orthogonal   Euler elements, and $Z_2(G)$ is generated by elements of the form
  $\zeta_{h,k,k'} := \zeta_{h,k} \zeta_{h,k'}^{-1}$, where
    $(h,k')$ is another orthogonal pair. 

If, in addition, $\g$ is simple, then $Z_2(G)$ is cyclic and generated
    by an element $\zeta_{h,k,k'}$, where the orthogonal pairs $(h,k)$
    and $(h,k')$ can be choosen as follows: 
If $\fs \cong \fsl_2(\R)^{\oplus r}$ as in
    {\rm Theorem~\ref{thm:1.4}}, where
    $h= (h_0, \cdots, h_0)$,
    $k = (k_0, \cdots, k_0)$ and 
    $k' = (-k_0, k_0,\cdots, k_0)$, then 
    \[ \zeta_{h,k,k'} = \exp(4\pi (z_0, 0,\cdots, 0)).\] 
\end{thm}

\subsection{Examples} 
The examples below show that
there is a natural correspondence between abstract wedges,
their order structure and their complements,
and wedge regions in manifolds, their inclusions and causal complements,
  as they occur in AQFT models.
The correspondence between abstract wedges and wedge regions
in Minkowski space is briefly discussed in  Example~\ref{ex:mink} below.
For de Sitter space we refer to \cite[Sect.~2.1]{Mo25}, 
\cite[Ex.~2.7]{MN24} and \cite{FNO25b}, and for the  chiral 
circle to \cite[Ex.~2.10(c)]{MN21}.
We also refer to the lecture notes \cite{Ne26} for various concrete
  examples.
\begin{example}\label{ex:mink}
Consider  Minkowski spacetime $M=\RR^{1,d}$ of space dimension $d>0$,
 with the Lorentzian metric given by the bilinear form  
 \[ (x_0,\bx) (y_0, \by) := x_0 y_0 - \bx \by = x_0 y_0 - x_1 y_1
   - \cdots - x_d y_d.\]
 It isometry group is the {\it Poincar\'e group}
 \[ \cP := \R^{1,d} \rtimes \cL, \quad \mbox{ where } \quad
   \cL := \OO_{1,d}(\R) \]
 is the {\it Lorentz group}. The causal structure on $M$ is specified
 by the {\it closed forward light cone} 
 \[ C=\{x\in M:x^2\geq0, x_0\geq0\}\subset\RR^{1,d}.\]
 It is left invariant by the  {\it orthochronous} subgroup
 $\cP^\up = \R^{1,d} \rtimes \cL^\up$. Writing
 $\cP_+ = \R^{1,d} \rtimes \SO_{1,d}(\R)$
 for the subgroup of orientation preserving (=proper)
 isometries, 
 \[ \cP^\up_+ = \cP^\up \cap \cP_+  = \R^{1,d} \rtimes \cL_+^\up \]
 is the identity component.

The {\it causal complement} of a region $\cO\subset M$ is
\begin{equation}
  \label{eq:O'}  \cO'=\{y\in\RR^{1,d}: ( \forall x\in\cO)\, (x-y)^2<0\}^\circ.
\end{equation}
We always have $\cO \subeq \cO''$, and we
say that $\cO$ is {\it causally complete} if $\cO=\cO''$.
A \textit{wedge region} in $\RR^{1,d}$
is a Poincar\'e transform  $W=gW_R,  g\in \cP_+^\up$,
where
\[   W_R=\{x\in\RR^{1,d}:|x_0|<x_1\}\]
is the  Rindler wedge  $W_R$.
Wedge regions are always causally complete with $(gW_R)' = g.(-W_R)$. 
 
We consider the one-parameter group of boosts $\Lambda_{W_R}(t)=\exp(th_{W_R}) \in \cP_+^\uparrow$, where
\[ h_{W_R} x= (x_1, x_0, 0,\ldots, 0).\]
There is a one-to-one correspondence between wedge regions
and one-parameter groups of boosts, resp., their generators, via
\[ h_W:=\Ad(g) h_{W_R}\stackrel{1:1}\longleftrightarrow
  \Lambda_W:=g\Lambda_{W_R}g^{-1}\stackrel{1:1}
  \longleftrightarrow W=gW_R \quad \mbox{ for } \quad g \in\cP_+^\up, \] 
where $h_W\in\Lie(\cP_+^\up)$ generates $\Lambda_W(t)=\exp(th_W)$.
To each wedge region $W$, we also associate the involution
\[ j_W := e^{\pi i h_W} \res_{\R^{1,d}} \in \cP_+^\down,
  \quad \mbox{ and in particular } \quad
  j_{W_R}(x) =(-x_0,-x_1,x_2,\ldots,x_d). \]
This  establishes a one-to-one
correspondence between abstract wedges
$W = (h_W, j_W) \in\cW_+(\cP_+)$ and wedge regions $W \subeq \RR^{1,d}$.
Further, $j_W(W) = W'$ (cf.~\eqref{eq:O'}).
We further note that Poincar\'e covariance and
$W_R' = - W_R$, resp., $h_{W_R'} = - h_{W_R}$,
imply that $\Lambda_{W'}(t)=\Lambda_W(-t)$, resp., $h_{W'}=-h_W$.
 Note that
\[ \cP_+=\cP_+^\up\cup\cP_+^\downarrow=\cP_+^\up\cup j_{W_R}\cP_+^\up.\] 

As $Z(\cP_+^\up) = \{e\}$, the projection map
  $q : \cG_E(\cP_+) \to \cE(\L(\cP)) = \cO_{h_W}$ is bijective,
  so that the abstract wedge space  $\cG_E(\cP_+)$ can be 
  identified with the set $\cE(\L(\cP))$ of Euler elements.  Associating to each 
  Euler element $h$, considered as an affine vector field on $M$, 
  its positivity region
  \begin{equation}
    \label{eq:WM+}
    W_M^+(h) = \{ x \in M \colon h(x) \in C^\circ\},
  \end{equation}
  we further obtain a bijection from
  $\cE(\L(\cP))$ onto the set of wedge regions in $M$
  (\cite[Lemma~4.12]{NO17}).

For $d>1$, there exists for the Euler elements $h_{W_R}$, associated to $W_R$,
Euler elements $h_{W_{s}}\in\cO_{h_{W_R}}$, associated to 
the wedges
\begin{equation}
  \label{eq:Ws}
  W_{s}:=\{x\in\RR^{1,d}:|x_0|<x_s\},\quad s=2,\ldots,d,
\end{equation}
such that all pairs $(h_{W_R}, h_{W_s})$ are orthogonal
(cf.~\cite[Sect.3]{MN21}). For $d=1$,
the Poincar\'e--Lie algebra contains no orthogonal pair
of Euler elements. 
For more details we refer to \cite{Mo25} and the references therein.
\end{example}

\begin{example}\label{ex:mob} (a) (M\"obius group on the circle,
  \cite[Ex.~2.10(c)]{MN21})
The M\"obius group $G := \Mob :=\PSL_2(\R)$ 
acts on the compactification $\R_\infty = \R \cup \{\infty\}
\cong \bS^1$ 
of the real line by 
\[ g.x := \frac{a x + b}{cx + d}  \qquad 
\mbox{ for }\quad g = \pm\pmat{a & b \\ c & d}\in \PSL_2(\R).\] 
The group $\Mob$ is generated by any two of the following
one-parameter subgroups (cf.~Example~\ref{ex:sl2}):
 \begin{itemize}
 \item Dilations: $\delta(t)(x)=  e^t x$ for $t \in \R$,
   generated by the Euler element 
   $h_0 := \displaystyle{\frac12 \pmat{1  & 0 \\ 0 & -1}}$.
 \item Rotations: $\rho(\theta).x
   = \frac{\sin(\theta/2) + \cos(\theta/2)x}{\cos(\theta/2) - \sin(\theta/2)x}$,
   generated by $z_0 := \displaystyle{\frac{1}{2}\pmat{ 0 & 1 \\ -1 & 0}}$.
    \item $\kappa(t).x
   = \frac{\sinh(t/2) + \cosh(t/2)x}{\cosh(t/2) + \sinh(t/2)x} 
   =\frac{e^{t}(1+x) + x-1}{e^t(1+x)+1 - x}$ for $t \in \R$,
      generated by the Euler element
   $k_0 := \displaystyle{\frac12 \pmat{0 & 1 \\ 1& 0}}$. 
 \end{itemize}

We call  a non-dense, non-empty open connected subset 
$I \subeq \bS^1$ an {\it interval} and 
write $\cI(\bS^1)$ for the set of intervals in $\bS^1$. 
The stabilizer of the positivity region
$W^+_{\bS^1}(h_0) = (0,\infty)$ in $G$ is 
the dilation subgroup $\delta(\R)$, which coincides with 
the stabilizer of $h_0$ under the adjoint action. 
So
  \begin{equation}
    \label{eq:moeb-euler}
 \cW_+(W_0) \to \cE(\fsl_2(\R)) \to \cI(\bS^1), \quad
 (h,\sigma) \mapsto h \mapsto W_{\bS^1}^+(h)
  \end{equation}
are equivariant bijections.
In particular, we obtain for an interval $I = W^+_{\bS^1}(h)$ 
the reflection $\tau_I=\tau_h$ and the one-parameter group 
$\delta_I(t) = \exp(th)$. The involution $\tau_I$ is
orientation reversing and maps $I$
to the complementary open interval that we denote~$I'$. 

\nin (b) (Covering of the M\"obius group on the line;
\cite[Ex.~2.10(d)]{MN21}) Consider  the universal covering $\tilde G$
of the M\"obius group and let $\tilde G_{\tau_{h_0}}
:= \tilde G\rtimes \{\1,\tau_{h_0}\}$, where
$\tau_{h_0}$ is the Euler involution, integrated to $\tilde G$. Let $q_G:\tilde G_{\tau_{h_0}} \rightarrow G_{\tau_{h_0}}$ be the covering map.
The action of $G$ on $\bS^1$ lifts 
to an action of the connected group $\tilde G$
on the universal covering $\tilde\bS^1 \cong \R$, 
where $\tilde\rho(\theta).x = \theta + x$
is the lift to $\tilde G$ of the one-parameter group
$\rho(\theta)$ of rotations in $G$. We specify the simply connected covering
\[ q_{\bS^1} \colon \R \to \R_\infty, \quad 
q_{\bS^1}(\theta) = q_{\bS^1}(\tilde\rho(\theta).0) = \rho(\theta).0.\] 
The involution
$\tau_{h_0}$ acts on the covering $\RR$ as the point reflection $\tau_{h_0}.x = - x$ in the base 
point~$0$. We also note that
$ \ker(q_G)  = \tilde\rho(2\pi \Z)= Z(\tilde G)\simeq \Z$ is 
the group of deck transformations of the covering~$q_{\bS^1}$, which acts by 
\begin{equation}
  \label{eq:deck}
  \tilde\rho(2\pi n).x = x + 2 \pi n \quad \mbox{ for } \quad
  n \in \Z, x\in\RR.
\end{equation}

We call  a non-empty interval $I \subeq \R$ {\it admissible} 
if its length is strictly smaller than $2\pi$ 
and write $\cI(\R)$ for the set of admissible intervals. 
An interval $I \subeq \R$ is admissible if and only if 
there exists an interval $\uline I \in \cI(\bS^1)$ 
such that $I$ is a connected component of $q_{\bS^1}^{-1}(\uline I)$. 
The group $Z(\tilde G)$ 
acts {simply} transitively on the set of these connected components. 
As $G$ acts transitively on $\cI(\bS^1)$, \eqref{eq:deck} shows that 
$\tilde G$ acts transitively on the set $\cI(\R)$.
We have an equivariant bijective map 
\[ \cW_+(\tilde G)\ni g.(h_0,\tau_{h_0})
\mapsto {g.(0,\pi)}\in\cI(\RR), \qquad g\in \tilde G.\] 
The connected components of the positivity region
  $W_\R^+(h_0)$ (cf.\ \eqref{eq:WM+}) of the Euler  element $h_0$
  are the intervals
  \begin{equation}
    \label{eq:I2n}
    \tilde I_{2n} := (2n\pi,(2n+1)\pi), \quad n \in \Z.
    \end{equation}
    Likewise the intervals $\tilde I_{2n+1}$, $n \in \Z$, are
    the connected components of the positivity domain of $-h_0$.
    Furthermore, $\tilde G^{h_0} = Z(\tilde G) \exp(\R h_0)$
    implies that $Z_2(\tilde G)=2Z(\tilde G)$ and
      $Z_3(\tilde G) = Z(\tilde G)$. 
For further details, see \cite[Ex.~2.10(d)]{MN21}.
\end{example}

\subsection{Space- and timelike orthogonal pairs}\label{sect:stort}

Let $(\g, C)$ be a simple Lie algebra and
$C \subeq \g$ be a pointed generating $\Inn(\g)$-invariant
closed convex cone.
We recall that such a cone exists if and only if
$\g$ is {\bf hermitian} (cf.\ \cite{MNO23}, \cite{Vi80}). 
Let $h \in \cE(\g)$ be an Euler element (so that $\g$ is of tube
type by \cite{MN21}) and we obtain the symmetric Lie algebra
$(\g, \tau_h^\g)$ 
(cf.~Appendix~\ref{app:list}).
We use the notation from Theorem~\ref{thm:1.4}. 
Here the restricted root system $\Sigma(\g,\fa)$ is of type C$_r$
and $\Sigma(\g^*, \fc)$ is of type A$_{r-1}$, so that 
the $G^h$-orbits in $\cE(\g) \cap \fq$ are represented by elements
\[ k^0, k^1, \ldots, k^r.\]
So we have $r+1$ pairs $(h,k^j)$, $j = 0,\ldots, r$, representing
the $\Inn(\g)$-conjugacy classes of pairs of orthogonal Euler elements
in~$\g$.

It is instructive to describe these elements in a subalgebra
  of the form 
\[ \fs := \fs_1 + \cdots + \fs_r \cong \fsl_2(\R)^{\oplus r},\]
where 
\begin{equation}
  \label{eq:h-eq}
  h  = (h_0,\cdots, h_0) \quad \mbox{ and } \quad 
k^j := (\underbrace{k_0, \cdots k_0}_{j\ \text{times}},
  \underbrace{-k_0, \cdots, -k_0}_{r-j\ \text{times}}).
\end{equation}
(\cite[Cor.~3.5]{MNO23}). 
It contains the $3$-dimensional Lie subalgebras
\[ \fs^j := \Spann \{ h, k^j, [h,k^j]\} \cong \fsl_2(\R).\]

We recall from the introduction that
the orthogonal pairs $(h,k)$ occur in $3$ types:
\begin{itemize}
\item {\it positive/negative timelike}: $[h,k] \in \pm C$, and 
\item {\it spacelike} $[h,k] \not\in C \cup - C$. 
\end{itemize}

\begin{prop} \mlabel{prop:timelike-pair}
  We have $\fs^j \cap C \not=\{0\}$ if and only if
  $j = 0,r$.  More specifically, the pair $(h,k^0)$ is negative
  timelike, and $(h,k^r)$ is positive timelike. In particular, 
  $\Inn(\g)$ acts transitively on the set of negative/positive timelike
  pairs. 
\end{prop}

\begin{prf} We write $e_j \in \fs_j$ for the element corresponding to $e_0 \in \fsl_2(\R)$. First we observe that 
  \[ C_\fs := C \cap \fs = \sum_{j = 1}^r C_j \quad \mbox{ with } \quad
    C_j := C \cap \fs_j \quad \mbox{ and }\quad
    C_j \cap \g_1(h) = [0,\infty) e_j.\] 
  It follows that $C_\fs$ is a pointed generating invariant
  cone in $\fs$. This Lie algebra contains the $2^r$ invariant 
  cones $\sum_{j=1}^r \eps_j C_j$ with $\eps_j \in \{\pm 1\}$. 

  We have
  \[ [h, k^j]
    = \sum_{\ell=1}^j (e_\ell + \theta(e_\ell))
   -  \sum_{\ell=j+1}^r (e_\ell + \theta(e_\ell))
   = e^j + \theta(e^j) \quad \mbox{ for } \quad
   e^j = e_1 + \cdots + e_j - e_{j+1} - \cdots - e_r.\] 
 So $\fs^j \cap \g_1(h) = \R e^j$, and
 $C \cap \fs^j$ is non-zero if and only if
 $C \cap \R e^j \not=\{0\}$. Here we use that
 every closed generating invariant cone in $\fsl_2(\R)$ 
 intersects the restricted root spaces.
 Now
 \[ C \cap \fs_1(h) = C \cap (\R e_1 + \cdots + \R e_r)
   = [0,\infty) e_1 + \cdots + [0,\infty) e_r\]
 implies that $C \cap \R e^j$ is non-zero if and only if
 $j \in \{0,r\}$.
In particular, we have
 \[ [h,k^0] = -e - \theta(e) \in - C \quad \mbox{ and } \quad 
   [h,k^r] = -[h,k^0] = e +  \theta(e) \in  C.\]
 This completes the proof.
\end{prf}

\section{Geometric Bisognano--Wichmann and \\ Spin--Statistics}
\mlabel{sec:spinstat}

We keep the context of the preceding section, where $G$ is
a connected Lie group, $h \in \g$ an Euler element
and $\tau_h$ the corresponding involution on $G$.
We further fix a pointed generating closed convex
  $\Ad^\eps$-invariant cone $C\subeq\fg$, and the abstract wedge space
\[ \cW_+ = G.W_0 \subeq \g \times G_{\tau_h} \quad \mbox{ with } \quad 
  W_0 = (h,\tau_h) \]
from \eqref{eq:cw+b} in Section~\ref{subsec:4}, endowed with the
$C$-order.

\subsection{Nets of standard subspaces on abstract wedges}
\mlabel{subsec:nets-real}

In this section we recall the concept of
a (one-particle) net of standard subspaces   and its relevant properties.
Fundamental facts on 
standard subspaces are recalled in Appendix~\ref{subsec:stand-subs}.
  
\begin{defn}\label{def:ssnet} 
Let $(U,\cH)$ be a unitary representation of $G$,
$\Stand(\cH)$ the set of standard subspaces of $\cH$, and 
\begin{equation}
  \label{eq:neta}
\sH \colon \cW_+ \to \Stand(\cH) 
\end{equation}
be a map (see~\eqref{eq:cw+} for $\cW_+$),
also called a {\it net of standard subspaces on abstract wedges}. 
In the following we denote this data as $(\cW_+,U,\sH)$.
If we want to stress the dependence of the net $\sH$
on the geometric and the representation data, we refer to it as a
quadruple $(G,\cW_+, U,\sH)$.

We consider the following properties: 
\begin{itemize}
\item[\rm(HK1)] {\bf Isotony (Iso):} $\sH(W_1) \subeq \sH(W_2)$ for $W_1 \leq W_2$
  in the sense of \eqref{eq:cG-ord}. 
\item[\rm(HK2)] {\bf Covariance (Cov):} $\sH(gW) = U(g)\sH(W)$ for 
$g \in {G}$, $W \in \cW_+$. 
\item[\rm(HK3)] {\bf Spectral condition (SC):} 
$ C \subeq C_U := \{ x \in \g \colon -i\cdot \partial U(x) \geq 0\}.$ 
We then say that $U$ is {\it $C$-positive}.  
\item[\rm(HK4)]{\bf Central twisted locality (CTL):} 
For every $\alpha \in Z(G)^-$
with $W_0^{'\alpha} = (-h, \alpha \tau_h) \in\cW_+$, there exists a
unitary $Z_\alpha\in U(G)'$ satisfying 
\begin{equation}
  \label{eq:ZJ2}
 \quad 
J_{\sH(W_0)} Z_\alpha J_{\sH(W_0)}=Z_\alpha^{-1}, 
\end{equation}
and such that 
\begin{equation}\label{eq:tcl2}\sH(W_0^{'\alpha}) \subseteq Z_\alpha \sH(W_0)'.
\end{equation} 

\item[\rm(HK5)] {\bf Bisognano--Wichmann property (BW):} 
$U(\exp(th)) = \Delta_{\sH(W_0)}^{-it/2\pi}$ for $t\in \R$. 
\item[(HK6)] \textbf{Central twisted Haag Duality (CTHD)}:
 For every $\alpha \in Z(G)^-$
 such that $W_0^{'\alpha}\in\cW_+$, there exists 
$Z_\alpha \in U(G)'$ satisfying \eqref{eq:ZJ2} and 
\[ \sH(W_0^{'\alpha})=Z_\alpha \sH(W_0)'.\]
\item[(HK7)] \textbf{Twisted} $G_{\tau_h}$\textbf{-covariance  (TCov)}: 
  For every 
  $\alpha \in Z(G)^-$ such that $W_0^{'\alpha}\in\cW_+$, there exists an
  extension $U^\alpha$  of $U$ to an
  antiunitary representation of $G_{\tau_h}$, such that:
\begin{equation}
  \label{eq:twitcovara}
 \sH(g *_\alpha W) = U^\alpha(g) \sH(W) \quad \mbox{ for } \quad 
g \in G_{\tau_h}, W \in \cW_+,
 \end{equation}
where $*_\alpha$ is the $\alpha$-twisted action of 
$G_{\tau_h}$ on $\cG_E(G_{\tau_h})$, defined by
\begin{equation} \label{eq:twistact}
 g *_\alpha (x,\sigma) 
:=\begin{cases}
g.(x,\sigma) & \text{ for } g \in G \\ 
g.(x,\alpha^{-1}\sigma) & \text{ for } g \in G \tau_h
\end{cases}
\end{equation} 
(see \cite{MN21} for more details).  
\item[(HK8)] \textbf{Modular reflection (MRef):} 
(HK7) holds for the extension $U^\alpha$, specified by 
\[ U^\alpha(\tau_h)= Z_\alpha J_{\sH(W_0)},\]  
where $Z_\alpha$ is a unitary satisfying \eqref{eq:ZJ2}.
\item [\rm(SS)] \textbf{Spin--Statistics:} (HK4) holds and, in addition, 
  $Z_\alpha^2 = U(\alpha).$
\item[(Reg)] \textbf{Regularity:}
  There exists an $e$-neighborhood $N \subeq G$ such that
$\bigcap_{g \in N} U(g)\sH(W_0)$ 
  is cyclic. 
\end{itemize}

   \end{defn}

\begin{defn} \mlabel{def:regular-net} 
We call the net $\sH$ {\it regular} if it satisfies (Reg),
(Iso), (Cov) and (SC). 
\end{defn}

\begin{remark}
  (a) The properties (Cov), (SC) and (BW) imply isotony (Iso)
  by \cite[Prop.~4.10]{MN21}. 

\nin (b) If $C=\{0\}$, then the order on $\cG_E(G_{\tau_h})$ is trivial, so that
  (Iso)  and (SC) hold trivially.

 \nin (c) For the sake of easier verifiability, we formulate
    conditions (HK4-6) and (HK8) as conditions on $W_0$, but
    if the covariance condition (Cov) is satisfied,
    the homogeneity of $\cW_+ = G.W_0$ implies that 
    these conditions hold for every $W \in \cW_+$ if 
    they hold for $W_0$.

  \nin (d) If $h$ is not symmetric, i.e., $-h \not\in \cO_h$, then
    (CTL) = (HK4), (SS) and (HK6-8) hold trivially because
    $W^{'\alpha}$ is never contained in $\cW_+$.     
\end{remark}

Let $(U,\cH)$ be an antiunitary representation of $G_{\tau_h}$ and
$W  = (x,\sigma) \in \cG_E$ be an Euler wedge.
Then we have a natural map 
\begin{equation}
  \label{eq:bgl}
  \sH_U \colon \cG_E(G_{\tau_h}) \to \Stand(\cH), \qquad
  \sH_U(W):=\Fix(U(\sigma) e^{\pi i \partial U(x)}), 
\end{equation}
called the {\it Brunetti--Guido--Longo (BGL)  net} associated to $U$.
By the following theorem, it satisfies all the
  assumptions in Definition \ref{def:ssnet}. Therefore 
  a key feature of our setting is that it starts from a unitary
  representation $U \colon G\to \U(\cH)$. Such a representation may
  have no antiunitary extension to $G_{\tau_h}$, but if there is
  one, then it follows from \cite[Thm.~2.11]{NO17} that they
  are all conjugate under unitary operators
  $T \in U(G)'$. If $J = U(\tau_h)$, then $J_T = TJT^{-1}$ corresponds to the
  conjugate representation, so that all these extensions are parametrized
  by the operators~$J_T$. If, in addition, $JTJ = T^{-1}$, then
  $J_T = T^2 J$, and these modifications appear in the (CTL) condition.
  Theorem~\ref{thm:ee} provides an important sufficient condition for the
  existence of an extension.
  
  \begin{theorem} \mlabel{thm:BGL}
    {\rm(\cite[Thm~4.12, Prop.~4.16]{MN21})} 
    Let $(U,\cH)$ be an antiunitary
    $C$-positive representation of $G_{\tau_h}$, where $h\in \cE(\fg)$. 
    Then the BGL net $\sH_U$
 satisfies {\rm(HK1)-(HK8)} and \rm{(SS)} in {\rm Definition~\ref{def:ssnet}}.
\end{theorem}

\begin{thm} \mlabel{thm:ee} {\rm(Euler Element Theorem)}
  {\rm(\cite[Thm.~3.1]{MN24})} 
  Let $G$   be a connected Lie group with 
  Lie algebra $\g$ and $h \in \g$. 
  Let $(U,\cH)$ be a unitary 
  representation of $G$ with discrete kernel. 
  Suppose that $\sV$ is a standard subspace
    and $N \subeq G$ an identity neighborhood such that 
  \begin{itemize}
  \item[\rm(a)] $U(\exp(t h)) = \Delta_\sV^{-it/2\pi}$ for $t \in \R$,
    i.e., $\Delta_\sV = e^{2\pi i \, \partial U(h)}$, and 
  \item[\rm(b)] $\sV_N := \bigcap_{g \in N} U(g)\sV$ is cyclic.
  \end{itemize}
  Then $h$ is {central or}
  an Euler element, and the conjugation $J_\sV$ satisfies
  \begin{equation}
    \label{eq:J-rel}
 J_\sV U(\exp x) J_\sV = U(\exp \tau_h^\g(x)) \quad \mbox{ for } \quad
\tau_h^\g = e^{\pi i \ad h}, x \in \g.
  \end{equation}
\end{thm}

\begin{rem}\label{rmk:ext} (a) 
Theorem \ref{thm:ee} does not require an antiunitary
    extension of $U$ to $G_{\tau_h}$, only a unitary representation of~$G$.
    Actually, it is an important consequence of \eqref{eq:J-rel}
    that $\tau_h^\g\in \Aut(\g)$ 
    integrates to an involutive automorphism 
    on the group $\uline G := U(G) \cong G/\ker(U)$
(cf.~\cite[Remark 3.6]{MN24}).
Whenever the involution $\tau_h$ on $G$ exists,
$U$ extends by \eqref{eq:J-rel}
to an antiunitary representation of $G_{\tau_h}$.
This is in particular the case
on the simply connected  covering group $q_G \colon \tilde G \to G$, 
and the corresponding unitary representation $U \circ q_G$ of~$\tilde G$.

\nin (b) In view of (a), up to substituting $G$ by~$\uline G$,
the representation $U$ extends antiunitarily to
$\uline G_{\tau_h}$ by $U(\tau_h)=J_{\sV}$, so that 
$\sV=\Fix(U(\tau_h)e^{\pi i\cdot\partial U(h)})$ holds by Theorem \ref{thm:ee}.
By construction, the BGL net $\sH_U$ satisfies
\[  \sH_U(g.W_0) = U(g)\sV \quad \mbox{ for } \quad g \in G,
W_0=(h,\tau_h)\in\cG_E(G_{\tau_h}).\]
\end{rem}

  \begin{prop} \mlabel{prop:BGL-natural} {\rm(Naturality of the
      BGL net)} 
    For $j = 1,2$, let $G_j$ be connected Lie groups and $h_j \in \g_j$
    Euler elements for which the groups $G_{j, \tau_{h_j}}$ exist.
    Further, let $\phi \colon G_{1, \tau_{h_1}}  \to G_{2, \tau_{h_2}}$
      be a morphism of Lie groups with 
      $\L(\phi)h_1 = h_2$ and $\phi(e,\tau_{h_1}) = (e,\tau_{h_2})$.
      Then  we obtain a natural map 
    \[ \phi^{\cW} \colon
      \cW_+^{G_1} = G_1.(h_1, \tau_{h_1}) \to \cW_+^{G_2} = G_2.(h_2, \tau_{h_2}),
      \quad (x,\sigma) \mapsto (\L(\phi)x, \phi(\sigma)).\]
    If $(U_j, \cH_j)_{j = 1,2}$ are antiunitary 
    representations of $G_{j, \tau_{h_j}}$, 
    and $\Phi\colon \cH_1 \to \cH_2$ is unitary,  satisfying
    $U_2 \circ \phi = \Phi \circ U_1(\cdot) \circ \Phi^{-1},$
    then the corresponding BGL nets satisfy
    \begin{equation}
      \label{eq:h-nat}
      \sH^{\rm BGL}_{U_2} \circ \phi^{\cW}= \Phi \circ \sH^{\rm BGL}_{U_1}
      \colon \cW_+^{G_1} \to \Stand(\cH_2). 
    \end{equation}
  \end{prop}

  \begin{prf} The proof is an easy  consequence of the definitions.
    As both sides  of \eqref{eq:h-nat} are $G_1$-covariant
    with respect to the representation
    $U_2 \circ \phi$ of $G_1$ on $\cH_2$, it suffices to verify
    that they take the same values on $W_1 = (h_1, \tau_{h_1})$.
    For the map on the left, this is the standard subspace~$\sV_L$ 
    with modular objects 
    \[ \Delta_L := e^{2\pi i\cdot \partial U_2(h_2)}
      \quad \mbox{ and }  \quad J_L = U_2(\phi(\tau_{h_1}))
      = U_2(\tau_{h_2}).\]
    For the map on the right, we obtain the standard subspace
    $\sV_R = \Phi(\sV_1)$, where
    \[ \Delta_1 := e^{2\pi i \cdot\partial U_1(h_1)}
      \quad \mbox{ and }  \quad J_1 = U_1(\tau_{h_1}).\]
    Hence the assertion follows from
    \[ \Phi \Delta_1 \Phi^{-1} = e^{2\pi i \cdot\partial U_2(h_2)}
      \quad \mbox{ and }  \quad \Phi J_1 \Phi^{-1} = U_2(\tau_{h_2}).\qedhere\]
 \end{prf}

  \begin{remark} \mlabel{ex:quotientBGL} (a) If $(U,\cH)$ is an
    antiunitary representation
    of $G_{\tau_h}$ with discrete kernel
    and $\uline G := G/\Gamma$, for a discrete central subgroup
    $\Gamma \subeq \ker(U)$ {such that $\tau_h(\Gamma)=\Gamma$},
    then the quotient map
    $\phi \colon G \to \uline G$ satisfies $\L(\phi) = \id_\g$
    (if we identify the Lie algebras of both groups in the natural way).
    Now Proposition~\ref{prop:BGL-natural} applies to $\phi$ and
    leads to a natural map
    $\phi^\cW \colon \cW_+^G \to \uline \cW_+^{\uline G}$, for which 
    \[  \sH^{\rm BGL}_{\uline U} \circ \phi^{\cW}= \sH^{\rm BGL}_{U},\]
and the representation $\uline U(g\Gamma) := U(g)$ 
is well-defined and satisfies $U = \uline U \circ \phi$
by Theorem~\ref{thm:spst1}.

    \nin (b) The preceding construction works in particular
    for $\Gamma = \ker U$ if $\ker U$ is discrete.
\end{remark}

\begin{lem}\label{lem:fact}
Suppose that $G_1$ and $G_2$ are connected Lie groups
    with the same Lie algebra~$\g$, and $h \in \cE(\g)$ is such that
    $\tau_h$ integrates to $G_1$ and $G_2$. If
    $\phi \colon G_{1,\tau_{h}} \to G_{2,\tau_{h}}$ is a
    covering morphism with discrete kernel $D\subeq Z(G_1)$
    and $(U, \cH)$ an antiunitary
    representation of $G_{1,\tau_{h_1}}$, then
the BGL net $\sH_{U}  \colon  \cW_+^{G_1} \to \Stand(\cH)$ 
  factors through a map 
  \[ \oline  \sH_{U}  \colon  \cW_+^{G_2} \to \Stand(\cH)
    \quad \mbox{ with }  \quad
    \oline \sH_{U}  \circ \phi^\cW = \sH_{U}  \]
  if and only if 
  \begin{equation}
    \label{eq:d-incl}
    \partial_h(G_1^{h_1}) \cap D = Z_2(G_1) \cap D \subeq \ker(U).
  \end{equation}
\end{lem}

\begin{prf} We consider the action of $G_1$ on
  $\cG_E(G_{2,\tau_{h_2}})$ induced by $\phi$.
  Then the stabilizer of $(h_2, \tau_{h_2}) \in \cG_E(G_{2,\tau_{h_2}})$ in $G_1$ is
  \begin{equation}
    \label{eq:stabil}
 G_1^{(h_2, \tau_{h_2})} = \{ g \in G_1^{h_1} \colon \partial_h(g) \in D \}.
  \end{equation}
  The BGL net $\sH_{U}$ on $\cW_+^{G_1}$
  factors through a map  on $\cW_+^{G_2} \cong
  G_2/G_2^{(h_2, \tau_{h_2})}$ if and only if $G_1^{(h_2, \tau_{h_2})}$ fixes
  $\sH_{U} (h_1,\tau_{h_1})$. By \eqref{eq:stabil}, this means that,
  for $g \in G_1^{h_1}$ {with $\partial_h(g) \in D$,} we have 
  $\partial_h(g) \in \ker(U)$. This is \eqref{eq:d-incl}.
\end{prf}

  As the left hand side of \eqref{eq:d-incl} may be a proper subgroup of $D$,
  this condition does {\bf not} imply that $D \subeq \ker(U)$,
  i.e., that $U$ factors through an antiunitary representation
  of $G_{2,\tau_{h_2}}$. We only need that the intersection
   $Z_2(G_1) \cap D$ is contained in $\ker(U)$. If, however,
   $D \subeq Z_2(G_1)$ and $D \subeq \ker U$,
   then $U$ factors through~$G_2$.

   \subsection{Results on nets of standard subspaces}

In the presence of symmetric Euler elements,
  the following theorem derives the Bisognano--Wichmann property and the
  Central Twisted Haag Duality
  for an isotone covariant, central twisted local
  net of standard subspaces from the spectral condition.  
  We further assume that the $\Ad^\eps(G_{\tau_h})$-invariant cone $C$ is
  ``large'' in a suitable sense that will be specified below.

  \begin{theorem}\label{thm:1BW}
    {\rm(Geometric Bisognano--Wichmann Theorem)} 
Let $\sH \colon \cW_+\to \Stand(\cH)$ be a net of standard subspaces
satisfying {\rm(Iso), (Cov), (SC), (CTL)}. 
Assume that the cones $C_\pm = \pm C \cap \g_{\pm 1}$ have inner points, 
  that $\g_0 = [\g_1, \g_{-1}]$, and
  that $\ker U$ is discrete. 
  Suppose further that $W_0^{'\alpha}\in\cW_+$ for some
    $\alpha \in Z(G)^-$, i.e., that $h$ is symmetric. Then
  the Bisognano--Wichmann property {\rm (BW)} and
  twisted Haag duality {\rm(CTHD)}
  hold.
\end{theorem}

\begin{proof} For $W = (x,\sigma) \in \cW_+$, we  define  the map 
  \[  z_W \colon \R \to \U(\cH), \quad
    z_W(t) := \Delta_{\sH(W)}^{it} U(\exp(2\pi t x))  \]
  If $g\in G^W$, then by covariance $U(g)$ commutes with $z_W$.
  By covariance (Cov) and Lemma~\ref{lem:sym}, 
  $\Delta_{\sH(W)}^{it}$
  commutes with  $U(\exp(2\pi t x))$, so that
  $z_W$ defines a one-parameter group.

  Let $y \in C_\pm$, so that $\mp i \partial U(y) \geq 0$ by (SC),  
 and thus $U(\exp ty) \sH(W_0) \subeq \sH(W_0)$ for $t \geq 0$
 by (Iso/Cov).  Therefore Borchers--Wiesbrock Theorem~\ref{Borch} implies the relation
\[   \Delta_{\sH(W_0)}^{-it/2\pi}  U(\exp y) \Delta_{\sH(W_0)}^{it/2\pi}
  = U(\exp e^{\pm t} y).\]
Thus 
\[ U(\exp th) U(\exp y) U(\exp th)^{-1} =  U(\exp e^{\pm t} y) \]
shows that  $z_{W_0}$ commutes with $\exp(C_\pm)$,
so that $\g_{\pm 1} = C_\pm - C_\pm$ entails that it commutes with
$G_{\pm 1}$. By assumption, $\g$ is generated  by
  $\g_{\pm 1}$.
Therefore the arcwise connected 
(hence integral by Yamabe's Theorem \cite[Thm.~9.6.1]{HN12}) subgroup
of $G$ generated by the two subgroups~$G_1$ and $G_{-1}$
coincides with~$G$, so that $z_{W_0}(\R) \subeq U(G)'$.
By covariance, this entails that
 $z_{g.W_0} = U(g) z_{W_0} U(g)^{-1} = z_{W_0}$ for $g \in G$,
 i.e., $z_W$ does not depend on~$W$.
For $\alpha \in Z(G)^-$ with $W_0^{'\alpha} \in \cW_+$, this further leads to
\begin{equation}
  \label{eq:zw1}
 z_{W_0}(t)\sH(W_0^{'\alpha})=z_{W_0^{'\alpha}}(t)
 \sH(W_0^{'\alpha})=\sH(W_0^{'\alpha}).
\end{equation}

By the covariance condition (Cov), the two subspaces
$\sH(W_0^{'\alpha})$ and $Z_\alpha\sH(W_0)'$ are invariant
under $U(\exp \R h)$, and by \eqref{eq:zw1} also under
the modular group $\Delta_{\sH(W_0)}^{i\R}$. Now
their equality, which is (CTHD), 
follows from (CTL) and Lemma~\ref{lem:inc}.

Further $z_W = z_{W^{'\alpha}}$ leads to 
\[ z_{W^{'\alpha}}(t)=\Delta_{\sH(W^{'\alpha})}^{it}U(\exp(-2\pi tx))
  =\Delta_{\sH(W)}^{-it}U(\exp(-2\pi tx))=z_W(-t)=z_{W^{'\alpha}}(-t)\]
and since $z_W = z_{W^{'\alpha}}$ is a one parameter group, it is constant~$\1$.
This means that (BW) is satisfied. 
\end{proof}

The following theorem proves the regularity property for a net of standard subspaces that is not necessary the BGL net (cf.~\eqref{eq:bgl}).
Based on the Euler Element Theorem~\ref{thm:ee}, the
regularity property for BGL nets is treated in \cite[Sect.~4]{MN24};
see also \cite{BN25} and \cite[\S 5.3]{Ne26}. 

\begin{theorem}\label{thm:reg} {\rm(Regularity Theorem)} 
  Let $\sH \colon \cW_+\to \Stand(\cH)$ be a net of standard subspaces.
 Suppose that $C_\pm = \pm C \cap \g_{\pm 1}$ generate  $\fg_{\pm1}$ and that $\ker U$ is discrete. 
  Then the following implications hold:
\begin{enumerate}
\item[{\rm(a)}] {\rm(Iso)} and  {\rm(Cov)} imply {\rm(Reg)}. 
\item[{\rm (b)}] {\rm(SC)}, {\rm(BW)}
and  $\g_0 =\R h+ [\g_1, \g_{-1}]$ imply {\rm(Reg)}.   
\end{enumerate}  
\end{theorem}

\begin{proof} (a) Suppose that (Iso) and (Cov) hold.
  As the cones $C_\pm$ have non-empty interior in $\g_{\pm 1}(h)$
  and $\L(G^{W_0}) = \g_0(h)$,   the semigroup 
  \[ S_{W_0} := \exp(C_+) G^{W_0} \exp(C_-) \]
  has an interior point $s_0$. Let $N \subeq G$ be an identity neighborhood
  with $N^{-1} s_0 \subeq S_{W_0}$. From (Iso/Cov)  we then derive that,
  for $g \in N$ we have $U(g^{-1} s_0) \sH(W_0) \subeq \sH(W_0)$, so that
  $U(s_0) \sH(W_0) \subeq \bigcap_{g \in N} U(g)\sH(W_0)= \sH(W_0)_N$.
  
\nin (b): If (SC) and (BW) are satisfied, then
Borchers' Theorem~\ref{Borch} implies that
\[ \exp(C_+)\exp(\RR h)\exp(C_-)\subseteq S_{\sH(W_0)}
:= \{ g \in G \colon U(g) \sH(W_0) \subeq \sH(W_0) \}.\] 
Hence the Lie wedge $\L(S_{\sH(W_0)})$, contains $C_\pm$ and $h$.
As $C_\pm$ span $\fg_{\pm1}$ and  $\g_0 = \R h+  [\g_1, \g_{-1}]$,
the Lie algebra generated by $\L(S_{\sH(W_0)})$ coincides with $\g$.
So $S_{\sH(W_0)}$ has interior points by \cite[Thm.~3.8]{HN93}.
Now (Reg) follows as under (a). 
\end{proof}

  \begin{remark} \label{rmk:hermgen}
(a) The conditions that the cones  $C_\pm$ have non empty interior in $\fg_{\pm1}$ or, equivalently, that they generate $\fg_{\pm1}$,
    is satisfied when $\fg$ is a hermitian simple Lie algebra.
    Indeed, for any such Lie algebra, 
      there exists  a nontrivial pointed closed convex $\Inn(\fg)$-invariant
      cone $C \subset\fg$ (\cite[Prop.~3.11(a)]{MN21}).
Then $C \cap - C$ and $C-C$ is are $\ad \fg$-invariant subspaces,
    and since $\g$ is simple, $C$ is pointed and generating.
Then the projections $p_{\pm 1} \colon \g \to \g_{\pm 1}$ satisfy 
 $C_\pm=p_{\pm1}(C)$ by \cite[Lem.~3.2]{NOO21}, and thus both cones
 $C_\pm$ have non-empty interior in  $\fg_{\pm 1}$.
Hermitian simple real Lie algebras containing
   Euler elements are listed in Table 2 in Appendix~\ref{app:list}. 
 
    {\nin (b) The condition that the cones $C_\pm$ generate $\fg_{\pm 1}$ does not apply to the Poincar\'e group. Indeed, in this case
 { $C = \oline{V_+}$ and
$\fg_1 \cap C$ is a half-line not generating $\g_1$.}
In particular, the above argument does not apply to Poincar\'e covariant nets of standard subspaces. An algebraic proof of the Bisognano--Wichmann property for Poincar\'e covariant nets of standard subspaces was given in \cite{Mo18}.
Lorentz covariant nets of standard subspaces on de Sitter space also do not fit into this framework, since $C = \{0\}$ in this case, so that the generating property fails. 
On the other hand, the condition is satisfied for the Lie algebra
of the conformal group $\PSO_{2,d}(\R)_e$ of Minkowski space, 
as well as for the M\"obius group $\Mob$ of the chiral circle
{(the case $d = 1$)}. Their Lie algebras are simple hermitian.
In this sense, Theorem \ref{thm:1BW} generalizes
the conformal Bisognano--Wichmann
theorems in \cite{BGL93,Lo08,DLR07}.
}
  \end{remark}

  In our context, the following Spin--Statistics Theorem 
    is a consequence of regularity and covariance, which entail
    that the net $\sH$ under consideration is the BGL net of an
    antiunitary extension of $U$.

  \begin{thm}{\rm(Spin--Statistics Theorem for nets of standard subspaces)}
    \mlabel{thm:spst1}
    Let
    \[ \sH \colon \cW_+ \to \Stand(\cH)\] 
  be a regular net of standard subspaces for
  the  unitary representation $(U,\cH)$ of $G$, such
  that the {\rm(BW)} condition  is satisfied.
  Assume that $U$ has discrete kernel and that $h$ is symmetric.
  Then $\sH$ also satisfies~{\rm(CTL)}, {\rm(HK6)-(HK8)} and~{\rm(SS)}, 
  and $\sH$ coincides with the BGL net $\sH_{\tilde U}$
  of the antiunitary extension of $U$ specified by
  $\tilde U(\tau_h) := J_{\sH(W_0)}$. 
\end{thm}
  
\begin{proof}
By (Reg), Theorem~\ref{thm:ee} and Remark~\ref{rmk:ext}, 
  $U$ extends to an antiunitary representation $\tilde U$ of~$G_{\tau_h}$
  such that $\sH(W_0)$ is the standard subspace
  $\Fix(J e^{\pi i \cdot \partial U(h)})$, where
  $J=\tilde U({\tau_h})$. 
Then  $\sH(W_0)$ is the standard subspace associated to
  $(h,\tau_h)$ by the BGL construction
  \begin{equation}
    \label{eq:BGL}
 \sH _{\tilde U}(x,\sigma) = \Fix(\tilde U(\sigma)
e^{\pi i \cdot \partial U(x)})
  \end{equation}
from \eqref{eq:bgl}. 
  The covariance of the nets $\sH$ and $\sH _{\tilde U}$ thus implies 
  that $\sH=\sH _{\tilde U}$ holds on all of $\cW_+ = G.W_0$.
Therefore \cite[Thm.~4.12]{MN21} applies. 
  As $h$ is symmetric, (SS), (CTL) and (CTHD)
  follow from \cite[Prop.~4.16]{MN21}, where these
  properties are contained in (HK4) and (HK6). 
\end{proof}

\begin{remark}\label{rmk:ss1}
  (a) By Remark \ref{rmk:hermgen}, if $\fg=\Lie (G)$ is a hermitian simple Lie algebra, the cones $C_{\pm}$ have inner points in $\fg_{\pm1}$.
In addition, every Euler element $h \in \g$ is symmetric by
    \cite[Prop.~3.11(b)]{MN21}.
    If $\sH$ is a net of standard subspaces satisfying the hypotheses of Theorem \ref{thm:1BW}, then conditions (BW) and (CTHD) are met, and Theorem \ref{thm:reg}(a) ensures that $\sH$ also satisfies (Reg). Consequently, all the assumptions of Theorem \ref{thm:spst1} are fulfilled, and conditions
(HK7), (HK8) as well as (SS) are satisfied. 

\nin (b) The assumption of the symmetry of $h$ is crucial in
    the Theorem~\ref{thm:spst1}. If $h$ is not symmetric,
    then $W^{'\alpha} \not\in \cW_+$ for all $\alpha \in Z(G)^-$,
    so that the assumptions of (CTL), (SS) and (HK6-8) are never satisfied
    and the statement becomes vacuous. 

   \nin  (c)  Given an antiunitary representation $U$
    of the group $G_{\tau_h}$, then the BGL net   \eqref{eq:BGL} 
    is regular for a large class of Lie groups~$G$.
  Indeed, if $G$ is a semidirect product $G=N\rtimes L$ with $L$ reductive,
  then regularity holds if $C_U \cap \fn_{\pm1}(h)$ is generating
  in $\fn_{\pm1}=\Lie(N)\cap \fg_{\pm1}$  (\cite[Thms.~4.11, Cor.~4.24]{MN24}).
  See also \cite{BN25} and \cite[\S 5.3]{Ne26}. 
  for more results in this direction. 
  
\nin  (d)  Let $(G, \cW_+, U, \sH)$ be a net of standard subspaces such that
(HK1)-(HK8), (Reg) and (SS) are satisfied. We then have
$Z_\alpha^2 = U(\alpha)$. Here $Z_\alpha$ is by far not unique.
If, for instance $T \in \cZ(U(G_{\tau_h})')$ is a unitary involution,
then $Z_\alpha' := Z_\alpha T \in U(G)'$ also satisfies
(CTL), (SS), (CTHD): 
\[  J Z_\alpha' J = Z_\alpha^{-1} T
  = Z_\alpha^{-1} T^{-1}
  = T^{-1} Z_\alpha^{-1} 
  = (Z_\alpha')^{-1} \]
and $(Z_\alpha')^2 = Z_\alpha^2 T^2 = U(\alpha)$.

If the unitary representation $U\res_G$
is irreducible, then
\[ U(G)' \cap \U(\cH) = \T \1 \quad \mbox{ and } \quad
  U(G_{\tau_h})' \cap \U(\cH)= \{ \pm\1\},\]
so that the only ambiguity
in the choice of $Z_\alpha$ is that we may replace $Z_\alpha$ by $-Z_\alpha$.
\end{remark}

 \begin{rem}\label{rmk:quot}
 Note that, for a quotient group 
$\underline G := G/\Gamma$, $\Gamma \subeq Z(G)$ discrete, 
   the complementary abstract wedge $(W_0^{\underline G})' = (-h,\tau_h)$
   may be contained in $\cW_+^{\underline G}$, even if $W_0'\notin\cW_+^G$. 

   As an example, consider $G=\SL_2(\R)$ and
   $\Gamma =Z(G)=\{\pm\textbf{1}\}$.
   In view of \cite[Ex.~2.10(e)]{MN21},
   $W^{'-\textbf{1}}=(-h,-\tau_h)\in \cW_+^G$ but $W'=(-h,\tau_h)\notin\cW_+^G$.
   On the other hand, for $\underline G=\PSL_2(\R)$, 
   we have $(-h, \tau_h) \in \cW_+^{\uline G}$.
   In this case the quotient map $q \colon G \to \uline G$ induces a
   bijection $q^\cW \colon \cW_+^G \to \cW_+^{\uline G}$
   of homogeneous $G$-spaces. Here $-\1 \in G$ acts trivially
   on $\cW_+^G$, so that the $G$-action factors through an action of
   $\uline G$.
  \end{rem}

\begin{cor} \mlabel{cor:3.16} Assume that the cones $C_\pm$ span
  $\g_{\pm 1}$. Let $(U,\cH)$ be a unitary representation of $G$ with discrete kernel and $\sH:\cW_+\to \Stand(\cH)$ be a net of standard subspaces satisfying
  {\rm(HK1)-(HK4)}. Then $\sH$ is the BGL net $\sH_{\tilde U}$,
  associated   to the antiunitary representation
  $\tilde U$ of $G_{\tau_h}$, extending $U$ by
  $\tilde U(\tau_h)=J_{\sH(W_0)}$. 
\end{cor}

\begin{prf}    As Theorems~\ref{thm:1BW} and \ref{thm:reg}(a) apply,
    the assumptions of Theorem~\ref{thm:spst1} hold, so that
    (HK5)-(HK8), (SS) and {(Reg)} are satisfied.
    This implies the corollary. 
  \end{prf}

\begin{rem} In the context of Corollary~\ref{cor:3.16}, 
  it follows with Remark~\ref{ex:quotientBGL},
that, for $\uline G := G/\ker(U)$, the net of standard subspaces
$(\uline G, \cW_+^{\uline G}, \uline U, \uline \sH)$ is the BGL net of the
antiunitary representation $\uline{\tilde U}$ of $\uline{G}_{\tau_h}
\cong G_{\tau_h}/\ker(U)$.
In particular $\uline\sH$ satisfies (HK1)-(HK8), (Reg) and~(SS).
\end{rem}

\begin{ex}  
  Let $G=\tilde\PSL_2(\RR)\times\tilde\PSL_2(\RR)
  {(=\widetilde\Mob\times \widetilde{\Mob})}$, where
  $\tilde \PSL_2(\R)$   is the universal covering of $\PSL_2(\R)$.
  Its Lie algebra $\fg=\fsl_2(\RR)\oplus\fsl_2(\RR)$ is semisimple.
  Let $C\oplus C \subeq \g$ be a $G$-invariant cone,
  where $C$ is a  non-trivial $\Ad(\SL_2(\R))$-invariant cone  in $\fsl_2(\RR)$. Let $h=h_1\oplus h_2\in\fg$ be a symmetric Euler element,
  where $h_1, h_2 \not=0$,  let $G_{\tau_h}$ be the corresponding
  extension, and $W_0=(h,\tau_h)\in\cW_+$.
  Let $U$ be a unitary representation of $G$
  with discrete kernel satisfying the spectral condition (SC).
  Let $\sH:\cW_+=G.W_0\to \Stand(\cH)$
  be a net of standard subspaces satisfying (HK1)-(HK4).
  As $C_{\pm}$ generate $\g_{\pm1}$, Theorems~\ref{thm:1BW} and \ref{thm:reg}
  show through Theorem~\ref{thm:spst1} that 
  (HK5)-(HK8), (SS) and (Reg) hold as well.

  In particular $\sH$ is the BGL net associated to the antiunitary representation $\tilde U$ of $G_{\tau_h}$ extending $U$ by $U(\tau_h)=J_{\sH(W_0)}$ and
  $\Gamma := \ker U$ is invariant under $ \tau_h$. Taking
Remark~\ref{ex:quotientBGL}(b)  into account, we define
  $\uline{\sH}$ as the BGL net on $\cW_+^{\uline G}$ associated to $\tilde {\uline U}$, which  satisfies (HK1)-(HK8), (Reg) and (SS). 
   
  This completely classifies central twisted local, isotonous nets
  of standard subspaces on the abstract wedge space $\cW_+$,
  which are covariant for a faithful positive energy
  representation~$U$ of some covering $\underline G$ of
  the conformal group
  $\PSO_{2,2}(\R) \cong \PSL_2(\RR) \times  \PSL_2(\RR) $ of $\R^{1,1}$:
  \textit{every such net is the BGL net of the associated covariant representation $U$ extended to $\underline{G}_{\tau_h}$ by $U(\tau_h)=J_{\sH(W_0)}$.}
\end{ex}

The following theorem shows that any two centrally
  complementary wedges $W^{'\alpha}$ and $W^{'\beta}$ in $\cW_+$
  can be transformed into each other by elementary steps, 
  related to passing from one orthogonal pair $(h,k)$ to another
  one $(h,k')$.

\begin{theorem}\label{thm:slspsp1} {\rm(Twist Generation Theorem)} 
  Let $(G,\cW_+,U,\sH) $  be a regular net of
  standard subspaces with the
  Bisognano--Wichmann property~{\rm(BW)}, for which $U$ has discrete kernel.
  Let $W_0=(h,\tau_h)$ and $W^{'\alpha}_0=(-h,\alpha\tau_h)\in \cW_+=G.W_0$.

   Then there exist Euler elements $k_1, \ldots, k_n$,
   orthogonal to $h$,  such that $\alpha=\zeta_1 \cdots\zeta_n$, where
   $\zeta_j :=\zeta_{h,k^j} \in Z_3(G)\subset Z(G)$, 
 and \eqref{eq:tcl2} holds in the form 
   \begin{equation}\label{eq:sltc}
     \sH(W^{'\alpha}_0) = Z_{\zeta_1}\cdots Z_{\zeta_n} \sH(W_0)'
     \quad \mbox{ with }  \quad Z_{\zeta_j}^2=U(\zeta_j). 
   \end{equation}
   By covariance, an analogous relation holds for every couple
$W^{'\alpha}$ and $W$  of centrally complemented wedges in $\cW_+$.
\end{theorem}

\begin{proof}
  The net $\sH$ satisfies the assumptions of Theorem \ref{thm:spst1}, in particular (SS) holds.
  By Theorem~\ref{cor:z3gen}, 
  there exist Euler elements $k_j$, orthogonal to $h$,
  such that $\alpha = \zeta_{h,k_1} \cdots \zeta_{h,k_n}$.
  The pairs $h, k_j$ generate subalgebras $\fs_j \cong \fsl_2(\R)$,
  and $\zeta_{h,k_j} \in Z(S_j)$ for the corresponding
  integral subgroups $S_j \subeq G$,
  so that \eqref{eq:sltc} follows from (SS).
\end{proof}

\section{Nets of standard subspaces on causal spacetimes}
\label{sect:appssn}

We now describe some applications of the geometric Spin--Statistics
and the Bisognano--Wichmann Theorems to nets of standard subspaces in
$1$+$2$-dimensional Minkowski space and in $1$+$1$ dimensions
with conformal symmetries.
We refer to Appendix \ref{subsec:4.2} for fundamental facts on the conformal group of the Minkowski spacetime $\RR^{1,d-1}$, the Euler elements of the conformal Lie algebra, and the massless scalar  antiunitary representation $U$ of the proper Poincar\'e group $\cP_+$.

\subsection{On 1+2-dimensional Minkowski space}
\label{sect:12ssconf}

As an  example we now explain how our abstract Spin--Statistics Theorem applies to conformal nets on $1$+$2$-dimensional
Minkowski space in the abstract Euler wedge setting. The AQFT model is treated in \cite{GL95}. 
Conformal nets on higher dimensional Minkowski spaces will be discussed in \cite{MN25}.
Let $\R^{1,2}$ be the $d := 1+2$-dimensional Minkowski
spacetime. 
Taking Appendix~\ref{subsec:4.2} 
into account, let $\tilde U$ be an antiunitary representation of $G_{\tau_h}$, where $G$ is the double covering of $\SO_{2,3}(\RR)_e$, extending the massless scalar antiunitary representation $U$ of  the proper Poincar\'e group $\cP_+ \cong ({\cP_+^\up})_{\tau_h}$,  where  $h\in\Lie(\cP_+^\up)$ is the generator of the one-parameter group of boosts
associated to the Rindler wedge $W_R=\{x\in\RR^{1,2}:|x_0|<x_1\}$. Note that
$G_{\tau_h}$ is defined by Remark~\ref{rmk:ext}, 
and $\tilde U$ satisfies the
spectrum condition with respect to the pointed cone~$C$,
generated by the orbit $\Ad(G)\be_0\subset\fg$, where
$\be_0\in \Lie (\cP_+^\up)$ is a generator of time-translations.
Let $W_h=(h,\tau_h)
\in \cW_+^{\cP_+^\uparrow}$ be the Euler wedge associated to $W_R$ as in Remark \ref{ex:mink}. 
The spatial $\pi$-rotation in $\cP_+^\up$ acts as $r(\pi).W_h=W_h'$
(in accordance with $r(\pi)W_R=W_R'$). In particular $W'_h\in\cW_+ $.
 Let
 \begin{equation}
   \label{eq:M2}
q_M \colon M^{(2)} := \bS^1 \times \bS^2 \to M
 \end{equation}
 be the  double covering of the compactified Minkowski spacetime $M$,
 and
{ \[ \iota: \RR^{1,2}\to M, \quad
   \iota(x) =
   \Big[ \frac{1 - \beta(x,x)}{2} : x_0 : x_1: x_2  : - \frac{1 + \beta(x,x)}{2}\Big]   \in \bP(\R^{2,3}), \]
 the canonical embedding, where
 $\beta(x,y) = x_0y_0 - \bx \by$ is the Lorentzian form
 on $\R^{1,2}$ (see e.g. \cite{GL96} or \cite{MN25}). Note that this
 map naturally lifts to
 \[ \iota: \RR^{1,2}\to M^{(2)}, \quad
   \iota^{(2)}(x) =
   \Big(\Big(\frac{1 - \beta(x,x)}{2\delta(x)},\frac{x_0}{\delta(x)}\Big),
   \Big(\frac{x_1}{\delta(x)}, \frac{x_2}{\delta(x)},
   - \frac{1 + \beta(x,x)}{2\delta(x)}\Big)\Big) \in
   \bS^1 \times \bS^2,\]
 where
 \[ \delta(x) := \Big(x_0^2 + \Big(\frac{1 - \beta(x,x)}{2}\Big)^2\Big)^{1/2} 
   = \Big(x_1^2 + x_2^2 + \Big(\frac{1 + \beta(x,x)}{2}\Big)^2\Big)^{1/2}.\]
}

We write 
\[ q\colon G\to \SO_{2,3}(\RR)_e \] 
for the quotient map and
identify $\Lie(\cP_+^\up)$ with a Lie subalgebra of the
conformal Lie algebra $\so_{2,3}(\R)$. 
Let $T = \R \be_0\subset \RR^{1,2}$ be the set of points of the time axes and
\[ \oline T := \overline{\iota(T)} \subeq M \] 
be the compactified timeline. 
For the subgroup $\Mob\subset \SO_{2,3}(\RR)$
  of transformations fixing $\oline T$,
the inverse image 
\[ \Mob^{(2)}=q^{-1}(\Mob) \]  is the connected double covering of
$\Mob$, cf.~Appendix \ref{subsec:4.2}.
Indeed, write $(\rho(\theta))_{\theta \in \R}$
for the one-parameter group of rotations in $\Mob$
and $\tilde\rho$ for the one-parameter group of $G$ lifting $\rho$, then
\[ Z(G)=\{e,\alpha\}\simeq \ZZ_2, \quad \mbox{  where } \quad
  \alpha=\tilde \rho(2\pi) \]
is the $2\pi$-rotation in the $2$-fold covering $\Mob^{(2)}\cong \SL_2(\R)$.
Since $U$, hence also  $\tilde U$, is 
irreducible  and $\alpha\in Z(G)^-$ is an involution,
$\tilde U(\alpha)=-\1$.

Let $\sH_U$ and $\sH_{\tilde U}$ be the $\BGL$ nets associated respectively to $U$ and $\tilde U$, see \eqref{eq:BGL}.
Now $\cW_+^{\cP_+^\up}\subset\cW_+^G$,  and the BGL net 
$\sH_U$ extends to $\sH_{\tilde U}$
(Proposition \ref{prop:BGL-natural} and Remark~\ref{rmk:ext}(b)).
In particular, for every wedge $W\in \cW_+^{\cP_+^\up}$, we have
$W'\in  \cW_+^{\cP_+^\up}$ (see Remark \ref{ex:mink}) and by equivariance,
$W\in \cW_+^G$ also implies $W'\in\cW_+^G$.
By $Z_2(G) = Z_3(G)$ and \eqref{eq:lem:5.2}, 
both complements  $W'$ and $W^{'\alpha}$ of $W_h$, are contained in $\cW_+^G$
because
\[ Z(G)=Z(G)^- =Z_2(G)=Z_3(G)\cong \Z_2.\]
By the BGL construction, we have the twisted duality relations:
\begin{equation}
  \label{eq:twisteddualityrel}
  \sH_{\tilde U}( W_h^{'\alpha})=i\sH_{\tilde U}( W_h)',\qquad   \sH_{\tilde U}( W_h')=\sH_{\tilde U}( W_h)',
\end{equation}
because
\[ \Delta_{\sH_{\tilde U}( W_h^\alpha)}=\exp(2\pi i\cdot \partial U(h))\quad\text{and}\quad J_{\sH_{\tilde U}( W_h^\alpha)}
= U(\alpha)J_{\sH_{\tilde U}( W_h)} =-J_{\sH_{\tilde U}( W_h)}. \]
We may thus put $Z_{\alpha}=i\textbf{1}$
(but also $Z_\alpha =-i$ realizes the relation
  $Z_\alpha^2=-\textbf{1}$) and $Z_{\textbf{1}}=\textbf{1}$.

We will postpone a more explicit correspondence of abstract Euler wedges and their locality property with regions of the
$2$-fold covering space $M^{(2)}$. Just remark that, given
an abstract wedge $W=(h,\tau)$,
there exist $g_1$ and $g_2$ in $G$ such that $g_2 W = W'$, and
$g_1g_2W=g_1W'=W^{\alpha}$, where $W^\alpha=(h,\alpha \tau)$.
The geometric action of $g_1g_2$ 
on $M^{(2)}$ corresponds to the ``helicoidal shift'' in \cite{GL96}.
We will describe conformal theories in greater
generality in the forthcoming paper \cite{MN25}.

The Euler elements $h = h_{1}$ and $h_{2}$ associated to $W_R$ and
\[ W_2=\{x\in\RR^{1,2}: |x_0|<x_2\},\] respectively (cf.\ \eqref{eq:Ws}),
are orthogonal,   and the Lie algebra generated by
$h_{1}$ and $h_{2}$ is
  $\L(\cL_+^\up)\cong \fsl_2(\RR)$.
  For the restriction $U|_{\cL_+^\up}$, the operator
  $i\partial U([h_1,h_2])$ is not semibounded by \cite[Prop.~1.6]{GL95}.
  As a consequence, $[h_1,h_2]\notin -C\cup C$, and the orthogonal pair
  $(h_1, h_2)$ of Euler elements is spacelike  (cf.~Section \ref{sect:stort}).
  Indeed, the orthogonal spacelike pairs $(h,\pm h_2)$ are
$G$-conjugate to the spacelike pairs $(h_{-1,2},\pm h_{-1,1})$ of Appendix \ref{subsec:4.2}.

  We consider the embedding $\Mob\into \SO_{2,3}(\RR)$ as the
  subgroup fixing the ``time circle''~$\oline T$.
  The corresponding Lie subalgebra $\mathfrak{mob} \subeq \so_{2,3}(\R)$
  is isomorphic to $\fsl_2(\RR)$. With the notation from
  Example \ref{ex:sl2}, the one-parameter group $\exp(th_0)$, 
 corresponding to $\exp(t h_{-1,3})$ of Appendix \ref{subsec:4.2},  
acts  by dilations on~$\RR^{1,2}$: 
\[ \delta(t)x:=\exp(th_0)x=e^{t}x \quad \mbox{ for } \quad t\in \R, x\in \RR^{1,2}. \] 
Here $h_0\in\cO_h$ is an Euler element in $\cE(\fg)$ with the positivity domain
\[ V_+=\{x\in\R^{1,2}:x^2 = x_0^2 - \bx^2>0, x_0>0\}.\] 
The Euler element $k_0\in\fsl_2(\RR)$ corresponds to the conformal
vector field with positivity domain
\[ O=\Big\{x\in\RR^{1,2}: |x_0|+\sqrt{ x_1^2+x_2^2}<1\Big\}.\] 
One can describe the geometric action of $\exp(tk_0)$
(corresponding to $\exp(th_{0,3})$ of Appendix~\ref{subsec:4.2})
following \cite{HL82}: The {action of the subgroup $\Mob \subeq \SO_{2,3}(\R)$ on $\R^{2,3}$} fixes $\be_1$ and $\be_2$,
  and acts effectively on the span of $\be_{-1}, \be_0$ and $\be_3$, 
hence commutes with the space rotation in the $\be_1$-$\be_2$-plane  
(see \cite[Sect.~8.2]{BDL07}).
 In particular, $\exp(\R k_0)$
commutes with the space rotations and it is enough to
describe it on the $x_0$-$x_1$ plane in
{$\R^{1,2} \cong \iota^{(2)}(\R^{1,2})$,
  which is mapped by $\iota^{(2)}$ into the torus
  $\bS^1 \times \bS^1 \subeq \bS^1 \times \bS^2$ specified by the equation
  $x_2 = 0$.}
In the coordinates $u_\pm=x_0\pm x_1$, we have
\begin{align*} 
\exp(t k_0)(u_+,u_-)
  &=\left(\frac{e^{t}(1+u_+) + u_+-1}{e^t(1+u_+)+1 - u_+} ,
\frac{e^{t}(1+u_-) + u_--1}{e^t(1+u_-)+1 - u_-}\right),\\
&=\left(\frac{\cosh(t/2)u_++\sinh(t/2)}{\sinh(t/2)u_++\cosh(t/2)} ,
\frac{\cosh(t/2)u_-+\sinh(t/2)}{\sinh(t/2)u_-+\cosh(t/2)}\right),
\qquad t\in\RR.
  \end{align*}
The formula for the chiral coordinate also appear in \cite{LM24}.
The elements $h_0$ and $k_0$ generate
$\mathfrak{mob}\subset \so_{2,3}(\RR)$,
and they are orthogonal Euler elements.
Note that the restriction $U|_{\Mob^{(2)}}$ is of positive energy
  and that the pair $(h_0,k_0)$ is positive timelike,
  i.e., $[h_0,k_0]\in C$ corresponds to a positive operator.
Likewise $(-h_0,k_0)$ is negative timelike.

\subsection{On two-dimensional conformal nets}\label{sect:2dimss}

In this section we apply our
Spin--Statistics result to two-dimensional conformal models. We refer to Example \ref{ex:mob} for the correspondence between Euler wedges and admissible intervals of the real line, as the universal covering of the circle $\bS^1$.

Let $H=\cP^\up_+ = \R^{2}\rtimes \SO_{1,1}(\R)_e$ be the
identity component of the $2d$-Poincar\'e group, 
and its $\ZZ_2$-extension $H_{\tau_h} \cong  \cP_+$.
The Euler wedge $(h,\tau_h)$
determines the right wedge region $W_R$ as in Example \ref{ex:mink}.

Consider the group inclusion $H\into \tilde G$,
where $\tilde G=\tilde\Mob\times\tilde\Mob$ is the universal covering of
$G=\Mob\times\Mob$, hence $Z(\tilde G)= \tilde\rho(2\pi \Z)\times\tilde\rho(2\pi \Z)\simeq \ZZ\times \ZZ$,
where $\tilde\rho \colon \R \to \tilde G$ is the lift
of the rotation group in $\SL_2(\R)$. 
We identify the central elements $(\tilde\rho(2\pi),e)$ and
$(e,\tilde\rho(2\pi))$ with $(1,0)$ and $(0,1)$ in $\Z^2$,
respectively. Whenever convenient, we shall use this
identification to describe elements of 
the center~$Z(\tilde G)$.
The Euler element in $\fh$ corresponds to the Euler element
$h = (h_0, -h_0) \in \g \cong \fsl_2(\R)^{\oplus 2}$ (cf.~\eqref{eq:h0k0-intro}). 
The corresponding involution on
$\tilde G$ is $\tau_h =\tau_{h_0}\times \tau_{h_0}$.

The group $G$ acts by conformal diffeomorphisms
on compactified $2d$-Minkowski space, which in light ray
coordinates $u_{\pm} = x_0 \pm x_1$ 
is
\[ M^{1,1} = \bS^1 \times \bS^1,\] 
where the two factors of $G$ act by M\"obius transformations. 
Accordingly, the group $\tilde G$ acts naturally on the
  {\it Einstein universe}
  \[ \tilde M^{1,1} \cong \R \times \R,\]
  which is also called superworld (see \cite{BGL93,LM75}).
  We shall denote by $(\tilde u_+, \tilde u_-)$
  the natural chiral coordinates of $\tilde M$.
We identify $H$ with the integral subgroup of
$\tilde G$
with Lie algebra
\[ \fh =  (\R(e_0,0) + \R(0,e_0)) \rtimes \R h,\]
where $[h_0, e_0] = e_0$, and $\exp(te_0)$ acts by translations
on $\R \subeq \R_\infty \cong \bS^1$. Then
Minkowski space embeds as the open subset
\begin{equation}
  \label{eq:M11-emb}
M_0:= H.(0,0) = (-\pi, \pi) \times (-\pi, \pi) \subeq \tilde M^{1,1}
\end{equation}
(cf.\ Example~\eqref{ex:mob}(a)). 
We also have an inclusion $H_{\tau_h} \into \tilde G_{\tau_h}$,
which leads to a natural inclusion of abstract wedge spaces
\[ \cW_+^H \into \cW_+^{\tilde G}.\]

In $\tilde M^{1,1}$, the wedge regions corresponding to
the Euler element $h =  (h_0, -h_0)$ in $\tilde M^{1,1}$ are of the form
\[ \tilde I_{2n} \times \tilde I_{2m+1}, \quad n,m \in \Z\] 
(cf.~\eqref{eq:I2n} in Example~\eqref{ex:mob}(b)). Accordingly,
the wedge regions in $\tilde M^{1,1}$ can be identified with
$\tilde J_1 \times \tilde J_2 \subeq \R^2$, where the length
of both intervals is smaller than $2\pi$.
\begin{figure}[ht]\centering
\begin{tikzpicture}[scale=0.75]
        \draw [->] (-3,0) --(5,0);
         \draw [->] (0,-3)--(0,4);
          \draw [thick] (-2,0)-- (0,2) node [ left] {$M$};
          \draw [thick] (0,2)-- (2,0);
          \draw [thick] (-2,0)-- (0,-2);
          \draw [thick] (2,0)-- (0,-2);
          
           \draw [ thick] (0,0)-- (1,1);
          \draw [thick] (0,0)-- (1,-1);
             \draw [ thick] (0,0)-- (-1,1);
          \draw [thick] (0,0)-- (-1,-1);

          \draw [thin,->] (-3,-3) -- (3.5,3.5) node [above right] {$\tilde u_+$};
          \draw [thin,->] (3,-3) -- (-3.5,3.5) node [above left] {$\tilde u_-$};
  \fill [color=black,opacity=0.2]
               (0,0) -- (-1,1) -- (-2,0) --(-1,-1);
           \fill [color=black,opacity=0.6]
              (0,0) -- (1,1) -- (2,0) -- (1,-1);

\end{tikzpicture}
\caption{The Minkowski space $M_0$ in $\tilde M^{1,1}=\RR^2$. The dark grey region is the Rindler wedge~$W_R$, and the light grey region is the left wedge $W_L$.} 
\label{fig:1}
\end{figure}
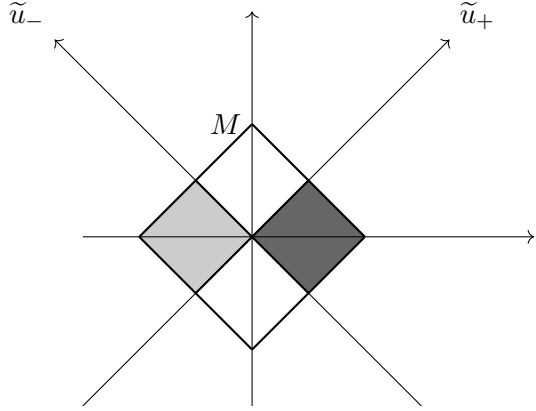

The center $Z(\tilde G)\cong \Z^2$ is generated by
the two elements $(\zeta_0,e)$ and $(e,\zeta_0)$,
where $\zeta_0 = \exp(2\pi z_0) = \tilde\rho(2\pi)$, with $z_0$ from \eqref{eq:z0}, 
corresponds to a $2\pi$-rotation in the chiral coordinates. 
According to the embedding \eqref{eq:M11-emb} of $\R^{1,1}$ into
$\tilde M_{1,1}$, the right Rindler wedge $W_R \subeq
\R^{1,1}$ corresponds to the square 
\[ W_R  = (0,\pi) \times (-\pi,0),\]
and its causal dual $W_L = -W_R$ corresponds to the square
\[ W_L  = (-\pi,0) \times (0,\pi).\]
We may thus identify the homogeneous spaces  $\cW_+^{\tilde G}$
with the orbit of $W_R$ for the $\tilde G$-action on $\tilde M^{1,1}
\cong \R^2$.
The spacelike rotations
\begin{equation} \label{eq:rs}
   r_s(\theta) := (\tilde \rho(\theta), \tilde \rho(-\theta))
   = (\exp(\theta z_0), \exp(-\theta z_0))
\end{equation} 
then satisfies
\begin{equation}
  \label{eq:rswr}
  r_s(-\pi) W_R = W_L.
\end{equation}
On  the level of abstract wedges,
\begin{equation}\label{eq:EWrswr}
 r_s(-\pi).(h,\tau_h) = (-h, r_s(-\pi) \tau_h r_s(\pi))
  = (-h, r_s(-2\pi) \tau_h) = (-h, (\zeta_0^{-1}, \zeta_0) \tau_h). 
 \end{equation}
Every other central complement of $W_R$ is of the form
$z.W_L$ for some $z \in Z(\tilde G)$, and
on the level of abstract wedges
\[ (h,\tau_h)^{'\alpha} = (-h, \alpha (\zeta_0^{-1}, \zeta_0) \tau_h), \quad
  \alpha \in Z_2(\tilde G) = 2 Z(\tilde G).\]
Let $\tilde U$ be an antiunitary positive energy representation of
{ $\tilde G_{\tau_h}$}, with discrete kernel. Let
$\cW_+^{\tilde G} = \tilde G.(h,\tau_h)$ be the 
orbit of Euler wedges and $\sH$ be a net of standard subspaces on
$\cW_+^{\tilde G}$, satisfying (HK1)-(HK4). Then 
(Reg) follows from Theorem \ref{thm:reg}
  and  (BW), and (CTHD) from Theorem~\ref{thm:1BW}.
Further, (HK5)-(HK8) and  (SS) hold by 
    Theorem \ref{thm:spst1}, which entails in particular
    that $\sH=\sH_{\tilde{U}}$ is the BGL net of $\tilde U$.  
    We thus  obtain the following result:    

\begin{proposition}\label{prop:iden}
  Let $(\sH(W))_{W\in\cW_+^H}$ be  a Poincar\'e covariant
  net of standard subspace in
  1+1 dimensions, satisfying {\rm(HK1)-(HK3)}, 
  that extends 
 to a $\tilde G$-covariant net of standard subspaces  satisfying {\rm(HK1)-(HK4)}.
  Then the extension $(\sH(W))_{W\in\cW_+^{\tilde G}}$ is the BGL net of standard subspaces of the representation $\tilde U$ of $\tilde G_{\tau_h}$.  
\end{proposition}

If the net $(\sH(W))_{W \in \cW_+^H}$
indexed by the space of affine wedges in $\R^{1,1}$
satisfies the assumptions of
Proposition~\ref{prop:iden}, then the extension 
$(\sH(W))_{W\in \cW_+^{\tilde G}}$ 
is the BGL net of standard subspaces of the representation $\tilde U$ of
$\tilde G_{\tau_h}$.  In particular, by the (SS) condition, 
\begin{equation}\label{eq:ss}
Z_\alpha^2=U(\alpha).
\end{equation}

The Spin--Statistics result provides
a correspondence between the symmetry group representation, symplectic relations among complementary wedge subspaces and the Tomita modular operators.

The wedge space $\cW_+^{\tilde G} = \Ad(\tilde G).(h,\tau_h)$,
  for $h = (h_0, h_0)$, is one of four orbits in $\cG_E(\tilde G_{\tau_h})$.
  For $\tilde\SL_2(\R)$, we have two orbits,
  represented by $(\pm h_0, \tau_{h_0})$,  exchanged under the
  complementation map.
  For  $\tilde G$ we thus obtain $4$ orbits, represented by
  \[ ((h_0, h_0), \tau_h), \quad
    ((-h_0, -h_0), \tau_h), \quad 
    ((h_0, -h_0), \tau_h), \quad
    ((-h_0, h_0), \tau_h).\]
  For $h = (h_0, h_0)$, the twisted central complements of the elements of $\cW_+^{\tilde G}$
  lie in the orbit of $W^{'\alpha} = ((-h_0, -h_0), \alpha\tau_h)$ for
  $\alpha = (1,1)$. All other twisted 
  central complements of the elements in $\cW_+^{\tilde G}$
  are obtained by acting with elements of
  the center $Z(\tilde G) \cong \Z \times \Z$,
contained in the coset
  \[ (1,1) + (2\Z \times 2 \Z) = \{(n,m) \in \Z^2 \colon n,m \ \mbox{ odd}\}.\]
  
By the (SS) property, \eqref{eq:ss} provides 
  a relation between $U(\alpha)$ and $Z_\alpha$. In particular,
  given $\sH(W)$ and its modular operators $J_{\sH(W)}$ and $\Delta_{\sH(W)}$, the central complement $W^{'\alpha}$ maps to the standard subspace    
  $\sH(W^{'\alpha})$ with associated modular 
operators
\[ J_{\sH(W^{'\alpha})}=Z_{\alpha}J_{\sH(W)} \quad \mbox{ and } \quad
  \Delta_{\sH(W^{'\alpha})}=\Delta_{\sH(W)}^{-1}.\] 

The unitary representation of the central generators $(1,0)$, $(0,1)$  can be described in terms of modular conjugations as follows.
Consider the wedge regions (see Figure \ref{fig:2})
\[ V_+=(0,\pi)\times(0,\pi)\quad \mbox{ and }  \quad
  S_+=(-\pi/2,\pi/2)\times (0,\pi).\]
Taking \eqref{eq:gengen} into account, we have 
\[ J_{S_+}J_{\sH(V_+)}J_{S_+}J_{\sH(V_+)}
  =\tilde U(\tau_{k_0} \tau_{h_0}\tau_{k_0} \tau_{h_0},e)
\ {\buildrel \eqref{eq:gengen}\over = } \  \tilde U(1,0).\]
Analogously one can check that, for  $V_+=(0,\pi)\times(0,\pi)$ and $S_-=(0,\pi)\times(-\pi/2,\pi/2)$, we have $J_{S_-}J_{\sH(V_+)}J_{S_-}J_{\sH(V_+)}=U(0,1)$ (cf. \cite{MT19}). 

\begin{remark} In Proposition \ref{prop:iden}, the assumptions on the net
  $(\sH(W))_{W\in \cW_+^H}$ do not include (CTL), namely (HK4). This is because there are no symmetric Euler elements in $\cE(\fh)$, hence $W'$ is never  contained in the
  $H$-orbit of $W$, for every $W\in\cW_+^H$.
  On the other hand, every Euler element in $\fg$ is symmetric, so that
  every wedge $W\in \cW^{\tilde {G}}_+$ has central complements in the $\tilde G$-orbit. 
  In particular, once $\cW_+^H$ is identified as a subset of $\cW_+^{\tilde{G}}$,
  every wedge $W\in \cW_+^H$ has central complements in
  $\cW_+^{\tilde{G}}$, as in \eqref{eq:rswr}.
 \end{remark}

\begin{figure}[ht]
\centering
\begin{minipage}{0.48\textwidth}
\centering
\begin{tikzpicture}[scale=0.75]
        \draw [->] (-3,0) --(5,0);
        \draw [->] (0,-3)--(0,4);
        \draw [thick] (-2,0)-- (0,2) node [left] {$V_+$};
        \draw [thick] (0,2)-- (2,0);
        \draw [thick] (-2,0)-- (0,-2);
        \draw [thick] (2,0)-- (0,-2);

        \draw [thick] (0,0)-- (1,1);
        \draw [thick] (0,0)-- (-1,1);

        \draw [thin,->] (-3,-3) -- (3.5,3.5) node [above right] {$\tilde u_+$};
        \draw [thin,->] (3,-3) -- (-3.5,3.5) node [above left] {$\tilde u_-$};

        \fill [color=black,opacity=0.2]
               (0,0) -- (1,1) -- (0,2) --(-1,1);
\end{tikzpicture}
\end{minipage}
\hfill
\begin{minipage}{0.48\textwidth}
\centering
\begin{tikzpicture}[scale=0.75]
        \draw [->] (-3,0) --(5,0);
        \draw [->] (0,-3)--(0,4);
        \draw [thick] (-2,0)-- (0,2) node [left] {$S_+$};
        \draw [thick] (0,2)-- (2,0);
        \draw [thick] (-2,0)-- (0,-2);
        \draw [thick] (2,0)-- (0,-2);

        \draw [thick] (0.5,0.5)-- (-0.5,1.5);
        \draw [thick] (-0.5,-0.5)-- (-1.5,0.5);
        \draw [thick] (-0.5,-0.5)-- (0.5,0.5);

        \draw [thin,->] (-3,-3) -- (3.5,3.5) node [above right] {$\tilde u_+$};
        \draw [thin,->] (3,-3) -- (-3.5,3.5) node [above left] {$\tilde u_-$};

        \fill [color=black,opacity=0.6]
               (0.5,0.5) -- (-0.5,1.5) -- (-1.5,0.5) --(-0.5,-0.5);
\end{tikzpicture}
\end{minipage}

\caption{The light grey region is the the forward light cone $V_+$ and the dark grey region is the strip $S_+$. Both wedge regions are contained in $M_0$.}
\label{fig:2}
\end{figure}
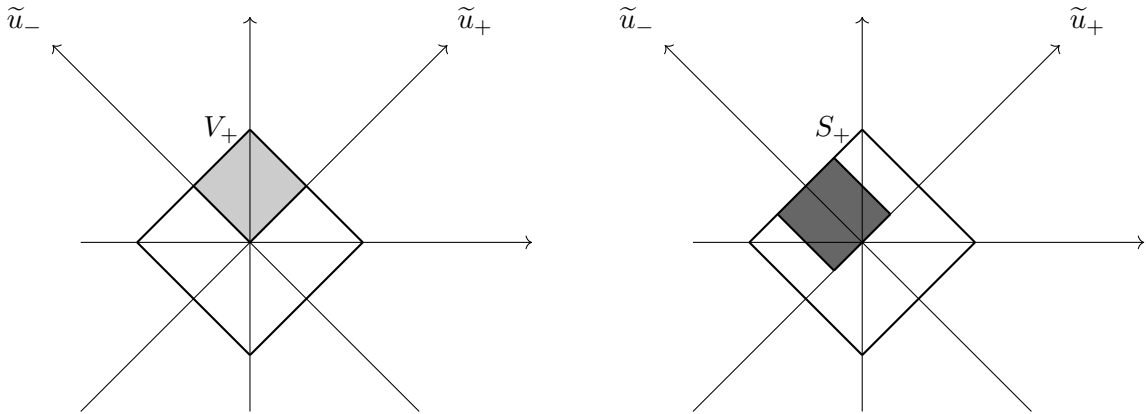

\begin{remark} (a) In the circle $\bS^1 \cong \R_\infty$, Euler elements
  in $\fsl_2(\R)$ correspond to non-dense open intervals
  $I$, which are the positivity regions of the associated vector fields
  on $\R_\infty$, see Example \ref{ex:mob}(a). We write $I'$ for the interior of the complement of $I$.

  Then $h_0$ and $k_0$ correspond to the intervals
  \[ I_{h_0} = (0,\infty) \quad \mbox{ and } \quad
    I_{k_0} = (-1,1).\]
  Their complements are
  \[ I_{h_0}' = (-\infty,0) \quad \mbox{ and }\quad
    I_{k_0}' = (1,\infty) \cup \{ \infty \} \cup
    (-\infty, -1).\]

  Orthogonal pairs of Euler elements are associated to
  ``orthogonal pairs'' of intervals  and therefore conjugate to one
  of the pairs
  \[ (I_{h_0}, I_{k_0}) \quad \mbox{ or } \quad (I_{h_0}, I_{k_0}').\]
  More generally, an interval corresponds to an Euler element
  orthogonal to $h_0$ if and only it is of the form
  $I_a := (-a,a)$ for some $a > 0$, or its complement $I_a'$.
\medskip

\nin (b) We now consider the conformal compactification
  $M^{1,1} \cong \bS^1 \times \bS^1$ of $2$-dimensional Minkowski space
  $\R^{1,1}$, where the product structure corresponds to
  lightray coordinates
  (see Subsection~\ref{sect:12ssconf}). We write
  $\cD_h \subeq M^{1,1}$ for the positivity region
  (diamond) corresponding to the Euler element $h_{-1,2}
  = (h_0, h_0) \in \so_{2,2}(\R)$.
  Then
  \[ \cD_{(h_0,h_0)} = V_+ \]
  is the positive light cone and 
  \[ \cD_{(h_0,-h_0)} = W_R, \quad \cD_{(-h_0,h_0)} = W_L\]
  are the right/left Rindler wedges in $\R^{1,1}$. In lightray coordinates,
  this turns into
  \[ \cD_{(h_0,h_0)} = \R_+ \times \R_+=V_+, \qquad
    \cD_{(h_0,-h_0)} = \R_+ \times \R_- = W_R, \quad
    \cD_{(-h_0,h_0)} = \R_- \times \R_+ = W_L\]
(cf.~Figures \ref{fig:1} and \ref{fig:2}).
  Likewise
  \[ \cD_{(k_0,k_0)} = I_{k_0} \times I_{k_0}, \quad
     \cD_{(k_0,-k_0)} =  I_{k_0} \times I_{k_0}', \quad
     \cD_{(-k_0,k_0)} = I_{k_0}' \times I_{k_0}, \quad 
     \cD_{(-k_0,-k_0)} = I_{k_0}' \times I_{k_0}'\] 
(cf.~Figure~\ref{fig:4}).

 On Minkowski space in chiral coordinates $(u_+,u_-)$
 we have the following geometric actions of the one parameter groups generated by the above Euler elements:
\begin{itemize}
\item $\exp(t(h_0,h_0))(u_+, u_-)=(e^tu_+,e^tu_-)$ 
  is the one parameter group of dilations in $\RR^{1,1}$
\item $\exp(t(-h_0,h_0))(u_+,u_-)=(e^{-t}u_+,e^tu_-)$
  is the one parameter group of boosts having $W_R$ as positivity region
\item $\exp(t(k_0,k_0))(u_+,u_-)
=\left(\frac{\cosh(t/2)u_++\sinh(t/2)}{\sinh(t/2)u_++\cosh(t/2)} ,
\frac{\cosh(t/2)u_-+\sinh(t/2)}{\sinh(t/2)u_-+\cosh(t/2)}\right)$
are the conformal boosts fixing $$O=\{(x_0,x_1)\in \RR^{1,1}:|x_0|+|x_1|<1\}
=\{(u_+,u_-):| u_+|<1, |u_-|<1\}.$$
\end{itemize}

We further note that $((h_0,h_0), (k_0,k_0))$ is a positive timelike
  orthogonal pair of Euler elements generating $\mathfrak{mob}
  \subeq \so_{2,2}(\R)$, and $\tilde U|_{\tilde\Mob}$ satisfies the spectrum condition. Moreover 
  \[ [\pm(h_0,-h_0), (k_0,k_0)]\notin\pm C \quad \mbox{ and } \quad
    [(-h_0,-h_0), (k_0,k_0)]\in -C.\] 
From Lemma~\ref{lem:lemc1} and the subsequent discussion it follows that
  \[ \zeta_{(h_0, h_0), (k_0, k_0)} = (1,0) \quad \mbox{ and } \quad
    \zeta_{(h_0, h_0), (k_0, -k_0)} = (0,1).
  \]
\end{remark}

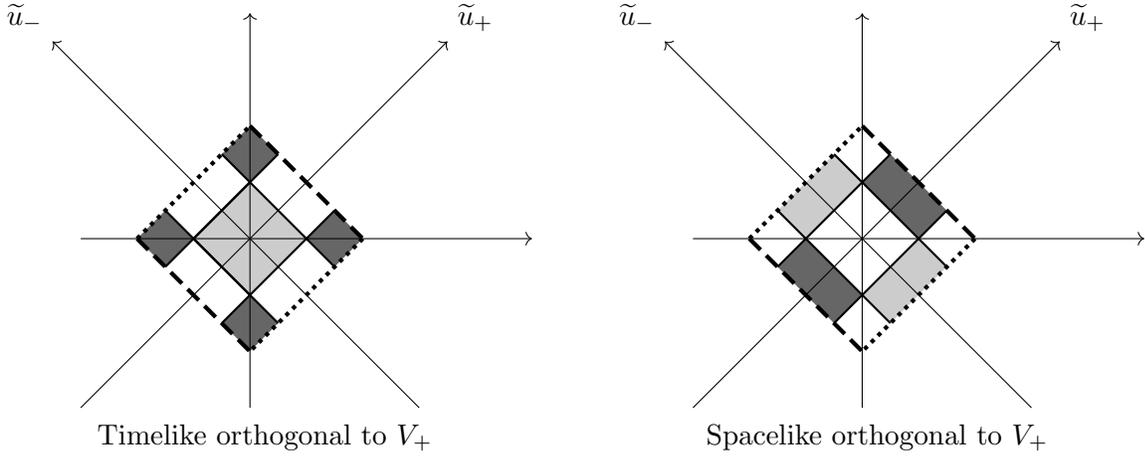
\begin{figure}[ht]
\centering
\begin{minipage}{0.48\textwidth}
\centering
\begin{tikzpicture}[scale=0.75]
        \draw [->] (-3,0) --(5,0);
        \draw [->] (0,-3)--(0,4);
         \draw [ultra thick, dotted]  (-2,0)-- (0,2);
        \draw [ultra thick, dash pattern=on 6pt off 3pt]  (0,2)-- (2,0);
        \draw [ultra thick, dash pattern=on 6pt off 3pt]  (-2,0)-- (0,-2);
        \draw [ultra thick, dotted] (2,0)-- (0,-2);

        \draw [thick] (1,0)-- (0,1);
      \draw [thick] (1,0)-- (0,-1);
        \draw [thick] (0,-1)-- (-1,0);
        \draw [thick] (-1,0)-- (0,1);

        \draw [thick] (1,0)-- (1.5,0.5);
      \draw [thick] (1,0)-- (1.5,-0.5);
        \draw [thick] (-1,0)-- (-1.5,0.5);
      \draw [thick] (-1,0)-- (-1.5,-0.5);      
        
                \draw [thick] (0,1)-- (0.5,1.5);
      \draw [thick] (0,1)-- (-0.5,1.5);
                \draw [thick] (0,-1)-- (0.5,-1.5);
      \draw [thick] (0,-1)-- (-0.5,-1.5);

        \draw [thick] (0,-1)-- (-1,0);
        \draw [thick] (-1,0)-- (0,1);

        \draw [thin,->] (-3,-3) -- (3.5,3.5) node [above right] {$\tilde u_+$};
        \draw [thin,->] (3,-3) -- (-3.5,3.5) node [above left] {$\tilde u_-$};

        \fill [color=black,opacity=0.2]
               (1,0) -- (0,1) -- (-1,0) --(0,-1);
        \fill [color=black,opacity=0.6]
               (1,0) -- (1.5,0.5) -- (2,0) -- (1.5,-0.5);
        \fill [color=black,opacity=0.6]
               (-1,0) -- (-1.5,-0.5) -- (-2,0) -- (-1.5,0.5);
        \fill [color=black,opacity=0.6]
               (0,1) -- (0.5,1.5) -- (0,2) -- (-0.5,1.5);
        \fill [color=black,opacity=0.6]
               (0,-1) -- (-0.5,-1.5) -- (0,-2) -- (0.5,-1.5);

\end{tikzpicture}
Timelike orthogonal to $V_+$
\end{minipage}
\hfill
\begin{minipage}{0.48\textwidth}
\centering
\begin{tikzpicture}[scale=0.75]
        \draw [->] (-3,0) --(5,0);
        \draw [->] (0,-3)--(0,4);
        \draw [ultra thick, dotted]  (-2,0)-- (0,2);
        \draw [ultra thick, dash pattern=on 6pt off 3pt]  (0,2)-- (2,0);
        \draw [ultra thick, dash pattern=on 6pt off 3pt]  (-2,0)-- (0,-2);
        \draw [ultra thick, dotted] (2,0)-- (0,-2);

        \draw [thick] (1,0)-- (1.5,0.5);
        \draw [thick] (0,1)-- (0.5,1.5);
        \draw [thick] (0,1)-- (1,0);

        \draw [thick] (-1,0)-- (-1.5,-0.5);
        \draw [thick] (0,-1)-- (-0.5,-1.5);
        \draw [thick] (0,-1)-- (-1,0);

        \draw [thick] (-1,0)-- (-1.5,0.5);
        \draw [thick] (0,1)-- (-0.5,1.5);
        \draw [thick] (-1,0)-- (0,1);

        \draw [thick] (1,0)-- (1.5,-0.5);
        \draw [thick] (0,-1)-- (0.5,-1.5);
        \draw [thick] (1,0)-- (0,-1);

        \draw [thin,->] (-3,-3) -- (3.5,3.5) node [above right] {$\tilde u_+$};
        \draw [thin,->] (3,-3) -- (-3.5,3.5) node [above left] {$\tilde u_-$};

        \fill [color=black,opacity=0.6]
               (1,0) -- (1.5,0.5) -- (0.5,1.5) --(0,1);
        \fill [color=black,opacity=0.6]
               (-1,0) -- (-1.5,-0.5) -- (-0.5,-1.5) --(0,-1);

        \fill [color=black,opacity=0.2]
               (-1,0) -- (-1.5,0.5) -- (-0.5,1.5) --(0,1);
        \fill [color=black,opacity=0.2]
               (0,-1) -- (0.5,-1.5) -- (1.5,-0.5) --(1,0);

\end{tikzpicture}
Spacelike orthogonal to $V_+$
\end{minipage}

\caption{$M^{1,1}$ is obtained by identifying the dotted and the dashed edges,
  respectively, of $M_0\subset\tilde{M^{1,1}}$.
  In the first picture, the dark gray region is $\cD_{(-k_0,-k_0)}$ and the light grey region is~$\cD_{(k_0,k_0)}$. In the second picture the dark gray region is $\cD_{(-k_0,k_0)}$ and the light grey region is~$\cD_{(k_0,-k_0)}$. }
\label{fig:4}
\end{figure}

\section{Applications to nets of von Neumann algebras}\label{sect:appvnan}

In this section we explain how the results on nets of standard
subspaces derived in the preceding section can be transcribed
to nets of von Neumann algebras. 

\subsection{Results on nets of von Neumann algebras}
\mlabel{subsec:res-vNa}

We start from the definition of a net of von Neumann algebra on the wedge space, and specify some of its properties. 

\begin{defn}\label{def:vnnet}  
Let $G_{\tau_h} = G \rtimes \{e,\tau_h\}$
  and $\cW_+ = G.W_0 =G.(h,{\tau_h}) \subeq \cG_E(G_{\tau_h})$
  be as above
  and $C \subeq \g$ be a closed convex $\Ad^\eps(G_{\tau_h})$-invariant cone.
    We write $\mathrm{vNA}(\cH)$ for the set of von Neumann subalgebras of 
    $B(\cH)$. 
Let $(U,\cH)$ be a unitary representation of $G$  and 
\begin{equation}
  \label{eq:netb}
\cA \colon \cW_+ \to \mathrm{vNA}(\cH) 
\end{equation}
be a map, also called a {\it net of von Neumann algebras}. 

We consider the following properties: 
\begin{itemize}
\item[\rm(HKA1)] {\bf Isotony:} $\cA(W_1) \subeq \cA(W_2)$ for $W_1 \leq W_2$. 
\item[\rm(HKA2)] {\bf Covariance:} $\cA(gW) = U(g)\cA(W)U(g)^{-1}$ for 
$g \in {G}$, $W \in \cW_+$. 
\item[\rm(HKA3)] {\bf Spectral condition:} 
$ C \subeq C_U := \{ x \in \g \colon -i \partial U(x) \geq 0\}.$ 
\item[\rm(Vac)] {\bf Vacuum Property:} The space $\cH^G$
  of $G$-fixed vectors is $1$-dimensional 
  and contains a unit vector $\Omega$ that is cyclic
  and separating for all algebras $\cA(W)$.
  We then write
  \[ \sH(W) := \sH_{\cA(W),\Omega} := \oline{\cA(W)_{sa}.\Omega}\]
  for the corresponding
standard subspace, where, for a $*$-algebra $\cM$, we put
{\[ \cM_{sa} := \{ M \in \cM \colon M^* = M\}.\] }
\item[\rm(HKA4)]{\bf Central twisted locality:} 
For every $\alpha \in Z(G)^-$  with $W^{'\alpha}\in\cW_+$, there exists a unitary $Z_\alpha\in U(G)'$
satisfying  $Z_\alpha\Omega=\Omega$,
$J_{\sH(W)} Z_\alpha J_{\sH(W)}=Z_\alpha^{-1}, $ and such that 
\begin{equation}\label{eq:tcl2b}
  \cA(W^{'\alpha}) \subseteq Z_\alpha \cA(W)'Z_\alpha^*.
\end{equation} 
\item[\rm(HKA5)] {\bf Bisognano--Wichmann property:} 
$U(\exp(tx)) = \Delta_{\cA(W),\Omega}^{-it/2\pi}$ for 
$W = (x,\sigma) \in \cW_+$ and $t \in \R$.
\item[(HKA6)] \textbf{Central twisted Haag Duality}:
 For every $\alpha \in Z(G)^-$ with $W^{'\alpha}\in\cW_+$, we have
  \[ \cA(W^{'\alpha})=Z_\alpha \cA(W)'Z_\alpha^*  \]
for a unitary $Z_\alpha$ as in {\rm(HKA4)}.
 \item[(HKA7)] \textbf{$G_{\tau_h}$-covariance:} 
For every 
  $\alpha \in Z(G)^-$ such that $W^{'\alpha}\in\cW_+$, there exists an
  extension $U^\alpha$  of $U$ from $G$ to
  an antiunitary representation of $G_{\tau_h}$,   such that 
\begin{equation}
  \label{eq:twitcovarb}
 \cA(g *_\alpha W) = U^\alpha(g) \cA(W) U^\alpha(g)^*\quad \mbox{ for } \quad 
g \in G_{\tau_h},
 \end{equation}
where $*_\alpha$ is the $\alpha$-twisted action of 
$G_{\tau_h}$ on $\cW_+$ defined in \eqref{eq:twistact}. 
\item[(HKA8)] \textbf{Modular reflection:} 
  (HKA7) holds, and the extensions can be defined by
  \begin{equation}\label{eq:HK8eq}
    U^\alpha(\tau_h)= Z_\alpha J_{\sH(W_0)} 
    \end{equation}
and  a unitary $Z_\alpha$ as in {\rm(HKA4)}.
\end{itemize}
\end{defn}

\begin{itemize}
\item[(SSA)] \textbf{Spin--Statistics:} (HKA4) holds and, in addition,
  $Z_\alpha^2 = U(\alpha)$.
  
One can also specify the regularity property for nets of von Neumann algebras: 
\item[\rm(Reg)] \textbf{Regularity:} For some $e$-neighborhood $N \subeq G$, the vector
  $\Omega$ is still cyclic (and obviously separating) for the
  von Neumann algebra
 \[ \cA(W_0)_N := \bigcap_{g \in N} U(g)\cA(W_0) U(g)^{-1}.\]
 \end{itemize}

 We shall denote a \textit{net of von Neumann algebras} 
     by a quadruple $(G,\cW_+, U,\cA)$, 
  consisting of a connected Lie group $G$,
     whose Lie algebra $\fg$ contains an Euler element $h$,
the orbit  $\cW_+=G.W_0  \subeq \cG(G_{\tau_h})$
of the abstract Euler wedge $W_0=(h,\tau_h)$, 
 a unitary representation $U$ of $G$ and a net of von Neumann algebras $\cA$  on a fixed Hilbert space $\cH$ on the set $\cW_+$ of wedges.

 We now apply the analog of the Bisognano--Wichmann Theorem and the the
 Spin Statistics Theorem to nets of von Neumann algebras.

The following result will be a key tool:
  \begin{prop} \mlabel{prop:lo-3-24} {\rm(\cite[Prop.~3.24]{Lo08})}
    Let $\cM$ be a von Neumann algebra with cyclic and separating
    vector $\Omega$ and $\sH_{\cM,\Omega} := \oline{\cM_{\rm sa}\Omega}$.
    Then the following assertions hold:
    \begin{itemize}
    \item[\rm(a)] If $\cN_1, \cN_2 \subeq \cM$ are von Neumann subalgebras
      with $\sH_{\cN_1,\Omega} \subeq  \sH_{\cN_2,\Omega}$, then
      $\cN_1 \subeq \cN_2$.
    \item[\rm(b)] If $\cN \subeq \cM'$ is a von Neumann subalgebra
      with $\sH_{\cN,\Omega} = \sH_{\cM,\Omega}'$, then
      $\cN = \cM'$. 
    \end{itemize}
  \end{prop} 

\begin{theorem}{\rm(Geometric Bisognano--Wichmann Theorem)}\label{thm:vnaBW}
  Let $\cA \colon \cW_+\to \vNA(\cH)$ be a net of von Neumann algebras
  satisfying {\rm(HKA1)-(HKA4)} and {\rm(Vac)}.
  Assume that the cones $C_\pm=\pm C\cap \fg_{\pm1}$ have
  inner points in $\fg_{\pm1}$,   that $\g_0 = [\g_1, \g_{-1}]$,
  and   that $\ker U$ is discrete.
  Suppose further that $W_0^{'\alpha}\in\cW_+$ for some
    $\alpha \in Z(G)^-$, i.e., that $h$ is symmetric. Then
  the Bisognano--Wichmann property {\rm (HKA5)} and
  twisted Haag duality {\rm(HKA6)}
  hold.
\end{theorem}

\begin{proof}   Consider the net of standard subspaces
 \[ \sH_\cA:\cW_+ \to \Stand(\cH), \quad
   \sH_\cA(W) := \sH_{\cA(W),\Omega} = \overline{\cA(W)_{sa}\Omega},\]
{ so that we have for any $W \in \cW_+$ that
 $\Delta_{\sH_\cA(W)}=\Delta_{\cA(W),\Omega}$ and $J_{\sH_\cA(W)}=J_{\cA(W),\Omega}$. 

The net of standard subspaces $(G,\cW_+,U, \sH_\cA)$ satisfies all the assumptions of~Theorem~\ref{thm:1BW}, so that (HKA5) holds, and
\[ \sH_\cA(W_0^{'\alpha})  =Z_\alpha \sH_\cA(W_0)'.\]
In view of $Z_\alpha \Omega=\Omega$, we also have
\[ \sH_{Z_\alpha \cA(W)'Z_\alpha^*,\Omega}
  = Z_\alpha \sH_{\cA(W)',\Omega}
  = Z_\alpha \sH_{\cA(W),\Omega}'
  = Z_\alpha \sH_\cA(W)' = \sH_\cA(W^{'\alpha})
  = \sH_{\cA(W^{'\alpha})}.\]
Therefore 
  \[ \cA(W^{'\alpha})=Z_\alpha \cA(W)'Z_\alpha^*  \]
  follows from Proposition~\ref{prop:lo-3-24}. This is (HKA6). 
}\end{proof}

One can state a regularity theorem for von Neumann algebras in analogy with Theorem~\ref{thm:reg}: 

\begin{theorem}\label{thm:vnareg}{ {\rm(Regularity Theorem for
    von Neumann algebras)} }
  Let $\cA \colon \cW_+\to \vNA(\cH)$ be a net of von Neumann algebras satisfying {\rm(Vac)}.
 Suppose that $C_\pm = \pm C \cap \g_{\pm 1}$ generate  $\fg_{\pm1}$ and that $\ker U$ is discrete. 
  Then the following implications hold:
\begin{enumerate}
\item[{\rm(a)}] {\rm(HKA1)}, {\rm(HKA2)} imply {\rm(Reg)}. 
\item[{\rm (b)}] {\rm(HKA3)}, {\rm(HKA5)} and
  $\g_0 =\RR h+ [\g_1, \g_{-1}]$ imply {\rm(Reg)}.   
\end{enumerate}  
\end{theorem}

\begin{proof}
  \nin(a) As in the proof of the Regularity Theorem \ref{thm:reg}, $S_{W_0}$
  has an interior point $s_0$ and there exists an $e$-neighborhood
  $N\subseteq G$ such that $N^{-1}s_0\subseteq S_{W_0}$.
  Then, by ${\rm(HKA1)}$ and ${\rm(HKA2)}$,
  we have that $\Ad(U(g^{-1}s_0)) (\cA(W_0))\subseteq \cA(W_0)$, hence
  \[ \Ad(U(s_0)) (\cA(W_0))\subseteq \bigcap_{s\in N}\Ad(U(s))(\cA(W_0)).\]
  Thus (Reg) follows from (Vac).

\nin(b) If (HKA3) and (HKA5) are satisfied, then
Borchers' Theorem for von Neumann algebras (\cite[Thm.~6.3.1]{Lo08b})
implies that
\[ \exp(C_+)\exp(\RR h)\exp(C_-)\subseteq S_{\cA(W_0)}
:= \{ g \in G \colon \Ad(U(g)) \cA(W_0) \subeq \cA(W_0) \}.\] 
Then, as in the proof of Theorem  \ref{thm:reg}(b),
$S_{\cA(W_0)}$ has interior points, and (Reg) follows
as in the second part of the  proof of (a). 
\end{proof}

Theorem \ref{thm:vnareg} can not be deduced directly from Theorem \ref{thm:reg}, applied to the net of standard subspaces
$\sH_{\cA}(W)=\overline{\cA(W)_{sa}\Omega}$. Indeed, if $\cB$ is a general set of wedges, the inclusion
\[ \overline{\left(\bigcap_{W\in \cB}\cA(W)_{sa}\right)\Omega}\subseteq\bigcap_{W\in \cB}(\overline{\cA(W)_{sa}\Omega})=\bigcap_{W\in \cB}\sH_{\cA}(W)\]
can be strict (cf.~the discussion and the examples in
\cite[Sects.~3.1 and~6.1]{LM24}).

\begin{thm}{\rm(Spin--Statistics Theorem for von Neumann algebra nets)} \label{thm:spst2}   Let
  \[ \cA \colon \cW_+\to \vNA(\cH) \]  be a net of von Neumann algebras
    for a unitary representation $(U,\cH)$ of $G$, satisfying
    {\rm(HKA1)-(HKA3), (Vac), (Reg)} and {\rm(HKA5)}.
Assume that $U$ has discrete kernel and that $h$ is symmetric. 
Then $\cA$ also satisfies {\rm(HKA4)}, {\rm(HKA6)-(HKA8)}
and {\rm(SSA)}, and $\sH_\cA$ is the BGL net of the
  antiunitary extension of $U$ specified by
  $\tilde U(\tau_h) = J_{\cA(W_0),\Omega}$. 
  \end{thm}

  \begin{proof}  Consider the net of standard subspace on $\cW_+$, defined by
    $\sH(W) := \overline{\cA(W)_{sa}\Omega}$.
    Then $\sH$ is a regular net satisfying (BW).
 Hence Theorem~\ref{thm:spst1} also implies (CTL), (HK6)-(HK8)
  and,     whenever $W^{'\alpha} \in \cW_+$, the relation
  $Z_\alpha^2=U(\alpha)$ follows from (SS).
    As $Z_\alpha$ commutes with $U(G)$ and $\cH^G = \C \Omega$ is
    one-dimensional by~(Vac),
    we have $Z_\alpha \Omega = c \Omega$ for some $c \in \T$.
    Since $U(\alpha)$ fixes $\Omega$, we obtain $c^2 = 1$, hence
    $c \in \{ \pm 1\}$. If $c = 1$, then $Z_\alpha$ fixes~$\Omega$,
    and otherwise we replace $Z_\alpha$ by $-Z_\alpha$, which fixes $\Omega$.
    We therefore obtain (HKA4) and thus  (SSA) {follows from (SS). }
    Further, (HKA6) follows from     $\sH(W)'=Z_\alpha \sH(W^{'\alpha})$ and
    Proposition~\ref{prop:lo-3-24}. 

    { To see that (HKA7/8) follow form (HKA2), (HKA6)
      and Theorem \ref{thm:ee},
      recall from Theorem~\ref{thm:spst1} that
      $U$ extends to an antiunitary  representation $U$
    of $G_{\tau_h}$ by
    \[ U(\tau_h)=J_{\cA(W_0),\Omega}=J_{\sH(W_0)}, \quad \mbox{ and that} \quad
      Z_\alpha U(\tau_h) \sH(W_0) = Z_\alpha\sH(W_0)' = \sH(W_0^{'\alpha}) \]
follows from (CTHD).
From  Proposition~\ref{prop:lo-3-24}, $Z_{\alpha}\Omega=\Omega$
and (HKA4), we thus obtain 
    \[ \Ad( Z_\alpha J_{\sH(W_0)})(\cA(W_0))=\cA(W_0^{'\alpha}).\]
For the extension $U^\alpha$ of $U\res_G$ by 
\[ U^\alpha(\tau_h)=Z_\alpha J_{\cA(W_0),\Omega} = Z_\alpha U(\tau_h),\]
we now obtain from (HKA2) and $Z_\alpha \in U(G)'$
that (HKA7) and  (HKA8) also hold.}
  \end{proof}

\subsection{Nets of von Neumann algebras on causal spacetimes}
 
Now we describe some applications of the geometric Spin--Statistics and the Bisognano--Wichmann theorems to nets of von Neumann algebras on $1$+$2$-dimensional Minkowski space, on the conformal circle, and on the $1$+$1$-dimensional
Einstein universe. This relates our results to the existing literature.

\subsubsection{On 1+3-dimensional  Minkowski space} 

In \cite[Thm.~2.11]{GL95} the authors considered an isotone net
$(\cA(\cO))_{\cO \subeq \R^{1,3}}$ of
von Neumann algebras on open subsets of Minkowski space $\RR^{1,3}$,
in a fixed Hilbert space $\cH$.
They assume that  there exists a cyclic vector $\Omega\in \cH$  for the algebras associated to {\it spacelike cones,} i.e., spacelike, convex, causally complete, pointed cones.
They require  \textit{graded commutation relations} of local algebras for some unitary involution $\Gamma$. 
This means that there exists a grading operator $\Gamma$ (a unitary involution)
  defining an involution $\Ad(\Gamma)$ on any  local algebra $\cA(\cO)$.
  This splits the local algebras $\cA(\cO)$
  into two subspaces \
 \begin{equation}\label{eq:ncc}
   \cA(\cO)_\pm=\{a\in\cA(\cO): \Gamma a\Gamma=\pm a\}. 
\end{equation}
 They assume that the commutation relations depend on the parity, namely $\cA(\cO)_+$ and $\cA(\cO')$ commute and $\cA(\cO)_-$ and $\cA(\cO')_-$ anti-commute.
  This can be expressed in terms of the Lie super bracket as follows:  
Endow $B(\cH)$ with the grading defined by $\Ad(\Gamma)$ and
  consider the  {\it Lie super bracket}
\begin{equation}\label{eq:ls1}
 [a,b]_s := ab - (-1)^{|a|\cdot |b|}ba, 
\end{equation}
where $|a| \in \{0,1\}$ is defined by $\Gamma a \Gamma = (-1)^{|a|} a.$ 
Then we  have
\begin{equation}\label{eq:ls2}
 [\cA(\cO), \cA(\cO')]_s = \{0\}.
\end{equation}
Guido and Longo 
further assume that the net satisfies {\it modular covariance}, namely
\begin{equation}\label{eq:modcov}
\Ad(\Delta_{\cA(W),\Omega}^{it}) \left(\cA(\cO)\right)=\cA(\Lambda_W(-2\pi t)\cO), 
\end{equation}
where $\cO$ is any open subset of  Minkowski space,
$W$ is a wedge region (cf.~Example~\ref{ex:mink}) and
  $\Lambda_W(t)$ the corresponding boost group.
  This property is weaker than the Bisognano--Wichmann property
{since it does not require the existence of a
    unitary positive energy representation of the Poincar\'e group
    for which the net is covariant. }

Let $\tilde G=\tilde \cP_+^\up$ be the double (=universal) covering of the Poincar\'e group $G=\cP_+^\up$, $h\in\cE(\fg)$, let $\tau_h$ be a corresponding 
involution on $\tilde G$ and let $U$ be a unitary representation of $\tilde G$.  
Then any $\tilde G$-orbit of abstract Euler couples in
$\cG_E(\tilde G_{\tau_h})$ is isomorphic to the $\tilde G$-orbit of wedge regions in Minkowski spacetime,
where  the $\tilde G$-action factors through the action of $G$ (cf.~\cite[Ex.~4.21]{MN21}).  The causal complement  $W'$ of a wedge region $W$
in Minkowski space is  associated with the abstract wedge $W^{'\tilde \rho(2\pi)}$  because the $\tilde G$-orbit of $W$ contains the central complement $W^{'\tilde \rho(2\pi)}$ but not~$W'$ (Remark~\ref{rem:z2pi1oh}(a)).
Here  $\tilde \rho$ is the lift to $\tilde G$ of a one-parameter group of rotations in a space-like plane. Note that
$\tilde \rho(2\pi)$  does not depend on the chosen plane.
With this identification, we can relate
nets on Minkowski spacetime on wedge regions to nets on abstract Euler wedges.

In \cite{GL95} the authors prove that the modular groups of the wedge algebras $\cA(W)$ {define} a representation of $\tilde G$ with positive energy,
fixing $\Omega$, and they verified
(HKA2), (HKA3), (Vac), (SSA), (HKA5)-(HKA8). In particular, by (SSA), we have
$Z_{\tilde \rho(2\pi)}^2=\Gamma=U(\tilde \rho(2\pi))$. 
This identity determines the unitary operator implementing the grading, as well as the commutation relations of the wedge algebras.

From our perspective, consider a unitary representation $U$ for
  which the net \break 
  $\cA \colon \cW_+ \to \vNA(\cH)$  of von Neumann algebras on $\cH$
  is covariant
  and satisfies  the hypotheses of Theorem \ref{thm:spst2}.
Since $Z(\tilde G)=\{\tilde \rho(2\pi), e\}$,  Theorem \ref{thm:spst2} implies 
    \begin{equation}
      \label{eq:ssrel}
      Z_{\tilde \rho(2\pi)}^2=U(\tilde \rho(2\pi)), 
      \quad \mbox{ where } \quad U(\tilde \rho(2\pi))^2=\textbf{1}.
  \end{equation}
Setting $\Gamma:= U(\tilde \rho(2\pi))$, the wedge algebras decompose as in \eqref{eq:ncc}
  and satisfy the graded commutation relations \eqref{eq:ls2}. 
  In our picture we do not assume any commutation relation  on wedge algebras, 
  but we require Poincar\'e covariance and the Bisognano--Wichmann property. In \cite{GL95} the authors assume the  modular covariance condition
\eqref{eq:modcov}, but require graded commutation relations.

  We further point out that, if $U$ factors through $G$, the quotient map $q:\tilde G\rightarrow G$ induces a  bijection $q^\cW \colon \cW_+^{\tilde G} \to \cW_+^{G}$  of homogeneous $\tilde G$-spaces. By Theorem \ref{thm:spst2}
  we conclude that $\Gamma $ can be chosen to be $\textbf{1}$,
    so that we may choose $Z = \1$, 
  and the  commutation relations are verified with the usual commutator
  bracket.
\medskip

\subsubsection{On the circle}

Let $G= \PSL_2(\RR)$ be the M\"obius group and $h\in\cE(\fg)$,
and consider the universal and the double covering,
$\tilde G$ and $G^{(2)}$ of $G$. The Euler involution $\tau_h^\g$
integrates to involutions $\tilde\tau_h$ on $\tilde G$ and
$\tau_h^{(2)}$ on $G^{(2)}$  and $\tau_h$ on $G$.
Then the covering maps
 \[ q^{(2)}:\tilde G_{\tilde \tau_h}\rightarrow G^{(2)}_{\tau_h^{(2)}}, 
  \quad 
  q:\tilde G_{\tau_h} \rightarrow G_{\tau_h}   \quad \mbox{ and } \quad \uline q:G^{(2)}_{\tau_h^{(2)}}\rightarrow G_{\tau_h} \]
satisfy $q^{(2)}(\tilde \tau_h)=\tau_h^{(2)}$ and  $q(\tilde \tau_h)=\tau_h$.
The group $\tilde G_{\tilde \tau_h}$ acts on $\bS^1$ by $g.x=q(g)x$ with $g\in \tilde G_{\tilde \tau_h}$, and analogously $\tilde G^{(2)}_{\tau_h^{(2)}}$
acts on $\bS^1$ through $\uline q$. The quotient maps induce surjections  of homogeneous $\tilde G$-spaces $\cW_+^{\tilde{G}} \to \cW_+^{ G^{(2)}}$ and $\cW_+^{\tilde{G}} \to \cW_+^{ G}$. The natural projection $\cW_+^{G^{(2)}} \to \cW_+^{G}$ 
  is bijective, so that both homogeneous spaces are isomorphic to
  the set of non-dense intervals on $\bS^1$, cf.~ Example \ref{ex:mob}, Remark \ref{rmk:quot} and \cite[Ex.~2.10(e)]{MN21}.

  In \cite[Thm.~2.1, Prop.~2.4]{DLR07} the authors start with
a net of von Neumann algebras $\cA(I)\subset B(\cH)$ on
the set $\cI(\bS^1)$
of non-dense open intervals in the circle $\bS^1$. By the previous paragraph one can identify $\cI(\bS^1)$
with the set $\cW_+^{G^{(2)}}$ of abstract wedges.
They call the net $\cA$ {\it a conformal net on $\bS^1$}
if it satisfies isotony, $\tilde G$-covariance, namely
\[ U(g)\cA(I)U(g)^*=\cA(q(g)I),\quad g \in \tilde G, \]
(HKA3), existence of a $\tilde G$-fixed vacuum vector $\Omega\in\cH$,
cyclicity of the vacuum vector for the algebra
$\bigvee_{I\in\cI}\cA(I):=\left(\bigcup_{I\in\cI}\cA(I)\right)''$, 
and the separating property of the vacuum vector for $\bigcap_{I\in\cI}\cA(I)$. In this setting, they
prove that $\Omega$ is cyclic and separating for all local algebras~$\cA(I)$
(Reeh--Schlieder property), the modular covariance property,
the existence of an 
antiunitary covariant extension of $U$ to  $\tilde G_{\tilde \tau_h}$, and
\[ U(\tilde \rho(2\pi))^2=\textbf{1}, \] where $\tilde \rho(\theta)$ is the lift to $\tilde G $ of the one parameter group of rigid
rotations in $G$. Here the modular covariance property takes the form 
\[ \Ad(\Delta_{\cA(I_1),\Omega}^{it}) \left(\cA(I_2)\right)
=\cA(\delta_{I_1}(-2\pi t)I_2), \] 
where $I_1, \,I_2$ are non-dense intervals in $\bS^1$
and $\delta_{I_1}(t)=\exp(th_{I_1})$,  where $h_{I_1}$ is the Euler element associated to the interval $I$, cf.~\cite[Ex.~2.10(c)]{MN21}. 
Assuming central twisted locality, the properties
(HKA5), (HKA6), (SSA), and the graded commutation relations \eqref{eq:ls2} hold, see \cite[Prop.~2.3]{DLR07}. 

From our perspective, we consider
  the covering group $G^{(2)}$ of $\PSL_2(\RR)$ 
  from above, and the action on its  abstract wedge spaces. 
  Let
  \[ \cA \colon \cW_+^{G^{(2)}} \to \vNA(\cH) \]
  be a net of von Neumann algebras satisfying (HKA1)-(HKA3)  and (Vac) with respect to a unitary positive energy representation $U$ of $\tilde G$.
We assume that (HKA4) holds in the following sense:
For $W_0= (h,\tau_h) \in \cW_+^{G^{(2)}}$
such that $W_0^{'\alpha} = (-h, \alpha \tau_h) \in \cW_+^{G^{(2)}}$
for some $\alpha \in Z(G^{(2)})^-$, there exists a
unitary $Z_\alpha\in U(\tilde G)'$ satisfying 
\begin{equation}
J_{\cA(W_0)} Z_\alpha J_{\cA(W_0)}=Z_\alpha^{-1}, 
\end{equation}
such that 
\begin{equation}
\cA(W_0^{'\alpha}) \subseteq Z_\alpha \cA(W_0)'Z_{\alpha}^*.
\end{equation}

We can lift 
  $\cA$ to a net $\tilde \cA  := \cA \circ q^{(2),\cW}$
on $\cW_+^{\tilde G}$, where $q^{(2),\cW} \colon \cW_+^{\tilde G} \to \cW_+^{G^{(2)}}$
is the natural covering map (Proposition~\ref{prop:BGL-natural}).
Since $\fsl_2(\RR)$ is
a hermitian simple Lie algebra, (Reg) holds for~$\tilde \cA$
by Theorem \ref{thm:vnareg}(a). If the net~$\tilde  \cA$ satisfies (HKA5), then Theorem \ref{thm:spst2} holds for
$\tilde \cA$. Since $\tilde \cA$ factors through $\cW_+^{(2)}$,
the net of standard subspace
$\tilde \sH(W)=\overline{\tilde \cA(W)_{sa}\Omega}$ on $\cW_+^{\tilde G}$ factors through a net $\sH$ on $\cW_+^{G^{(2)}}$. By (HKA5) and (HKA8), $\tilde\sH$ is the BGL net $\sH_U$ of $U$ that factors through $\cW_+^{G^{(2)}}$.
For $G_1 = \tilde G$ and $G_2 = G^{(2)}$, we then obtain 
   with Lemma \ref{lem:fact}, $\partial_h(G_1^{h_1}) \cap D \subeq \ker(U)$
 for $D=2\ZZ$. This implies that $U$ factors through
$G^{(2)}_{\tau_h^{(2)}}$. In particular $U(\tilde \rho(2\pi))^2=\textbf{1}$.

\subsubsection{On the Einstein universe}\label{sect:MTnew}

  In this section we use the notation introduced in Section~\ref{sect:2dimss}. 
We reformulate the following example,
    treated in \cite[App.~A]{MT19}, within our setting. We recall that $H_{\tau_h}=\cP_+$ and $\tilde G_{\tau_h}=(\tilde{\Mob}\times\tilde{\Mob})_{\tau_h}$,
  and $M_0$ is a copy of the $1$+$1$-dimensional Minkowski space  in the Einstein universe $\tilde M^{1,1}$.

Let $(\cA(W))_{W\in\cW^{\tilde G}}$ be a $\tilde{G}$-covariant
net of von Neumann algebras satisfying (Vac) and (HKA1)–(HKA4) 
and write $\tilde U$ for the corresponding $\tilde G$-representation,
assumed to satisfy the spectral condition.
Let $W_R$ and $W_L$ be the Rindler wedge and its complement
in $M_0\subset \tilde M^{1,1}$, respectively, and
let $(h,\tau_h)$ and $(-h,\alpha \tau_h)$ be the associated Euler couples,
where $\alpha=r_s(-2\pi)$ by \eqref{eq:EWrswr}. 
In this setting Theorems~\ref{thm:vnaBW}, \ref{thm:vnareg}, and \ref{thm:spst2} hold for $(\cA(W))_{W\in\cW_+^{\tilde G}}$.
We add the assumption that $Z_\alpha=\1$, and this implies 
\begin{equation}\label{eq:ALR}\cA(W_L)=\cA(W_R)'.
\end{equation}
This is the\textit{ locality condition } on the restriction of the net $\cA$ to (a copy of the) Minkowski spacetime $M_0\subset\tilde M^{1,1}$. 
By the Spin--Statistics Theorem~\ref{thm:spst2}, $\tilde U(\alpha)=Z_\alpha^2 =\textbf{1}$.

Now, let $W_1 \in\cW_+^{\tilde G}$ and $W_2=r_s(2k\pi)W_1$ with $k\in \ZZ$.
Then
\[ \cA(W_2)=\Ad(\tilde U(r_s(2k\pi)))\cA(W_1)=\cA(W_1).\] 
The net $(\cA(W))_{W\in\cW_+^{\tilde G}}$ factors through a net defined
on the wedge space $\cW_+^{\oline G}$, where
\[ \oline G=\tilde G/\langle r_s(2\pi)\rangle,    \]
and  $\oline G_{\tau_h}$ is its extension (see Remark \ref{rmk:ext}).
The associated homogeneous space is a quotient of $\tilde M^{1,1}$,
called the
{\it Einstein cylinder} and is described as follows. Consider the Minkowski space $M_0$ embedded in the conformal completion $\tilde M^{1,1}$:
\[ M_0=(-\pi,\pi)\times(-\pi,\pi)\subset\tilde M^{1,1} \]
in the lightray coordinates $(u_+,u_-)$ (cf.~Section~\ref{sect:2dimss}).
Define the extended space and time coordinates  by $u_+\pm u_-$.
In these coordinates, the Einstein cylinder is 
$\bS^1\times\RR=\tilde M^{1,1}/\sim$, where $x\sim y$ if and only if $x=r_s(2k\pi)y$
for some $k \in\Z$. Thus, the spatial direction of Minkowski space is compactified,
while the time direction remains non-compact, see Figure \ref{fig:3}. 

{We conclude by remarking that our results, Theorems \ref{thm:vnaBW} and \ref{thm:spst2}, extend this example, when the more general centrally twisted locality condition (HKA4), generalizing \eqref{eq:ALR}, is taken into account. 
  In particular, we obtain the following proposition, the analog
  of Proposition~\ref{prop:iden} for von Neumann algebra nets:

\begin{proposition}\label{prop:iden2}
  Let $(\cA(W))_{W\in\cW_+^H}$ be a Poincar\'e covariant
  net of standard subspaces in 1+1 dimensions, satisfying
  {\rm(HKA1)--(HKA3), (Vac)}, which extends to a $\tilde G$-covariant
  net of standard subspaces satisfying {\rm(HKA1)--(HKA4), (Vac)}.
{Assume that the unitary representation $U$ of $\tilde G$
  has discrete kernel.}
  Then the extension $(\tilde\cA(W))_{W\in\cW_+^{\tilde G}}$ also satisfies
  {\rm(HKA5)--(HKA8)} and {\rm(SSA)}, and $\sH_\cA$ is the BGL net of the
  antiunitary extension of $U$ specified by
  $\tilde U(\tau_h)=J_{\cA(W_0),\Omega}$.
\end{proposition}

\begin{proof}
  Let $\tilde U$ be an antiunitary positive energy representation of
{$\tilde G_{\tau_h}$}, with discrete kernel. Let
$\cW_+^{\tilde G} = \tilde G.(h,\tau_h)$ be the 
orbit of Euler wedges and let $\tilde\cA$ be a net of von Neumann algebras on
$\cW_+^{\tilde G}$, satisfying (HKA1)-(HKA4). Then 
(Reg) follows from Theorem~\ref{thm:vnareg}(a),
and  {(HKA5) (resp.~(BW)) 
and (HKA6) (resp.~(CTHD))} both follow from Theorem~\ref{thm:vnaBW}.
Thus \red{(HKA7)}-(HKA8) and  (SSA) are a consequence of 
    Theorem \ref{thm:spst2}. In particular $\sH_\cA$ is the BGL net of the
  antiunitary extension of $U$ specified by
  $\tilde U(\tau_h)=J_{\cA(W_0),\Omega}$.
\end{proof}

}

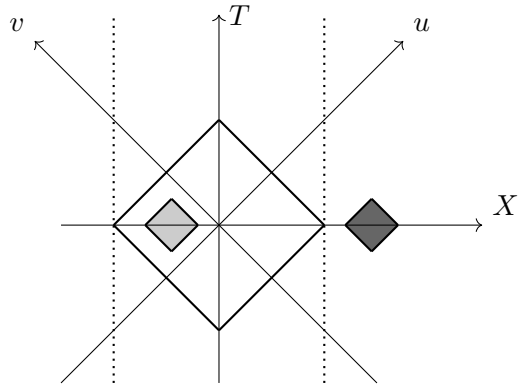
\begin{figure}[ht]\centering
\begin{tikzpicture}[scale=0.7]
        \draw [->] (-3,0) --(5,0) node [above right] {$X$};
         \draw [->] (0,-3)--(0,4) node [ right] {$T$};
          \draw [thick] (-2,0)-- (0,2);
          \draw [thick] (0,2)-- (2,0);
          \draw [thick] (-2,0)-- (0,-2);
          \draw [thick] (2,0)-- (0,-2);
          
          
\draw [ thick] (-0.4,0) -- (-0.9,0.5);
\draw [thick] (-0.4,0) -- (-0.9,-0.5);
\draw [thick] (-1.4,0) -- (-0.9,0.5);
\draw [thick] (-1.4,0) -- (-0.9,-0.5);

\draw [thick] (2.4,0) -- (2.9,0.5);
\draw [thick] (2.4,0) -- (2.9,-0.5);
\draw [thick] (3.4,0) -- (2.9,0.5);
\draw [thick] (3.4,0) -- (2.9,-0.5);

          \draw [thick,dotted] (-2,-3)--(-2,4);
          \draw [thick,dotted] (2,-3)--(2,4);
          
          \draw [thin,->] (-3,-3) -- (3.5,3.5) node [above right] {$u$};
          \draw [thin,->] (3,-3) -- (-3.5,3.5) node [above left] {$v$};
  \fill [color=black,opacity=0.2]
               (-0.4,0) -- (-0.9,0.5) -- (-1.4,0) --(-0.9,-0.5) ;
           \fill [color=black,opacity=0.6]
              (2.4,0) -- (2.9,0.5) -- (3.4,0) -- (2.9,-0.5);
   \end{tikzpicture}
   \caption{The cylinder is obtained by identifying the dotted lines in
       $\tilde M^{1,1}$.   The dark and the light grey wedge regions are identified because
 they are transformed into each other by $r_s(2\pi)$; see \eqref{eq:rs}.}
\label{fig:3}
\end{figure}

\section{Outlook}

Our study of wedge regions on causal homogeneous spaces and their relation with abstract Euler wedges has been developed in \cite{MNO23,MNO25}. This framework  allows finer localization  for states and algebras (see \cite{FNO25a, NO21, MN24}) and can be endowed with a notion of causal independence (or disjointness) on homogeneous spaces. In the forthcoming paper \cite{MN25}, a detailed analysis of a notion of causally disjoint regions for such spaces is presented.

A natural next step is to study the geometric properties of Doplicher–Haag–Roberts (DHR) representations for such models and the possible field reconstruction (see \cite{DHR71, DHR74, DR90}). In the general framework of \cite{Ro04}, one considers an isotonous, local net of von Neumann algebras
$(\cA(\cO))_{\cO\in\cK}$, defined over a partially ordered index set $\cK$, equipped with a notion of causal disjointness $\bot$, and a distinguished vacuum vector. This net is said to be in the vacuum representation.
On a Lorentzian manifold, $\cO_1\bot\cO_2$ corresponds to causally independent regions, while on the circle it corresponds to disjoint intervals. In all cases, locality is expressed by the inclusion
$\cA(\cO^\bot)\subset \cA(\cO)'$. This property is in general strengthened to duality, namely $\cA(\cO)'=\cA(\cO^{\bot})$.

A representation of the net is called a DHR representation if,
once an index $\cO_0$ is fixed, the representation is equivalent to the vacuum representation on $\cO_0^\bot$. Different geometries of the index set may lead to different structures of superselection sectors; for instance, braid  or permutation group statistics may arise (see \cite{Do09} for a survey and further references).

Within this formalism, our framework provides a new and flexible playground in which richer geometric structures exist and may lead to interesting phenomena. We emphasize that familiar examples, such as nets on de Sitter space, Minkowski spacetime, or on the chiral circle, appear as special cases of our general analysis. In our setting, these models are naturally generalized, 
since the cones of the
causal homogeneous spaces we consider may not be Lorentzian.

A second point we wish to highlight in this outlook is the possibility of recovering a larger group of covariant symmetries from the local geometric structure of nets of von Neumann algebras or nets of standard subspaces. A long-standing problem in quantum field theory is whether scale invariance (or dilation covariance) implies conformal covariance. \break In {$3$+$1$-dimensional} Minkowski spacetime it is widely expected that any ``non-pathological" relativistic quantum field theory that is unitary and covariant under the dilation–Poincar\'e group, should in fact be conformally covariant (see \cite{Na15}). However, a complete mathematical proof of this statement is still missing. In
$1$+$1$ dimensions there exist strong physical arguments supporting this claim \cite{Po88,Zo86}, and a mathematically rigorous result in the algebraic framework was established in \cite{MT19}.

This problem is closely related to the study of the symmetry properties of scaling limit theories \cite{BV95,BV98}. Such limiting theories are expected to exhibit dilation symmetry. However, the emergence of dilation covariance, although expected, is not automatic (see \cite{BDM10}). Then one expects  conformal symmetry to arise in the scaling limit as well, but a general proof of this phenomenon is still missing.

In our geometric framework, in the conformal case, the role of $\fsl_2(\RR)$ subalgebras becomes particularly transparent from our analysis, and the resulting interplay between AQFT and Lie theory may provide useful tools to approach this problem. 
More specifically, understanding when the modular groups associated with two orthogonal standard subspaces generate a representation of $\widetilde{\PSL}_2(\RR)$ can be regarded as a preliminary step towards the problem of extending symmetry from the dilation–Poincar\'e group to the full conformal group
in $3$+$1$-dimensional Minkowski spacetime. Indeed, the appearance of $\widetilde{\PSL_2}(\RR)$-subgroups is a characteristic feature of conformal symmetries, and identifying them through the modular structure of the net provides a geometric mechanism for detecting and constructing larger symmetry groups starting from more limited covariance assumptions. This is already visible in
the $1$+$1$-dimensional algebraic argument in \cite{MT19}. On the other hand,
the $3$+$1$-dimensional case requires a different analysis.

We can formulate a first  problem as follows: Let $(U,\cH)$ be an antiunitary representation of $G := \tilde\PSL_2(\R)_{\tau_{h_0}}$,
then the corresponding BGL net yields the standard subspaces
\[ \sV_1 := \sH_U(h_0, \tau_{h_0}) \quad \mbox{ and } \quad
  \sV_2 := \sH_U(k_0, \tau_{k_0}) = U(\tilde \rho(\pi/2))\sV_1 \]
which satisfy the orthogonality condition \eqref{eq:orthog-abstr} below.
More specifically, for irreducible representations,
the construction from \cite{MNO25} provides a real subspace $\sE \subeq \cH^{-\infty}$,
for which the net $\sH_\sE^{\tilde\dS^2}$ on open subsets of the simply connected
covering $\tilde\dS^2$ of de Sitter space $\dS^2$ extends the BGL net~$\sH_U$,
which we obtain by restricting to wedge regions.

 Let $\sV_1$ and $\sV_2$ be two standard subspaces of the  complex
  Hilbert space $\cH$. Suppose that the corresponding modular
  operators and conjugations satisfy
  \begin{equation}
    \label{eq:orthog-abstr}
 J_2 \Delta_1 J_2 = \Delta_1 \quad \mbox{ and }  \quad 
 J_1 \Delta_2 J_1 = \Delta_2.
  \end{equation}
 \textit{Does this imply that the two modular groups
  $\Delta_1^{i\R}$ and $\Delta_2^{i\R}$ generate a unitary
  representation of the universal covering group
  $\tilde\PSL_2(\R)$ of $\PSL_2(\R)$?}
{  If this configuration comes from the BGL net of an antiunitary representation
  $U$  of a finite-dimensional Lie group $G_{\tau_h}$ by
  $\sV_j = \sH_U(h_j, \sigma_j)$, then
 \[  -h_2 = \Ad(\sigma_1) h_2 = \tau_{h_1}(h_2) \quad \mbox{ and } \quad 
   -h_1 = \Ad(\sigma_2) h_1 = \tau_{h_2}(h_1),\]
 so that $(h_1, h_2)$ is a symmetrically orthogonal Euler pair,
 hence generates an $\fsl_2(\R)$-subalgebra of $\g$.}  {The germ of this observation already appears in \cite[Thm. 5.1.1]{Lo08}, where the conformal rotation is related to the modular groups of orthogonal wedge subspaces via the Bisognano–Wichmann property of $\sH_U$, where $U$ is
 an antiunitary positive-energy representation of $G_{\tau_h}=\widetilde{\PSL}_2(\RR)_{\tau_h}$.}
 
As the following class of  counterexample shows, in general
the answer to this question is negative, so some extra assumptions have
to be included. 
We start with a positive operator $A = e^H$ on a real Hilbert space $\cK$
and consider the complex Hilbert space $\cH := \cK_\C \oplus \cK_\C$.
Then the pair
\[  \Delta_1 := A \oplus A^{-1} \quad \mbox{ and } \quad
J_1(\xi,\eta) = (\oline\eta, \oline \xi) \]
specifies a standard subspace $\sV_1$.

With the unitary operator
\[ U := \frac{1}{\sqrt 2} \pmat{ 1 & 1 \\ i & -i} \in \U_2(\cH) \]
we obtain for
\[ \Delta_2 := U \Delta_1 U^{-1} \quad \mbox{ and }  \quad 
  J_2 := U J_1 U^{-1} \]
that $J_2(\xi,\eta) = (\oline \xi, \oline\eta)$ and 
\[ \Delta_2
  = \frac{1}{2}\pmat{ A + A^{-1} & -i(A-A^{-1}) \\
    i(A-A^{-1}) & A + A^{-1}}
  = \pmat{ \cosh(H) & -i \sinh(H) \\
    i\sinh(H)  & \cosh(H)}
  = \exp\pmat{ 0 & -i H \\ i H & 0}. \]
The two standard subspaces $\sV_1$ and $\sV_2 = U\sV_1$ satisfy
the above orthogonality condition \eqref{eq:orthog-abstr}. 

More abstractly, we obtain such examples in
representations on $\cH := L^2(\R, \cE)$, where $\cE$
carries an antiunitary representation of $\tilde\SL_2(\R)_{\tau_{h_0}}$.
We then obtain a unitary representation of the group
of measurable maps $\R \to \tilde\SL_2(\R)$, and
for $f \in C^\infty(\R,\R)$, we obtain in
$C^\infty(\R,\fsl_2(\R))$ elements
\[ h := f \otimes h_0 \quad \mbox{ and } \quad k := f \otimes k_0.\]
With
\[ \tau_{h_0}(h) = h, \quad \tau_{h_0}(k) = -k, \quad 
  \tau_{k_0}(h) = -h, \quad \tau_{k_0}(k) = k.\]
This leads to a pair of orthogonal standard subspaces in the associated
BGL net, but the Lie algebra generated by $h$ and $k$ is typically infinite-dimensional.

To conclude that the modular groups $\Delta_1^{i\R}$ and $\Delta_2^{i\R}$ generate
a representation of $\tilde\SL_2(\R)$ thus requires additional regularity.

\begin{prob}
\mlabel{prob:x}   Assume that, in addition to \eqref{eq:orthog-abstr}, 
that, for all $t \in \R$, the real subspaces 
\begin{equation}
  \label{eq:extracond}
 \Delta_2^{it} \sV_1 \cap \sV_2 \quad \mbox{ and }  \quad
 \Delta_1^{it} \sV_2 \cap \sV_2
\end{equation}
are standard. Does this imply that $\sV_{1,2}$ come from the
BGL net of an antiunitary representation of
$\tilde\SL_2(\R)_{\tau_h}$?
\end{prob}

From \cite{FNO25a} we know that \cite[Prop.~4.27]{MN21} applies to all
unitary representations of $\tilde\SL_2(\R)$, so that
the condition \eqref{eq:extracond}
is necessary.

\appendix

\section{Appendices}

\subsection{List of simple three graded Lie algebras}\label{app:list}

The following table contains all Euler elements in simple
  real Lie algebras, cf.~\cite{MN21}.
\medskip

\begin{tabular}{||l|l|l|l|l||}\hline
& $\g$  & $\Sigma(\g,\fa)$  & $h$ & $\g_1(h)$  \\ 
\hline\hline 
1 & $\fsl_n(\R)$ & $A_{n-1}$ & $h_j, 1 \leq j \leq n-1$ & $M_{j,n-j}(\R)$  \\ 
2 & $\fsl_n(\H)$ & $A_{n-1}$ & $h_j, 1 \leq j \leq n-1$ & $M_{j,n-j}(\H)$  \\ 
3 & $\su_{n,n}(\C)$ & $C_{n}$ & $h_n$ & $\Herm_n(\C)$  \\ 
4 & $\sp_{2n}(\R)$ & $C_{n}$ & $h_n$ & $\Sym_n(\R)$   \\ 
5 & $\fu_{n,n}(\H)$ & $C_{n}$ & $h_n$ & $\Aherm_n(\H)$  \\ 
  6  & $\so_{p,q}(\R)$ & {$ B_p, 1 \leq p<q$} & $h_1$ & $\R^{p+q-2}$   \\
  &  & {$D_p, 1 \leq p = q$} &&  \\
7  & $\so^*(4n)$ & $C_n$ & $h_n$ & $\Herm_n(\H)$   \\ 
8  & $\so_{n,n}(\R)$ & $C_n$ & $h_n$ & $\Alt_n(\R)$   \\ 
9 & $\fe_6(\R)$ & $E_6$ & $h_1=h_6' $ & $M_{1,2}(\bO_{\rm split})$   \\ 
10 & $\fe_{6(-26)}$ & $A_2$ & $h_1$ & $M_{1,2}(\bO)$   \\ 
11 & $\fe_7(\R)$ & $E_7$ & $h_7 $ & $\Herm_3(\bO_{\rm split})$   \\ 
12 & $\fe_{7(-25)}$ & $C_3$ & $h_3$ & $\Herm_3(\bO)$   \\ 
13 & $\fsl_n(\C)$ & $A_{n-1}$ & $h_j, 1 \leq j \leq n-1$ & $M_{j,n-j}(\C)$  \\ 
14 & $\sp_{2n}(\C)$ & $C_{n}$ & $h_n$ & $\Sym_n(\C)$   \\ 
15a & $\so_{2n+1}(\C)$ & $ B_{n}$ & $h_1$ & $\C^n$   \\ 
15b & $\so_{2n}(\C)$ & $ D_{n}$ & $h_1$ & $\C^n$   \\ 
16 & $\so_{2n}(\C)$ & $ D_{n}$ & $h_{n-1}, h_n$ & $\Alt_n(\C)$   \\ 
17 & $\fe_6(\C)$ & $E_6$ & $h_1 = h_6' $ & $M_{1,2}(\bO)_\C$   \\ 
18 & $\fe_7(\C)$ & $E_7$ & $h_7 $ & $\Herm_3(\bO)_\C$   \\ 
\hline
\end{tabular} \\[2mm] {\rm Table 1: Simple $3$-graded Lie algebras. We follow the conventions of the tables in [Bo90a] for the classification  cf.~\cite{MN21}}.\\

\smallskip

\begin{tabular}{||l|l|l|l|l||}\hline \label{tab:2}
{} $\g^\circ$ \mbox{(hermitian)}  & $\Sigma(\g^\circ, \fa^\circ)$ & $\g = (\g^\circ)_\C$ & $\Sigma(\g,\fa)$ & {\rm symm.\ Euler element}\  $h$ \\ 
\hline\hline 
$\su_{n,n}(\C)$ & $C_n$ & $\fsl_{2n}(\C)$ & $A_{2n-1}$ & $h_n$ \\ 
 $\so_{2,2n-1}(\R), n > 1$ & $C_2$ & $\so_{2n+1}(\C)$ & $B_n$ & $h_1$ \\ 
$\sp_{2n}(\R)$ & $C_n$ & $\fsp_{2n}(\C)$ & $C_n$ & $h_n$ \\ 
 $\so_{2,2n-2}(\R), n > 2$ & $C_2$ & $\so_{2n}(\C)$ & $D_{n}$ & $h_1$ \\ 
 $\so^*(4n)$ & $C_n$ & $\so_{4n}(\C)$  & $D_{2n}$ & $h_{2n-1},  h_{2n}$ \\ 
$\fe_{7(-25)}$ & $C_3$ & $\fe_7$ & $E_7$ & $h_7$ \\ 
\hline
\end{tabular} \\[2mm] {\rm Table 2: Simple hermitian Lie algebras $\g^\circ$
of tube type, i.e., containing an Euler element.}

\smallskip

Note that $\fsl_2(\R) \cong \so_{2,1}(\R) \cong \su_{1,1}(\C)$ as real hermitian Lie algebras. More isomorphisms 
are discussed in some detail in \cite[\S 17]{HN12}.

\subsection{The classification of orthogonal  Euler pairs} 
\mlabel{app:a.1}

\begin{thm} \mlabel{thm:1.4} {\rm(Classification Theorem for pairs 
of orthogonal Euler elements)} {\rm(\cite[Thm.~2.9]{MNO25})} 
Let $\g$ be a simple real Lie algebra and 
$h \in \cE(\g)$ a symmetric Euler element. 
Let $\theta$ be a Cartan involution with $\theta(h) = -h$ and 
$\fa \subeq \fp := \g^{-\theta}$ be a maximal abelian subspace  
containing~$h$, so that $\fa \subeq \fh_\fp$. In 
\[ \Sigma_1 := \{ \alpha \in \Sigma(\g,\fa) \colon \alpha(h) = 1\} \]
we pick a  maximal set $\{ \gamma_1, \ldots, \gamma_r\}$ of long
roots roots which are orthogonal in the sense that the differences
$\gamma_j - \gamma_k$ are not roots.
Then there exist elements $e_j \in \g_{\gamma_j}$ such that 
the Lie subalgebras 
\[ \fs_j := \Spann_\R \{e_j, \theta(e_j), [e_j, \theta(e_j)]\}, \quad 
j =1,\ldots, r, \] 
are isomorphic to $\fsl_2(\R)$.  
Then 
\begin{equation}
  \label{eq:defc}
  \fc := \Spann \{ k_1, \ldots, k_r \}
\quad \mbox{ for }  \quad k_j := e_j - \theta(e_j) 
\end{equation}
is maximal abelian in $\fq_\fp := \fq \cap \fp$ and contained in the 
subalgebra
\[ \g^* := \fh_\fk \oplus \fq_\fp = \Fix(\theta \tau_h^\g),\] 
so that we have an inclusion 
of restricted root systems $\Sigma(\g^*,\fc) \subeq \Sigma(\g,\fc)$. 
We consider the Euler elements 
\begin{equation}
  \label{eq:hj}
 k^j := k_1 + \cdots + k_j - k_{j+1} - \cdots - k_r \in \fc,  
 \quad j = 0,\ldots, r. 
\end{equation}
Then $\Sigma(\g,\fc)$ is of type $C_r$ and 
$\Sigma(\g^*,\fc)$ is of type $A_{r-1}$, $C_r$ or $D_r$,
as listed in {\rm Table 2} in {\rm\cite{MNO25}}. 
All $G_e^h$-orbits in $\cE(\g) \cap \fq$ are $G^h$-invariant, 
 and a set of representatives is given by: 
\begin{itemize}
\item[\rm(A)] $k^0, \ldots, k^r$, if $\Sigma(\g^*,\fc)$ is of type $A_{r-1}$ and
  $r \geq 1$. 
\item[\rm(C)] $k^r$, if $\Sigma(\g^*,\fc)$ is of type $C_r, r \geq 2$. 
\item[\rm(D)] $k^{r-1}, k^r$, if $\Sigma(\g^*,\fc)$ is of type $D_r, r \geq 2$.
\end{itemize}
Accordingly, the pairs $(h,k^j)$ represent the $\Inn(\g)$-conjugacy classes of
pairs of orthogonal Euler elements.
\end{thm}

\subsection{The conformal group and the Euler elements of the conformal Lie algebra}
\mlabel{subsec:4.2}

    We consider the covering group $G$
    of the connected conformal group $\PSO_{2,d}(\R)_e$ on $\R^{1,d-1}$,
 for which there exists a {\bf faithful}
    unitary representation $(U,\cH)$ of $G$, extending
    the massless scalar representation on the Hilbert space
    $\cH = L^2(\partial V_+)$ (where $V_+$ is the open positive
      light cone);    see \cite[\S 4.4]{NO21} for the existence.
    In $\so_{2,d}(\R)$, we consider the Euler element $h$, contained in
    $\so_{1,d-1}(\R)$, generating the Lorentz boost
    $h_{0,1}$, acting on $\R \be_0 + \R \be_1 \subeq \R^{1,d-1}$,
    corresponding to the
Rindler wedge $W_R$ (cf.~Example~\ref{ex:mink}). 

  We have three essentially different cases
  (cf.\ \cite[Prop.~8.5]{MN25} and \cite[\S 5.6.2]{Ne26}): 
  \begin{itemize}
  \item[(a)]  $d$ {\bf odd}: Then
   $\Ad(G)$ is the centerfree group~$\SO_{2,d}(\R)_e$ with
    \[ \pi_1(\SO_{2,d}(\R)) \cong \pi_1(\SO_2(\R) \times \SO_d(\R))
      \cong \Z \times \Z_2 \]
    and     $\pi_1(G) = 2 \Z \times \Z_2$
    is an index $2$ subgroup, so that
    $G$ is a $2$-fold covering of the centerfree group~$\SO_{2,d}(\R)_e$
for which we have a natural lift $\iota \colon \SO_d(\R) \into~G$. 
    In this case $Z(G) = Z(G)^- \cong \Z_2$,
    \[ Z(\tilde G)_1 = 2 \Z \times \{1\} \subeq \pi_1(G),\]
    hence $Z_1(G) = \{e\}$,  and $Z_2(\tilde G) \not\subeq \pi_1(G)$
    because
\begin{equation}
   \label{eq:dag11}
Z_2(\tilde G) = \{ (n,\oline n) \colon n \in \Z \}
 \subeq\Z \times \Z_2= \pi_1(\SO_{2,d}(\R)).
  \end{equation} 
We conclude that 
    \[ \{e\} = Z_1(G) \not= Z_2(G) = Z_3(G) = Z(G),\]
    which shows with \eqref{eq:lem:5.2} that $W_0' \in \cW_+$. 
    The non-trivial element $\alpha \in Z(G)$ specifies a twisted
      complement $W_0^{'\alpha} = (h, \alpha \tau_h)$.
    
  \item[(b)] $d \equiv 2 \mod 4$: Then $G = \PSO_{2,d}(\R)_e$
    is centerfree, so that
    \[ Z_1(G) = Z_2(G) = Z_3(G)=~\{e\}.\]
    In particular, $W_0' \in \cW_+$,
    and there are no other twisted complements of $W_0$ in~$\cW_+$.
    
  \item[(c)] $d \in 4 \Z$: Then $G = \SO_{2,d}(\R)_e$ with
    \[ Z(G) = \{ \pm \1\} = Z(G)^-.\]
    As $\tau_h$ acts trivially on this group,
    $Z_1(G)= \{e\}$. Since $\Ad(G)^h$ is connected (\cite[Thm.~7.8]{MNO23}),
    we also have
    $Z_2(G) = Z_1(G) = \{e\}$. Further, the $\pi$-rotation 
   $r_{-1,0}(\pi)$ in the $(\be_{-1},\be_0)$-plane  
    is an involution fixed by $\tau_h$, so that
    $\partial_h(r_{-1,0}(\pi)) = e$. We conclude that
    \[ Z_1(G) = Z_2(G)= Z_3(G)= \{e\}.\]
    In particular,  $W_0' \in \cW_+$ also holds in this case,
    and there are no other twisted complements of $W_0$ in $\cW_+$.
  \end{itemize}

Now we discuss the Euler elements of the conformal Lie algebra. For $d \geq 3$, we consider the natural inclusions
\[ \so_{2,1}(\R) \into \so_{2,2}(\R) \into \so_{2,d}(\R)
=\left\{\left(\begin{matrix}
A&B\\B^\top&D
\end{matrix}\right): A \in \so_2(\R),D\in \so_d(\RR), B\in M_{2,d}(\RR)\right\}. \] 
We write $\be_{-1}, \ldots, \be_{d}$ for  the standard basis of $\R^{2,d}$.
The Lorentz Lie algebra $\so_{1,d-1}(\R)$ is embedded in $\so_{2,d}(\R)$
as those matrices annihilating
$\be_{-1}$ and $\be_{d}$
(cf. \cite{HN12}).
The boost Euler elements $h_{j,k} $ are defined by
\[ h_{j,i}\be_i=\be_j,\; h_{j,i}\be_j=\be_i,\;h_{j,i}\be_\ell=0, \quad
  \mbox{ for } \quad j=-1,0,\  i \geq 1, \ \ell\neq j, i.\]

In $\fsl_2(\R)$ the orbits of orthogonal pairs of Euler elements
are represented by the pairs $(h_0, \pm k_0)$ (Example~\ref{ex:sl2}). 
We call an Euler element in $\fsl_2(\R)^{\oplus 2}$ {\it regular} 
if it is a regular element, i.e., if its centralizer is abelian.
Then there is only one orbit of regular Euler elements in
$\fsl_2(\R)^{\oplus 2}$,
represented by $h := (h_0, h_0)$, and there are $4$ orbits of
regular orthogonal pairs, represented by
\begin{equation}
  \label{eq:pairsinsl2}
 ((h_0, h_0), (k_0, k_0)),\quad
 ((h_0, h_0), (k_0, -k_0)),\quad
 ((h_0, h_0), (-k_0, k_0)),\quad
 ((h_0, h_0), (-k_0, -k_0)). 
\end{equation}

The following lemma clarifies Lemma C.1 in \cite{MNO25}
which suggested that $\fsl_2(\R)^{\oplus 2}$ contains only $3$ conjugacy classes
of orthogonal pairs of Euler elements.

  \begin{lem} \mlabel{lem:lemc1}
 In $\so_{2,2}(\R)$ there is only one conjugacy class of regular
  Euler elements. It contains in particular the boost elements
  $h_{-1,1}, h_{-1,2}, h_{0,1}, h_{0,2}$. The orbits of orthogonal pairs of
  regular Euler elements are represented by
  \[       (h_{-1,2}, \pm h_{0,2}) \quad \mbox{ and } \quad
(h_{-1,2}, \pm h_{-1,1})\]
  with
  \begin{equation}
    \label{eq:timelikepairso22}
    [h_{-1,2}, h_{0,2}] 
    = (2z_0, 0) \in \so_2(\R)^{\oplus 2},
    \qquad 
     [h_{-1,2}, h_{-1,1}]
    =  (0,-2z_0).   \end{equation}
The corresponding central elements of $\tilde\SO_{2,2}(\R)_e$ are 
\begin{equation}
  \label{eq:zetabrack}
 \zeta_{h_{-1,2},-h_{0,2}} = (1,0) \quad \mbox{ and } \quad
 \zeta_{h_{-1,2},h_{-1,1}}  =  (0,1)
\end{equation}
  in
  \[     \Z^2 \cong \pi_1(\SO_2(\R)) \times \pi_1(\SO_2(\R))
    \cong \pi_1(\SO_{2,2}(\R)) \subeq Z(\tilde\SO_{2,2}(\R)_e).\]
\end{lem}

\begin{prf}
  This follows from the calculations in the proof  of
  \cite[Lemma~C.1]{MNO25}.
\end{prf}

For $\g = \so_{2,d}(\R)$, $d \geq 3$,
the classification results in \cite{MNO25} show
that the four pairs in \eqref{eq:pairsinsl2}
also represent all $3$ conjugacy classes of orthogonal pairs of
Euler elements in~$\so_{2,d}(\R)$.
Hence two of these are conjugate. 
To determine which of them are, we consider in $\so_{2,d}(\R)$ the invariant
cone $C_\g$ generated by the element 
\[ z_\fk := \pmat{ 0 & 1 \\ -1 & 0} \oplus \{0\}
  = (2z_0, 0) \in \so_2(\R) \oplus \so_d(\R).\]
The bracket relations \eqref{eq:timelikepairso22} show that, with
respect to this cone $C_\g$, 
\begin{itemize}
\item the pair $(h_{-1,2},-h_{0,2})$ is timelike positive, 
\item the pair $(h_{-1,2},h_{0,2})$ is timelike  negative, 
\item the pairs $(h_{-1,2},\pm h_{-1,1})$ are spacelike. 
\end{itemize}
Since we know from the classification in \cite{MNO25}
that there are only $3$-orbits, and the properties of being timelike
positive/negative or spacelike are invariant,
it follows that the two spacelike pairs $(h_{-1,2},\pm h_{-1,1})$ are conjugate. 

To explain this fact geometrically, we observe that
the Lorentz Lie algebra $\so_{1,d-1}(\R)$ contains the Euler elements
$h_{0,1}$ and $h_{0,2}$, representing boost generators in
$\so_{1,2}(\R)$, acting on $\Spann \{ \be_0, \be_1, \be_2\} \cong \R^{1,2}$. 
As $\tau_{h_{0,1}}(h_{0,2}) = -h_{0,2}$ and $\tau_{h_{0,2}}(h_{0,1}) = -h_{0,1}$,
the pair $(h_{0,1}, h_{0,2})$ is symmetrically orthogonal. Under the space 
rotation $r_{2,3}(\pi)$,  the pair $(h_{0,1}, h_{0,2})$ is conjugate to 
$(h_{0,1}, -h_{0,2})$ in $\so_{1,d-1}(\R)$.
This implies in particular that the two spacelike pairs 
$(h_{-1,2},\pm h_{-1,1})$ are conjugate in $\so_{2,d}(\R)$, $d \geq 3$.

We now identify the pairs in
$\fs \cong \fsl_2(\R)^{\oplus 2}$, where the invariant cone
$C_\fs = C_\g \cap \fs$ contains $(z_0, z_0)$
(see Example~\ref{ex:sl2} for the notation): 
\begin{itemize}
\item The timelike positive pair $(h_{-1,2}, -h_{0,2})$ corresponds to 
  $((h_0, h_0), (k_0,k_0))$ in $\fs$. 
\item The spacelike pairs $(h_{-1,2}, \pm h_{-1,1})$
  correspond to  $((h_0, h_0), \pm(k_0,-k_0))$ in $\fs$.
\end{itemize}

\subsection{Standard subspaces}
\mlabel{subsec:stand-subs}

We call a closed real subspace $\sH$ of the complex Hilbert space 
$\cH$ {\it cyclic} if $\sH+i\sH$ is dense in~$\cH$, {\it separating}
if $\sH\cap i\sH=\{0\}$, and \textit{standard} 
if it is cyclic and separating. We write $\Stand(\cH)$ for
the set of standard subspaces of $\cH$. The {\it symplectic complement}
of a real subspace $\sH$ is defined by the symplectic form 
$\Im \langle\cdot,\cdot\rangle$ on $\cH$ via 
\begin{equation}
  \label{eq:sH'}
 \sH'=\{\xi\in\cH:  (\forall \eta \in \sH)\, 
 \Im\langle\xi,\eta \rangle=0 \} =(i\sH)^{\perp_{\Re\langle\cdot,\cdot\rangle}},
\end{equation}
 where $\Re\langle\cdot\,,\cdot\rangle$ is the real part of the scalar product of $\cH$. 
Then $\sH$ is separating if and only if $\sH'$ is cyclic, hence $\sH$ is standard if and only if $\sH'$ is standard.
For a standard subspace $\sH$, one can
define the {\it Tomita operator} $S_\sH$ as the closed antilinear involution
\[S_\sH: \sH+i\sH \to \sH + i \sH, \quad
  \xi+i\eta\mapsto \xi-i\eta.\] 
The polar decomposition $S_\sH=J_\sH\Delta_\sH^{\frac12}$ defines an antiunitary involution $J_\sH$, called the \textit{modular conjugation} of $\sH$, and a positive self-adjoint operator~$\Delta_\sH$, called the \textit{modular operator }of  $\sH$. 
Then 
\[  J_\sH\sH=\sH', \quad 
 \Delta^{it}_\sH\sH=\sH \qquad \mbox{ for every  } \quad 
t\in \R\] and
the {$J_\sH$ commutes with the modular group}:
\begin{equation}\label{eq: TR}
J_{\sH}\Delta^{it}_{\sH}J_{\sH}=\Delta^{it}_{\sH}\qquad \mbox{ for every  } \quad t\in \R.
 \end{equation}

\begin{proposition}\label{prop:11}{\rm (\cite[Prop.~3.2]{Lo08})} 
  The map $\sH\mapsto (\Delta_\sH, J_\sH)$ 
  is a bijection between the set of standard subspaces of $\cH$ 
  and the set of pairs $(\Delta, J)$, where $J$ is a conjugation,
  $\Delta > 0$ selfadjoint with $J\Delta J =\Delta^{-1}$, which
  is equivalent to $J$ commuting with $\Delta^{i\R}$. 
\end{proposition}

If $\sH\subset\cH$ is a standard subspace, and $(\Delta, J)$ the associated operators, then the symplectic complement $\sH'$
from \eqref{eq:sH'}
is associated to the couple $(\Delta^{-1},J)$ (\cite[Prop.~2.13]{Lo08}).

The following observation follows easily from Proposition~\ref{prop:11}.

\begin{lemma}\label{lem:sym}{\rm(\cite[Lem.~2.2]{Mo18})}
Let $\sH\subset\cH$ be a standard subspace  and let $U$ 
be a unitary or antiunitary operator on $\cH$. 
Then $U\sH$ is also standard and 
$U\Delta_\sH U^*=\Delta_{U\sH}^{\epsilon(U)}$ and $UJ_\sH U^*=J_{U\sH}$, 
where $\epsilon(U) = 1$ if $U$ is unitary and 
$\epsilon(U) = -1$ if it is antiunitary. 
\end{lemma}

\begin{lemma}\label{lem:inc} {\rm(\cite[Cor.~2.1.8]{Lo08})} 
Let $\sK\subset\sH\subset\cH$ be an inclusion of standard subspaces.
If $\Delta^{it}_\sH\,\sK=\sK$ for all $t \in \R$, then $\sH=\sK$
\end{lemma}

{For a proof of the following theorem we refer to
\cite[Thms.~3.15, 3.17]{Lo08} and  \cite[Thm.~3.2]{BGL02}.
In \cite{Lo08} the assumption for (a) is stated as $\pm P > 0$,
but the proof given there works for $\pm P \geq 0$, as shown in~\cite[Thm. 3.18]{NO17}.}

\begin{theorem} \label{Borch} {\rm(Borchers--Wiesbrock Theorem)}
  Let $\sH\subset\cH$ be a standard subspace and $U(t)=e^{itP}$ 
be a unitary one-parameter group on $\cH$ with   generator $P$.
\begin{itemize}
\item[\rm(a)] 
If $\pm P \geq 0$ and $U(t)\sH\subset \sH$ for all $t\geq 0$, then 
\begin{equation}\label{eq:DeltaJ} \Delta_\sH^{-is/2\pi}U(t)\Delta_\sH^{is/2\pi} 
= U(e^{\pm s}t)\quad \mbox{ and } \quad 
J_{\sH}U(t)J_{\sH}=U(-t) \quad \mbox{ for all } \quad 
t,s \in \R.
\end{equation}
\item[\rm(b)] If $\Delta_\sH^{-is/2\pi}U(t)\Delta_\sH^{is/2\pi} 
= U(e^{\pm s}t)$ for $s,t \in \R$, then  the following are equivalent:
\begin{enumerate} 
\item[\rm(1)] $U(t)\sH \subset \sH$ for $ t\geq 0$;
\item[\rm(2)] $\pm P \geq 0$. 
\end{enumerate}
\end{itemize}
\end{theorem}

\end{document}